\documentclass{article}
\usepackage{amsmath, amssymb, amsthm}
\newtheorem{assumption}{Assumption}

\newtheorem{definition}{Definition}
\newtheorem{thm}{Theorem}
\newtheorem{remark}{Remark}
\newtheorem{lemma}{Lemma}
\newtheorem{example}{Example}
\newtheorem{proposition}{Proposition}
\usepackage{xcolor}
\usepackage{url, hyperref, tikz, subcaption}
\usepackage{cancel}
\usepackage{soul}
\usepackage{float}
\usepackage{comment}
\RequirePackage[authoryear]{natbib}
\usetikzlibrary{graphs, positioning, calc}

\usepackage{xcolor}

\usepackage[normalem]{ulem}
\newcommand{\stkout}[1]{\ifmmode\text{\sout{\ensuremath{#1}}}\else\sout{#1}\fi}

\def\ci{\perp\!\!\!\perp}

\usepackage[letterpaper, left=2.5cm, right=2.5cm, top=2.5cm,
bottom=2.5cm,dvips]{geometry}

\title{Functional Estimation under Proxy-Based Full-Law Identification\thanks{We thank Eric Tchetgen Tchetgen and Yingyao Hu for enriching discussions.}}
\author{Helen Guo, AmirEmad Ghassami, Ilya Shpitser, Elizabeth L. Ogburn}

\begin{document}

\maketitle
\begin{abstract}
We state general conditions under which the full-data law is identified in the presence of latent variables, leveraging key observed variables (“proxies”) associated with unobserved variables. These assumptions extend those used in existing examples from the literature that recover the full-data law under relatively flexible model assumptions. We first describe graphical models compatible with our assumptions, and then consider target parameters that are functionals of the full-data law, developing estimating equations that give rise to consistent \(M\)-estimators using the observed data. Our approach replaces unobserved random variables appearing in full-data estimating equations with observed variables via a proxy-based weighting scheme. This strategy can be extended to construct observed-data influence functions, giving estimators with desirable properties such as multiple robustness and \(\sqrt{n}\)-consistency despite slower than parametric convergence of nuisance functions. 
\end{abstract}

\section{Introduction}

Latent variable models are  used across a broad range of scientific disciplines, including economics, political science, psychometrics, biology, public health, and medicine \citep{Aigner1984Latent, Medzihorsky2022VDem, Hancock2007Models,  Kopf2021Latent, Muthen1992Epidemiology, RabeHesketh2008Latent}. In a latent variable model, variable(s) of interest are unobserved (i.e., latent), while associated observed measurements (often termed “proxies”) are available, constraining the set of compatible latent probability distributions. This description encompasses a broad range of settings, including settings in which latent variable(s) are measured indirectly through error-prone observations. Under suitable assumptions, targets such as the expectation of the latent variable, or even the full-data law (i.e., the joint distribution of observed variables and latent variables of interest), can be identified. 

A major line of this work, considered nonparametric\footnote{Here, ``nonparametric'' model assumptions refer to conditional independence constraints and additional functional restrictions, provided these are sufficiently weak in the sense that under relative support size/dimensionality assumptions between the variables, distributions violating them can be approximated arbitrarily well by distributions that satisfy them \citep{semiparametric2020cui, Canay_testability2013}.} for its flexible model assumptions, recovers the full-data law by placing relatively mild conditions on the variability of the conditional law of observed variables across values of a latent variable of interest \citep{hu2008instrumental, kuroki14measurement, allman2009identifiability, deaner2023controllinglatentconfoundingtriple, zhou2024causalinferencehiddentreatment}. 
This line of work originates with Kruskal's uniqueness theorem for the Candecomp/Parafac decomposition of three-way arrays arising from latent variable models with finite support \citep{kruskal1977three}, and has since been extended to latent variable models with continuous support \citep{hu2008instrumental}. While identification theory in this line of work is well developed, work on estimation of target functionals of the full-data law remains relatively nascent \citep{zhou2024causalinferencehiddentreatment, Guo2026_outcome, egami2026debiasedinferenceaigenerateddata}. 

Related methods address unobserved variables without seeking to recover the full-data law. For example, the nonparametric instrumental variable method leverages similar conditions alongside additive noise assumptions to identify and estimate causal targets \citep{newey03instrumental}. Additionally, proximal causal inference techniques identify and estimate causal targets under similar conditions \citep{miao2018identifying, semiparametric2020cui, ghassami24causal, dukes2021proximal}, imposing different independence restrictions from the former line of work discussed  \citep{deaner2023controllinglatentconfoundingtriple}. These approaches are related to the former line of work in their use of observed variables to overcome challenges posed by unobserved variables, and relatively mild assumptions. However, rather than identifying a particular causal target, the former line of work seeks to recover the full-data law, from which a broad class of target functionals can be obtained.


We extend the assumptions from the former line of work to identify the full-data law in the presence of (possibly multiple) latent variables, and describe graphical models compatible with our assumptions. We then develop mean-zero estimating equations that give rise to consistent \(M\)-estimators. This approach can be extended to construct influence functions for target functionals of the full-data law. Our construction begins with an estimating equation under the full-data model, translating it into an observed-data equation by replacing unobserved variables in these expressions with observed variables via a proxy-based weighting scheme. The procedure demonstrates how to build estimators with desirable properties (such as robustness to model misspecification and $\sqrt{n}$-consistent estimation despite slower than parametric convergence of nuisance functions) using the observed data, bridging the gap between nonparametric latent variable model identification theory and modern estimation methodology. Our approach is inspired by previous work on binary latent treatments by \citet{zhou2024causalinferencehiddentreatment} and overlaps with work by \citet{egami2026debiasedinferenceaigenerateddata} considering models with finite support and a single latent variable.  

Section~\ref{sec:background} sets up the problems we consider and reviews fundamentals concepts necessary to understand this paper. Section~\ref{sec:identification} establishes identification of the full-data law, which subsequently enables the estimation results developed in Sections~\ref{sec:est_eqns} and~\ref{sec:ifs}. We provide many examples, constructing estimators of the expectation of latent variable(s), as well as of average counterfactual values in settings where important variables such as confounders or outcomes are latent. 

\section{Preliminaries}\label{sec:background}
\subsection{Problem Setup}\label{subsec:setup}
Consider collections of random variables \(\vec H\) and \(\vec O\), where \(\vec H\) consists of unobserved variables of interest and \(\vec O\) consists of observed variables. When these collections appear as arguments of distributions or are assigned values, we use the same notation to denote the corresponding random vectors obtained by arranging the variables in each collection in a fixed order.

\begin{definition}[Observed-data law]
The observed-data law is the joint distribution \(p(\vec{O})\). 
\end{definition}

\begin{definition}[Full-data law]
The full-data law is the joint distribution \(p(\vec{H},\vec{O})\).
\end{definition}

\begin{definition}[Identification of the full-data law]
The full-data law is identified
in a model if there exists a function \(f\) such that

\[
p(\vec{H}, \vec{O})
= f\!\big(p(\vec{O})\big).
\]
\end{definition}

This paper first establishes assumptions that identify the full-data law in the presence of latent variables. 
In some applications, however, the target of interest is invariant to latent state labeling,
making label identification unnecessary. In such settings, we discuss how our assumptions can be adapted to identify the full-data law only up to arbitrary permutation of the support values of \(\vec H\). 

\begin{definition}[Identification of the full-data law up to permutation of the support values of
\(\vec H\)]
The full-data law is identified up to arbitrary permutation of the support values of
\(\vec H\) if there exists a function \(f\) such that 
\(
f\!\big(p(\vec O)\big)\) gives the equivalence class of joint laws \(p(\vec H, \vec O)\) under arbitrary permutation of the support values of
\(\vec H\).
\end{definition}

Following our identification results, we construct observed-data estimating equations for target functionals of the full-data law. Additionally, we adapt these results to give observed-data influence functions. The rest of this section reviews relevant background necessary to understand our work.

\subsection{Directed Acyclic Graphs (DAGs)}
Directed acyclic graphs (DAGs) 
are frequently used to graphically represent statistical models. These graphs consist of nodes representing random variables connected by directed edges, with the additional restriction that the graph contains no directed cycles. We denote a DAG with latent variables \(\vec{H}\) and observed variables \(\vec{O}\) by \(\mathcal{G}(\vec{H}, \vec{O})\)
Conditional independences may be read from the graph via \emph{the d-separation criterion} \citep{pearl09causality}. Although the identification assumptions in this work are not stated in terms of 
graphs, many of these assumptions (such as conditional independences) may be difficult to reason about directly without a graphical model. Hence in this paper, we include graphical examples and give graphical criterion consistent with our identification assumptions.

In several examples throughout this paper, the target functional admits a causal interpretation as an average counterfactual value. Such an interpretation requires assumptions involving the \emph{potential outcome} \(Y(a)\), defined as the value that outcome \(Y\) would have attained under an intervention that sets treatment to \(A=a\). These assumptions may be encoded graphically via single-world intervention graphs (SWIGs), a special class of DAGs \citep{richardson13swig}. 
Because counterfactual assumptions are not the focus of this paper, we do not discuss them explicitly and instead assume that the conditions required for a causal interpretation of the target of interest hold whenever such an interpretation is invoked. Absent such conditions, our results remain valid for the corresponding statistical functionals, but not necessarily as results about causal targets.

\subsection{Acyclic Directed Mixed Graphs (ADMGs)}

DAGs are a special case of acyclic directed mixed graphs (ADMGs), which may additionally contain bidirected edges between variables. Conditional independences in ADMGs may be read off using \emph{m-separation}, a generalization of d-separation \citep{richardson02ancestral, handbook19graphical}. 

In our paper, observed variables \(\vec O\) and latent variables of interest \(\vec H\) are represented explicitly as nodes in an ADMG, while other latent variables are represented implicitly in bidirected edges, obtained via the latent projection operator from \citet{verma90equiv}. 
Similar to DAG notation, we denote such an ADMG by \(\mathcal{G}(\vec{H}, \vec{O})\). In contrast to the DAG representation discussed earlier, the ADMG representation captures all observed dependences while explicitly retaining observed variables \(\vec O\) and only a subset of latent variables. This distinction is useful because target functionals of interest may only depend on the joint law of observed variables and a subset of all latent variables (i.e., latent variables of interest). Moreover, our identification strategy imposes assumptions which may only be justified for this subset of all latent variables. The ADMG representation we use therefore allows us to explicitly retain precisely those latent variables for which are relevant for our target functionals and for which our assumptions are appropriate.

\subsection{Estimating Equations and Influence Functions} 


We develop \(M\)-estimators of the form

$$
\hat\psi
=
\arg\min_{\psi}
\mathcal{L}_n(\psi),
\qquad
\mathcal{L}_n(\psi)
=
\left\|
\mathbb{P}_n\{\varphi(O;\psi,\hat\eta)\}
\right\|^2,
$$

by constructing observed-data estimating functions \(\varphi(O;\psi,\eta)\) that satisfy
$$
\mathbb{E}\{\varphi(O;\psi,\eta)\}=0.
$$
Under standard regularity conditions, the resulting estimator \(\hat\psi\) is consistent for \(\psi\). 

We then adapt this construction to give observed-data influence functions, a special class of mean-zero estimating equations that yield \(M\)-estimators with additional desirable statistical properties.

\begin{definition}[Observed-data influence function]
Let \(\psi\) be a functional of the observed-data law \(p(\vec{O})\), and let \(\eta\) denote a collection of nuisance functions determined by the observed-data law.
A function \(\varphi(\vec{O};\psi,\eta)\) is an observed-data influence function for \(\psi\) if, for any regular parametric submodel \(\{p_\epsilon(\vec{O})\}\) with score \(s(\vec{O})\),
\[
\mathbb{E}\big[\varphi(\vec{O};\psi,\eta)\big]=0,
\quad \text{and} \quad
\left.\frac{d}{d\epsilon}\psi(p_\epsilon)\right|_{\epsilon=0}
=
\mathbb{E}\big[\varphi(\vec{O};\psi,\eta)\,s_\epsilon(\vec{O})\big].
\]
\end{definition}

Given random variables \(\vec{O}\), an estimator \(\hat{\psi}\) may be obtained by solving the empirical estimating equation
\[
\mathbb{P}_n\big[\varphi(O; \hat \psi, \hat \eta)\big]=0,
\]
where \(\varphi(O; \psi,  \eta)\) denotes an observed-data influence function for the target parameter \(\psi\). Under suitable regularity conditions, solution \(\hat{\psi}\) is \(\sqrt{n}\)-consistent and asymptotically normal with asymptotic variance \(\mathbb{E}[\varphi^2(O; \psi, \eta)]\). Such estimators may also be \emph{multiply robust}, meaning that consistency of \(\hat \psi \) is preserved even if some nuisance functions are misspecified \citep{KennedyTutorial, tsiatis06missing}.

Under identification of the full-data law, target functionals of the full-data law are also functionals of the observed-data law. However, the identification arguments developed in this work do not generally yield simple explicit expressions for such targets in terms of the observed-data law. Indeed, estimation of the full-data law typically relies on iterative optimization procedures, such as alternating least squares and gradient descent, with continuous-variable extensions employing sieve-based approximations \citep{ALS,gradient_descent, hu2008instrumental, zhou2024causalinferencehiddentreatment}. 

Consequently, standard approaches that derive observed-data influence functions by differentiating a closed-form representation of the target in terms of the observed-data law are not readily applicable. We instead begin with a full-data influence function, which is typically straightforward to derive in closed form, and develop a procedure for translating it into an observed-data influence function. For this purpose, we assume that under the identification map between the observed-data and full-data law, every regular parametric submodel of the observed-data law \(\{p_\epsilon(\vec{O})\}\) with score \(s(\vec{O})\) can be induced through marginalization by a regular parametric submodel of the full-data law \(\{p_\epsilon(\vec{H},\vec{O})\}\) with score \(s(\vec{H},\vec{O})\), and that \(\frac{d}{d \epsilon} \psi(p_\epsilon(\vec O))=\frac{d}{d \epsilon} \psi(p_\epsilon(\vec H,\vec O))\). Our construction systematically replaces unobserved random variables appearing in the full-data influence function with observed variables and appropriately chosen nuisance functions.

\begin{definition}[Full-data influence function]
Let \(\psi\) be a functional of the full-data law \(p(\vec{H},\vec{O})\), and let \(\eta\) denote a collection of nuisance functions determined by this law.  
A function \(\phi(\vec{H},\vec{O};\psi,\eta)\) is a full-data influence function for \(\psi\) if, for any regular parametric submodel \(\{p_\epsilon(\vec{H},\vec{O})\}\) with score \(s(\vec{H},\vec{O})\),
\[
\mathbb{E}\big[\phi(\vec{H},\vec{O};\psi,\eta)\big]=0,
\quad \text{and} \quad
\left.\frac{d}{d\epsilon}\psi(p_\epsilon)\right|_{\epsilon=0}
=
\mathbb{E}\big[\phi(\vec{H},\vec{O};\psi,\eta)\,s_\epsilon(\vec{H},\vec{O})\big].
\]
\end{definition}

After full-law identification is established, the observed-data law and full-data law are in one-to-one correspondence. Consequently, nuisance functions \(\eta\) defined as functionals of the full-data law are also determined by the observed-data law. By an abuse of notation, we use \(\eta\) to denote both representations, with the intended interpretation determined by context.

\section{Identification}\label{sec:identification} 
We adopt the problem setup discussed in Section~\ref{subsec:setup}, where collection \(\vec H\) consists of unobserved variables of interest and collection \(\vec O\) consists of observed variables. Again, when these collections appear as arguments of distributions or are assigned values, we use the same notation to denote the corresponding random vectors obtained by arranging the variables in each collection in a fixed order. 

\subsection{Identification of the Full-Data Law}\label{subsec:id_assumps_full1}

In this section, we establish assumptions which identify the full-data law. As previously noted, in some applications, the target of interest is invariant to latent state labeling, making label identification unnecessary. In such settings, our assumptions can be adapted to identify the full-data law up to arbitrary permutation of the support values of
\(\vec H\), as discussed in Section~\ref{subsec:id_assumps_full2} 

\begin{assumption}\label{assump:bounded_density}
\((\vec{H},  \vec{O})\) admit a joint density that is bounded. Each variable in \(\vec H\) or \(\vec O\) is either finite-support or continuously distributed with support \(\mathbb{R}^{d}\), where the dimension \(d\) may differ across variables. 
\end{assumption}

\begin{assumption}\label{assump:mutual_independence}
There exists a collection of variables
\(
\dot{\bigcup}_{H \in \vec H} \{W_{1_H}, W_{2_H}, W_{3_H}\}
\subseteq \vec O
\),where \(\dot{\bigcup}\) denotes a disjoint union,
such that, for each \(H \in \vec H\),
\(W_{1_H}, W_{2_H}, W_{3_H}\) are mutually independent conditional on
\(
H,
\vec O \setminus
\bigcup_{H' \in \vec H} \{W_{1_{H'}}, W_{2_{H'}}, W_{3_{H'}}\}.
\)
\end{assumption}

For notational convenience, we henceforth denote 
\[
\vec O_{-W}
=
\vec O \setminus
\bigcup_{H'\in\vec H}
\{W_{1_{H'}},W_{2_{H'}},W_{3_{H'}}\}\] which takes values \(\vec o_{-W}\).

\begin{assumption}\label{assump:m_sep1} 
For each \(H \in \vec H\), we have that 
\[ (W_{1_H}, W_{2_H}, W_{3_H}) \ci \left( \bigcup_{H' \in \vec H \setminus \{H\}} \{W_{1_{H'}}, W_{2_{H'}}, W_{3_{H'}}\}, \, \vec H \setminus \{H\} \right) \text{ conditional on }  H,\, \vec O_{-W}. \] 
\end{assumption}

\begin{assumption}\label{assump:m_sep2}
For each \(H \in \vec{H}\), we have that
\(
H \ci
\vec H \setminus \{H\}
\text{ conditional on } \vec O_{-W}. \)
\end{assumption}

In existing examples from the literature that recover the full-data law, there is often a single latent variable, \(\vec H = \{H\}\), rendering Assumptions~\ref{assump:m_sep1}-\ref{assump:m_sep2} vacuous. 
Figure~\ref{fig:proxies} and Figures~\ref{fig:3}(a)-(e) give examples of such settings. In particular, Figure~\ref{fig:proxies} is adapted from \citet{hu2008instrumental};
Figure~\ref{fig:3}(a) is drawn from \citet{kuroki14measurement},
\citet{allman2009identifiability}, and \citet{deaner2023controllinglatentconfoundingtriple}; and Figures~\ref{fig:3}(b),
\ref{fig:3}(d), and \ref{fig:3}(e) are drawn from
\citet{zhou2024causalinferencehiddentreatment}, \citet{Guo2026_outcome},
and \citet{Guo2026_confounder}, respectively. 

In Section~\ref{sec:graphical_char}, we describe graphical models compatible with our independence assumptions, including in settings with multiple latent variables, exemplified by the tree-structured graphs in Figures~\ref{fig:graphical_structure1}(a) and~\ref{fig:graphical_structure1}(b).

If more than three mutually independent proxies (here, \(W_{1_H}, W_{2_H}, W_{3_H}\)) are available for each \(H \in \vec H\), our identification and estimation results continue to hold. Work by \citet{egami2026debiasedinferenceaigenerateddata} considers identification and estimation in the setting with a single latent variable, where \(\vec H = \{H\}\) and its proxies have finite support. Our identification results apply more generally to multiple latent variables and allow both the latent variables and their proxies to have continuous support. Moreover, our estimation results discuss a broader range of target functionals; and, while our estimation results often assume (possibly multiple) finite-support latent variables, these results accommodate proxies with continuous support. However, beyond our scope, \citet{egami2026debiasedinferenceaigenerateddata} study how more than three conditionally mutually independent proxies can improve estimation, demonstrating efficiency gains from incorporating additional proxies.

\begin{assumption}\label{assump:completeness}
For each \(H \in \vec{H}\), each value $\vec o_{-W}$, and any square-integrable function \(g\),
\begin{align*}
&\mathbb{E}\{g(H) \mid W_{1_H}, \vec o_{-W}\} = 0 \ \text{a.s. iff } g(H) = 0 \ \text{a.s.};
\\
&\mathbb{E}\{g(H) \mid W_{2_H}, \vec o_{-W}\} = 0 \ \text{a.s. iff } g(H) = 0 \ \text{a.s.}
\end{align*}
\end{assumption}

Assumption~\ref{assump:completeness} gives a set of 
\emph{completeness} assumptions which imposes sufficient variability of distributions in \(\{p(W_{1_H} \mid H=h, \vec o_{-W}) \text{: } h  \text{ (in the support of } H \text{ when } \vec O_{-W} = \vec o_{-W})\}\) and \(\{p(W_{2_H} \mid H=h, \vec o_{-W}) \text{: } h\}\) for each value $\vec o_{-W}$. For a more detailed explanation of completeness assumptions, see \citet{tchetgentchetgen2024proximal}, \citet{semiparametric2020cui}, and \citet{Guo2025}. Assumption~\ref{assump:completeness} is considered mild in the sense that, provided \(W_{1_H}\) and \(W_{2_H}\) each have higher dimension than \(H\), any distribution violating them can be approximated arbitrarily well in total variation distance by distributions that satisfy them \citep{semiparametric2020cui, Canay_testability2013}.

\begin{assumption}\label{assump:distinctness}
For each \(H \in \vec{H}\), each value \(\vec o_{-W}\), and any two
distinct values \(h \neq h'\), 
\[
p\!\left(
W_{3_H}
\,\middle|\,
H=h,\,
\vec o_{-W}
\right)
\quad\text{and}\quad
p\!\left(
W_{3_H}
\,\middle|\,
H=h',\,
\vec o_{-W}
\right)
\]
differ with positive probability under the marginal distribution of
\(W_{3_H}\).
\end{assumption}

Assumption~\ref{assump:distinctness} is weaker than a completeness assumption, requiring only that the conditional distributions in \(\{p(W_{3_H} \mid H=h, \vec o_{-W}) \text{: } h  \text{ (in the support of } H \text{ when } \vec O_{-W} = \vec o_{-W})\}\) be distinct for each \(\vec o_{-W}\). Assumptions~\ref{assump:completeness}--\ref{assump:distinctness} imply that for every \(W_{k_H}\) in the set \(\{W_{1_H}, W_{2_H}, W_{3_H}\}\),  \(W_{k_H}\) and \(H\) are dependent conditional on \(\vec O_{-W}\). We therefore refer to \(W_{1_H}\), \(W_{2_H}\), and \(W_{3_H}\) as “proxies” of \(H\), and 
note that the roles of \(W_{2_H}\), \(W_{2_H}\), and \(W_{3_H}\) may be interchanged, provided the reassignment is made consistently throughout our assumptions.

\begin{lemma}\label{lemma:id_lemma1}
Under Assumptions~\ref{assump:bounded_density}--\ref{assump:distinctness}, for each value \(\vec o_{-W}\), the law
\[p(\vec H, \vec o_{-W}, \bigcup_{H \in \vec H} \{W_{1_H},W_{2_H},W_{3_H}\})\]
is identified up to arbitrary permutation of the support values of \(\vec H\).
\end{lemma}

A proof is provided in Appendix~\ref{app:id_proof1}. Note that Assumptions~\ref{assump:bounded_density}--\ref{assump:distinctness} uniquely determine the support of \((\vec H,\vec O)\). For practical purposes, however, the support is often assumed to be known \emph{a priori}, simplifying the search procedures used to recover the law. To recover the full-data law, we may additionally impose Assumption~\ref{assump:fix_labels}.

\begin{assumption}\label{assump:fix_labels}

There exists a known functional \(M\) such that for each value \(\vec o_{-W}\),
\[
M\!\left(p(\vec H = \vec h, \vec o_{-W}, \bigcup_{H \in \vec H} \{W_{1_H},W_{2_H},W_{3_H}\})\right)=\vec h.
\]
\end{assumption}

\begin{thm}\label{thm:id_thm_full}
Under Assumptions~\ref{assump:bounded_density}-\ref{assump:fix_labels}, the full-data law
\(p(\vec H,\vec O)\)
is identified.
\end{thm}

The proof follows directly from Lemma~\ref{lemma:id_lemma1} and Assumption~\ref{assump:fix_labels}. Note that Assumption~\ref{assump:fix_labels} may be replaced by any condition that uniquely determines the labels of the latent values of \(\vec H\) in Theorem~\ref{thm:id_thm_full}. 

The assumptions invoked in Theorem~\ref{thm:id_thm_full} are adapted from \citet{hu2008instrumental}, which state the results for the case of a single latent variable \(\vec{H} = \{H\}\) and only three observed variables \(\vec{O} = \{W_{1_H}, W_{2_H}, W_{3_H}\}\). See Figure~\ref{fig:proxies} for a graphical representation consistent with their assumptions. Theorem~\ref{thm:id_thm_full} also subsumes, up to minor modifications to Assumption~\ref{assump:fix_labels}, the arguments developed in work from \citet{zhou2024causalinferencehiddentreatment}, consistent with the model depicted in Figure~\ref{fig:3}(b), and \citet{Guo2026_outcome},  consistent with the model depicted in Figure~\ref{fig:3}(d).

Under identification of the full-data law, any functional of the full-data law is identified, including the mean of the latent variable(s) of interest \(\mathbb{E}[\vec H]\) in Figure~\ref{fig:proxies}, Figures~\ref{fig:3}(a)--\ref{fig:3}(e), and Figures~\ref{fig:graphical_structure1}(a)-(b), as well as the average counterfactual outcome \(\mathbb{E}[Y(a)]\) in Figures~\ref{fig:3}(a)--\ref{fig:3}(e). 

\subsection{Identification of the Full-Data Law  up to Permutation of the Support Values of \(\vec H\)}\label{subsec:id_assumps_full2}

In some applications, identifying the full-data law \(p(\vec H, \vec O)\) up to arbitrary permutation of the support values of \(\vec H\) is sufficient to identify the target of interest. For instance, the average counterfactual outcome $Y$ would have attained setting \(A=a\) in Figure~\ref{fig:3}(a), \[\mathbb{E}[Y(a)] = \iint y p(y \mid u, a)p(u) dy du  = \mathbb{E}[\mathbb{E}[Y \mid a, U]],\] is a functional of \(p(Y,A,U)\) invariant to arbitrary permutation of the support values of \(U\). Likewise, the targets corresponding to the average counterfactual outcome
\(\mathbb{E}[Y(a)]\) in Figures~\ref{fig:3}(c) and  \ref{fig:3}(e), given respectively by
\(
\iiint y\,p(y\mid a',m)p(m\mid a)p(a')\,dy\,dm\,da',
\)
and
\(
\iiiint y\,p(y\mid a',m,u)p(m\mid a,u)p(a',u)\,dy\,dm\,da'\,du,
\)
are functionals of the full-data law invariant to arbitrary
permutation of the support values of \(\vec H\)
\citep{pearl95causal, fulcher19robust, Guo2026_confounder}. 

In such settings, we may replace Assumption~\ref{assump:fix_labels} with Assumption~\ref{assump:m_sep3} to recover the full-data law up to arbitrary permutation of the support values of
\(\vec H\). Indeed, Assumptions~\ref{assump:bounded_density}-\ref{assump:distinctness} and Assumption~\ref{assump:m_sep3} give an argument from \citet{kuroki14measurement}, \citet{allman2009identifiability}, and \citet{deaner2023controllinglatentconfoundingtriple} under latent confounding, consistent with the model depicted in Figure~\ref{fig:3}(a).

\begin{assumption}\label{assump:m_sep3}
For each \(H \in \vec{H}\), we have that
\(
W_{1_H} \ci \vec O_{-W} \text{ conditional on }  H.
\)
\end{assumption}

\begin{thm}\label{thm:id_thm_up_to_label}
Under Assumptions~\ref{assump:bounded_density}-\ref{assump:distinctness} and Assumption~\ref{assump:m_sep3}, the full-data law \(p(\vec H, \vec O)\) is identified up to arbitrary permutation of the support values of \(\vec H\).
\end{thm}

By Lemma~\ref{lemma:id_lemma1}, \(
p(\vec H = \vec h_i, \vec o_{-W}, \bigcup_{H \in \vec H} \{W_{1_H},W_{2_H},W_{3_H}\})
\) is identified for each value \(\vec o_{-W}\). With the inclusion of Assumption~\ref{assump:m_sep3},
\(
p( \bigcup_{H \in \vec H} \vec W_{1_H} \mid \vec H=\vec h_i)
=
p(\bigcup_{H \in \vec H} \vec W_{1_H} \mid \vec H=\vec h_i,\vec o_{-W}),
\)
providing a common reference distribution so that labels \(\vec h_i\)
can be aligned across different values \(\vec o_{-W}\).

It is worth noting that one may instead first impose Assumptions~\ref{assump:bounded_density}--\ref{assump:distinctness} together with Assumption~\ref{assump:m_sep3} to identify the full-data law \(p(\vec H, \vec O)\) up to arbitrary permutation of the support values of \(\vec H\)
and subsequently impose a label-fixing assumption to identify \(p(\vec H,\vec O)\). We omit this formulation for brevity.

\section{Compatible ADMGs}\label{sec:graphical_char}
Checking \(m\)-separation in an ADMG allows us to determine whether the independences encoded by the graph are compatible with those imposed by our assumptions. In particular, ADMGs reflecting statistical models that satisfy the independences in Assumptions~\ref{assump:mutual_independence}--\ref{assump:m_sep2} are compatible with the assumptions used to identify \(p(\vec H,\vec O)\) in Theorem~\ref{thm:id_thm_full}. Similarly, ADMGs reflecting statistical models that additionally satisfy the conditional independence(s) in Assumption~\ref{assump:m_sep3} are compatible with the assumptions used to identify the full-data law \(p(\vec H,\vec O)\) up to arbitrary permutation of the support values of \(\vec H\) in Theorem~\ref{thm:id_thm_up_to_label}. The Fast Causal Inference (FCI) algorithm characterizes the complete class of ADMGs encoding a given set of independences \citep{spirtes01causation}, thereby allowing us to classify ADMGs reflecting statistical models compatible with our assumption sets.

\subsection{Compatability with Assumptions~\ref{assump:mutual_independence}--\ref{assump:m_sep2}}
ADMGs appearing in previous literature that identifies the full-data law \(p(\vec H, \vec O)\) often reflect Assumptions~\ref{assump:mutual_independence}--\ref{assump:m_sep2} via the sufficient graphical criterion given in Proposition~\ref{prop:graph_char_1}. For vertices \(V,V' \in \vec H \cup \vec O\) in an ADMG
\(\mathcal{G}(\vec H,\vec O)\), we say that \(V\) is a parent of
\(V'\) (and \(V'\) is a child of \(V\)) if \(V \to V'\), and that \(V\)
and \(V'\) are siblings if \(V \leftrightarrow V'\). 

The simple DAG in Figure~\ref{fig:proxies}, with \(\vec H = \{H\}\) and \(\vec{O}=\{W_{1_H},W_{2_H},W_{3_H}\}=\{W,Z,Y\}\), is an example of a graph  \(\mathcal{G}(\vec{H},\vec{O})\) that falls under the criterion in Proposition~\ref{prop:graph_char_1}. Figure~\ref{fig:3}(a)--(e) gives examples of a range of settings under this criterion.

\begin{proposition}\label{prop:graph_char_1}
If ADMG \(\mathcal{G}(\vec{H},\vec{O})\) has a single latent variable \(\vec H =\{H\}\), and \(W_{1_H}\), \(W_{2_H}\), \(W_{3_H} \in \vec O\) are children of \(H\) with no children or siblings, we have that statistical model given by \(\mathcal{G}(\vec{H},\vec{O})\) is compatible with Assumptions~\ref{assump:mutual_independence}--\ref{assump:m_sep2}.
\end{proposition}
 
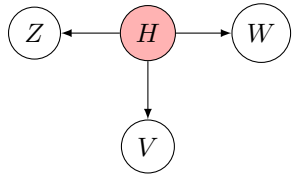
\begin{figure}[H]
    \centering
    \begin{minipage}{0.3\textwidth}
        \centering
        \begin{tikzpicture}[>=latex, node distance=1.5cm]
            \node (H) [draw, circle, fill=red!30] {\(H\)};
            \node (V) [draw, circle, below of=H] {\(V\)};
            \node (W) [draw, circle, right of=H] {\(W\)};
            \node (Z) [draw, circle, left of=H] {\(Z\)};
            \draw[->] (H) -- (Z);
            \draw[->] (H) -- (W);
            \draw[->] (H) -- (V);
        \end{tikzpicture}

        \vspace{0.3em}        
    \end{minipage}
    \hfill
    \caption{latent variable \(\vec H = \{H\}\) with proxies \(\{W_{1_H},W_{2_H},W_{3_H}\} = \{W,Z,V\}\).}
    \label{fig:proxies}
\end{figure}

\begin{figure}[H]
    \centering
    \begin{minipage}{0.31\textwidth}
            \centering
            \begin{tikzpicture}[>=latex, node distance=1.15cm, anchor=center]
            \node (A) [draw, circle] {A};
            \node (Y) [draw, circle, right of=A] {Y};
            \node (U) [draw, circle, fill=red!30, above of=A, xshift=0.6cm] {U};
            \node (W) [draw, circle, right of=U] {W};
            \node (Z) [draw, circle, left of=U] {Z};

            \draw[->] (A) -- (Y);
            \draw[->] (U) -- (Z);
            \draw[->] (U) -- (W);
            \draw[->] (U) -- (Y);
            \draw[->] (U) -- (A);
            \draw[->] (A) -- (Z);
        \end{tikzpicture}
\subcaption{latent confounder with proxies: \\
                    \(\vec{H} = \{U\}\); \\
                    \(\vec{O} = \{A,Y,W,Z\}; \)\\
                    \(\{W_{1_U},W_{2_U},W_{3_U}\} = \{W,Z,Y\}\).
}
            \vspace{0.2em}

        \end{minipage}\hfill
    \begin{minipage}{0.31\textwidth}
        \centering
               \begin{tikzpicture}[>=latex, node distance=1.15cm, anchor=center]
            \node (A) [draw, fill=red!30, circle] {A};
            \node (Y) [draw, circle, right of=A] {Y};
            \node (C) [draw, circle, above of=A, xshift=0.7cm] {C};
            \node (W) [draw, circle, below of=A, xshift=-0.7cm] {W};
            \node (Z) [draw, circle, below of=A, xshift=0.7cm] {Z};
            \draw[->] (A) -- (Y);
            \draw[->] (A) -- (W);
            \draw[->] (A) -- (Z);
            \draw[->] (C) -- (Z);
            \draw[->] (C) to[out=170, in=120] (W);
            \draw[->] (C) -- (Y);
            \draw[->] (C) -- (A);
        \end{tikzpicture}
\subcaption{latent treatment with proxies: \\
                    \(\vec{H} = \{A\}\); \\
                    \(\vec{O} = \{C,Y,W,Z\}; \)\\
                    \(\{W_{1_A},W_{2_A},W_{3_A}\} = \{W,Z,Y\}\).
}
        \vspace{0.2em}
    \end{minipage}\hfill
    \begin{minipage}{0.31\textwidth}
        \centering
         \begin{tikzpicture}[>=latex, node distance=1.15cm, anchor=center]
        \node (A) [draw, circle] {A};
        \node (M) [draw, circle, fill=red!30, right of=A] {M};
        \node (Y) [draw, circle, right of=M] {Y};
        \node (W) [draw, circle, below of=Y] {W};
        \node (Z) [draw, circle, below of=A] {Z};

        \draw[->] (A) -- (M);
        \draw[->] (M) -- (Y);
        \draw[->] (M) -- (W);
        \draw[->] (M) -- (Z);
        \draw[->] (A) -- (Z);

        \draw[<->, red, thick, bend left=30] (A) to (Y);
        
    \end{tikzpicture}
\subcaption{latent mediator with proxies: \\
                    \(\vec{H} = \{M\}\); \\
                    \(\vec{O} = \{A,Y,W,Z\}; \)\\
                    \(\{W_{1_M},W_{2_M},W_{3_M}\} = \{W,Z,Y\}\).
}
        \vspace{0.2em}

    \end{minipage}

    \vspace{0.8em}

    \makebox[\textwidth][c]{%
      \begin{minipage}{0.31\textwidth}
        \centering
        \begin{tikzpicture}[>=latex, node distance=1.15cm, anchor=center]
            \node (A) [draw, circle] {A};
            \node (Y) [draw, circle, fill=red!30, right of=A] {Y};
            \node (C) [draw, circle, above of=A, xshift=0.7cm] {C};
            \node (W) [draw, circle, below of=Y, xshift=-0.7cm] {W};
            \node (Z) [draw, circle, below of=Y, xshift=0.6cm] {Z};
            \node (V) [draw, circle, below of=Y, xshift=1.7cm] {V};

            \draw[->] (A) -- (Y);
            \draw[->] (Y) -- (W);
            \draw[->] (Y) -- (Z);
            \draw[->] (Y) -- (V);
            \draw[->] (A) -- (Z);
            \draw[->] (A) -- (W);
            \draw[->] (A) -- (V);
            \draw[->] (C) -- (Y);
            \draw[->] (C) -- (A);
            \draw[->] (C) to[out=10, in=90] (Z);
            \draw[->] (C) -- (W);
            \draw[->] (C) to[out=10, in=90] (V);
        \end{tikzpicture}
        \subcaption{latent outcome with proxies: 
                    \(\vec{H} = \{Y\}\); \\
                    \(\vec{O} = \{A,C,W,Z,V\}; \)\\
                    \(\{W_{1_Y},W_{2_Y},W_{3_Y}\} = \{W,Z,V\}\).
}
        \vspace{0.2em}
    \end{minipage}
        \hspace{0.04\textwidth}
        \begin{minipage}{0.31\textwidth}
            \centering
            \begin{tikzpicture}[>=latex, node distance=1.15cm, anchor=center,
                    transform shape]
    \node (A) [draw, circle] {A};
    \node (M) [draw, circle, right of=A] {M};
    \node (Y) [draw, circle, right of=M] {Y};
    \node (U) [draw, circle, fill=red!30, above of=A, xshift=1.1cm] {U};
    \node (W) [draw, circle, right of=U] {W};
    \node (Z) [draw, circle, left of=U] {Z};

    \draw[->] (A) -- (M);
    \draw[->] (M) -- (Y);
    \draw[->] (U) -- (Z);
    \draw[->] (U) -- (W);
    \draw[->] (U) -- (Y);
    \draw[->] (U) -- (A);
    \draw[->] (U) -- (M);
    \draw[->] (A) -- (Z);
    \draw[<->, red, thick, bend right=30] (A) to (Y);
\end{tikzpicture}
\subcaption{latent confounder with proxies; observed mediator: \\
                    \(\vec{H} = \{U\}\); \\
                    \(\vec{O} = \{A,M,Y,W,Z\}; \)\\
                    \(\{W_{1_U},W_{2_U},W_{3_U}\} = \{W,Z,Y\}\).
}
            \vspace{0.2em}

        \end{minipage}
    }

    \caption{Examples of ADMGs with proxy variables for a latent outcome (a), treatment (b), mediator (c), confounder (d)–(e), and confounder and mediator (e). Node partitions in each panel define the role of variables in our identification assumptions, with the partition in panel (a) drawn from \citet{kuroki14measurement},
\citet{allman2009identifiability}, and \citet{deaner2023controllinglatentconfoundingtriple}; and those in panels (b), (d), and (e) drawn from 
\citet{zhou2024causalinferencehiddentreatment}, \citet{Guo2026_outcome}, and 
\citet{Guo2026_confounder}, respectively.}
    \label{fig:3}
\end{figure}
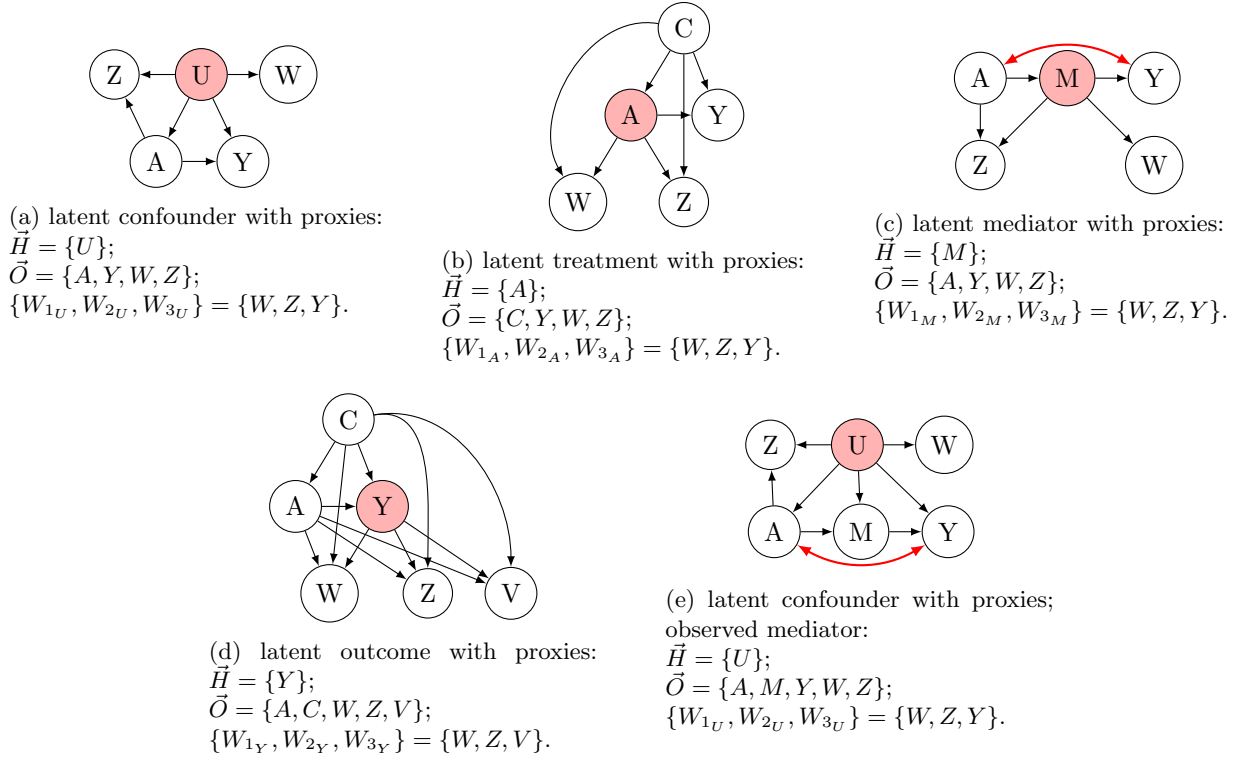

In Proposition~\ref{prop:graph_char_2}, we also give sufficient graphical criterion reflecting Assumptions~\ref{assump:mutual_independence}--\ref{assump:m_sep2} in settings with multiple latent variables. For example, the tree-structured graph in
Figure~\ref{fig:graphical_structure1}(a) falls within the class characterized by
Proposition~\ref{prop:graph_char_2}. Note that the characterization in Proposition~\ref{prop:graph_char_2} is intended to describe a class of potentially useful graphical
models rather than exhaust all compatible ADMGs. In particular, many ADMGs
outside this class may also reflect
Assumptions~\ref{assump:mutual_independence}--\ref{assump:m_sep2}. For an example, see Figure~\ref{fig:graphical_structure1}(b).

\begin{proposition}\label{prop:graph_char_2}
Let \(\mathcal{G}(\vec{H},\vec{O})\) be an ADMG such that for each \(H \in \vec H\), \(W_{1_H}\), \(W_{2_H}\), \(W_{3_H} \in \vec O\) are children of \(H\) with no children or siblings. Suppose further that each \(H \in \vec H\) has no siblings and no children other than \(W_{1_H}\), \(W_{2_H}\), \(W_{3_H}\). Then the statistical model given by \(\mathcal{G}(\vec{H},\vec{O})\) is compatible with Assumptions~\ref{assump:mutual_independence}--\ref{assump:m_sep2}.
\end{proposition}

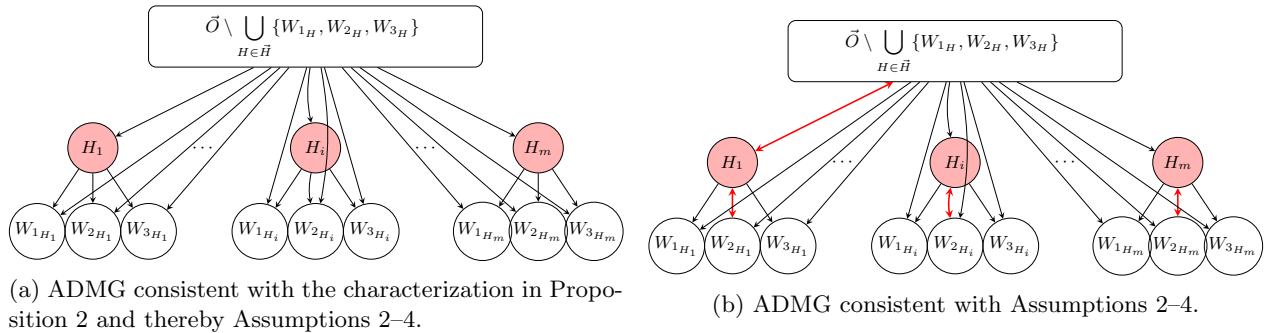
\begin{figure}[H]
\centering

\begin{subfigure}[H]{0.49\textwidth}
\centering
\resizebox{\linewidth}{!}{%
\begin{tikzpicture}[
    >=stealth,
    every node/.style={font=\small},
    latent/.style={circle, draw, minimum size=9mm, inner sep=1pt},
    observed/.style={circle, draw, minimum size=10mm, inner sep=1pt},
    remainder/.style={
        rectangle,
        draw,
        rounded corners,
        minimum width=6cm,
        minimum height=10mm,
        align=center,
        inner sep=4pt
    }
]

\node[remainder] (Orest) at (0,3.5) {
    $\displaystyle
    \vec O \setminus
    \bigcup_{H\in\vec H}
    \{W_{1_H},W_{2_H},W_{3_H}\}
    $
};

\node[latent,circle, fill=red!30] (H1) at (-4,1.5) {$H_1$};
\node[latent,circle, fill=red!30] (Hi) at (0,1.5) {$H_i$};
\node[latent,circle, fill=red!30] (Hm) at (4,1.5) {$H_m$};

\node at (-2,1.5) {$\cdots$};
\node at (2,1.5) {$\cdots$};

\node[observed] (W11) at (-5,0) {$W_{1_{H_1}}$};
\node[observed] (W21) at (-4,0) {$W_{2_{H_1}}$};
\node[observed] (W31) at (-3,0) {$W_{3_{H_1}}$};

\node[observed] (W1i) at (-1,0) {$W_{1_{H_i}}$};
\node[observed] (W2i) at (0,0) {$W_{2_{H_i}}$};
\node[observed] (W3i) at (1,0) {$W_{3_{H_i}}$};

\node[observed] (W1m) at (3,0) {$W_{1_{H_m}}$};
\node[observed] (W2m) at (4,0) {$W_{2_{H_m}}$};
\node[observed] (W3m) at (5,0) {$W_{3_{H_m}}$};

\foreach \W in {W11,W21,W31}
    \draw[->] (H1) -- (\W);

\draw[->] (Hi) -- (W1i);
\draw[->] (Hi) to[bend right=10] (W2i);
\draw[->] (Hi) -- (W3i);

\foreach \W in {W1m,W2m,W3m}
    \draw[->] (Hm) -- (\W);

\draw[->] (Orest) -- (H1);
\draw[->] (Orest) to[bend right=10] (Hi);
\draw[->] (Orest) -- (Hm);

\foreach \W in {W11,W21,W31,W1i,W3i,W1m,W2m,W3m}
    \draw[->] (Orest) -- (\W);

\draw[->] (Orest) to[bend right=-10] (W2i);

\end{tikzpicture}%
}
\caption{ADMG consistent with the characterization in
Proposition~\ref{prop:graph_char_2} and thereby
Assumptions~\ref{assump:mutual_independence}--\ref{assump:m_sep2}.}
\label{fig:graphical_structure1_a}
\end{subfigure}
\hfill
\begin{subfigure}[H]{0.49\textwidth}
\centering
\resizebox{\linewidth}{!}{%
\begin{tikzpicture}[
    >=stealth,
    every node/.style={font=\small},
    latent/.style={circle, draw, minimum size=9mm, inner sep=1pt},
    observed/.style={circle, draw, minimum size=10mm, inner sep=1pt},
    remainder/.style={
        rectangle,
        draw,
        rounded corners,
        minimum width=6cm,
        minimum height=10mm,
        align=center,
        inner sep=4pt
    }
]

\node[remainder] (Orest) at (0,3.5) {
    $\displaystyle
    \vec O \setminus
    \bigcup_{H\in\vec H}
    \{W_{1_H},W_{2_H},W_{3_H}\}
    $
};

\node[latent,circle, fill=red!30] (H1) at (-4,1.5) {$H_1$};
\node[latent,circle, fill=red!30] (Hi) at (0,1.5) {$H_i$};
\node[latent,circle, fill=red!30] (Hm) at (4,1.5) {$H_m$};

\node at (-2,1.5) {$\cdots$};
\node at (2,1.5) {$\cdots$};

\node[observed] (W11) at (-5,0) {$W_{1_{H_1}}$};
\node[observed] (W21) at (-4,0) {$W_{2_{H_1}}$};
\node[observed] (W31) at (-3,0) {$W_{3_{H_1}}$};

\node[observed] (W1i) at (-1,0) {$W_{1_{H_i}}$};
\node[observed] (W2i) at (0,0) {$W_{2_{H_i}}$};
\node[observed] (W3i) at (1,0) {$W_{3_{H_i}}$};

\node[observed] (W1m) at (3,0) {$W_{1_{H_m}}$};
\node[observed] (W2m) at (4,0) {$W_{2_{H_m}}$};
\node[observed] (W3m) at (5,0) {$W_{3_{H_m}}$};

\foreach \W in {W11,W31}
    \draw[->] (H1) -- (\W);
\draw[<->,red,thick] (H1) -- (W21);

\draw[->] (Hi) -- (W1i);
\draw[<->,red,thick] (Hi) to[bend right=10] (W2i);
\draw[->] (Hi) -- (W3i);

\foreach \W in {W1m,W3m}
    \draw[->] (Hm) -- (\W);
\draw[<->,red,thick] (Hm) -- (W2m);

\draw[<->,red,thick] (Orest) -- (H1);
\draw[->] (Orest) to[bend right=10] (Hi);
\draw[->] (Orest) -- (Hm);

\foreach \W in {W11,W21,W31,W1i,W3i,W1m,W2m,W3m}
    \draw[->] (Orest) -- (\W);

\draw[->] (Orest) to[bend right=-10] (W2i);

\end{tikzpicture}%
}
\caption{ADMG consistent with
Assumptions~\ref{assump:mutual_independence}--\ref{assump:m_sep2}.}
\label{fig:graphical_structure1_b}
\end{subfigure}

\caption{Examples of ADMGs consistent with
Assumptions~\ref{assump:mutual_independence}--\ref{assump:m_sep2}.}
\label{fig:graphical_structure1}

\end{figure}

\subsection{Compatability with Assumptions~\ref{assump:mutual_independence}--\ref{assump:m_sep2} and Assumption~\ref{assump:m_sep3}}

ADMGs appearing in the literature that identify the full-data law \(p(\vec H, \vec O)\) up to arbitrary permutation of the support values of \(\vec H\) often fall under the graphical criterion  given in Proposition~\ref{prop:graph_char_3}. See Figures~\ref{fig:3}(a), \ref{fig:3}(c), and \ref{fig:3}(e) for examples.
The addition that \(H\) is the only parent of \(W_{1_H}\) for each \(H \in \vec H\) gives that  \(W_{1_H} \ci \vec O \setminus \bigcup_{H' \in \vec H} \{W_{1_{H'}}, W_{2_{H'}}, W_{3_{H'}}\} \text{ conditional on }  H\) for each \(H \in \vec H\), satisfying Assumption~\ref{assump:m_sep3}.

\begin{proposition}\label{prop:graph_char_3}
If ADMG \(\mathcal{G}(\vec{H},\vec{O})\) has a single latent variable \(\vec H = \{H\}\), and \(W_{1_H}\), \(W_{2_H}\), \(W_{3_H} \in \vec O\) are children of \(H\) with no children or siblings, and \(H\) is the only parent of \(W_{1_H}\), we have that the statistical model given by \(\mathcal{G}(\vec{H},\vec{O})\) is compatible with Assumption~\ref{assump:mutual_independence}-\ref{assump:m_sep2} and Assumption~\ref{assump:m_sep3}.
\end{proposition}

In Proposition~\ref{prop:graph_char_4}, we also give sufficient graphical criterion reflecting Assumptions~\ref{assump:mutual_independence}--\ref{assump:m_sep2} and
Assumption~\ref{assump:m_sep3} in settings with multiple latent variables. For example, the tree-structured graph in
Figure~\ref{fig:graphical_structure2}(a) falls within the class characterized by
Proposition~\ref{prop:graph_char_4}. Again, this characterization  is intended to describe a class of potentially useful graphical
models rather than exhaust all compatible ADMGs. For an example outside this class that reflects
Assumptions~\ref{assump:mutual_independence}--\ref{assump:m_sep2} and
Assumption~\ref{assump:m_sep3}, see Figure~\ref{fig:graphical_structure2}(b).

\begin{proposition}\label{prop:graph_char_4}
Let \(\mathcal{G}(\vec{H},\vec{O})\) be an ADMG such that for each \(H \in \vec H\), \(W_{1_H}\), \(W_{2_H}\), \(W_{3_H} \in \vec O\) are children of \(H\) with no children or siblings, and \(H\) is the only parent of \(W_{1_H}\). Suppose further that each \(H \in \vec H\) has no siblings and no children other than \(W_{1_H}\), \(W_{2_H}\), \(W_{3_H}\). Then the statistical model given by \(\mathcal{G}(\vec{H},\vec{O})\) is compatible with Assumptions~\ref{assump:mutual_independence}--\ref{assump:m_sep2} and Assumption~\ref{assump:m_sep3}.
\end{proposition}

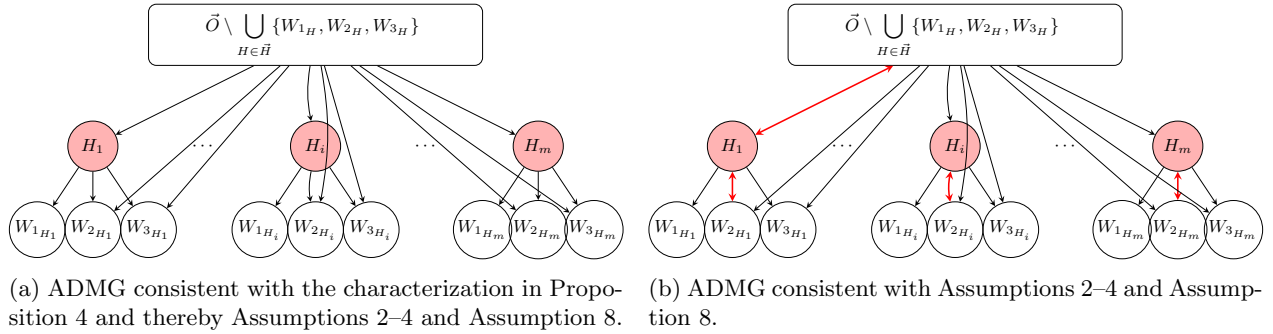
\begin{figure}[H]
\centering

\begin{subfigure}[H]{0.49\textwidth}
\centering
\resizebox{\linewidth}{!}{%
\begin{tikzpicture}[
    >=stealth,
    every node/.style={font=\small},
    latent/.style={circle, draw, minimum size=9mm, inner sep=1pt},
    observed/.style={circle, draw, minimum size=10mm, inner sep=1pt},
    remainder/.style={
        rectangle,
        draw,
        rounded corners,
        minimum width=6cm,
        minimum height=10mm,
        align=center,
        inner sep=4pt
    }
]

\node[remainder] (Orest) at (0,3.5) {
    $\displaystyle
    \vec O \setminus
    \bigcup_{H\in\vec H}
    \{W_{1_H},W_{2_H},W_{3_H}\}
    $
};

\node[latent,circle, fill=red!30] (H1) at (-4,1.5) {$H_1$};
\node[latent,circle, fill=red!30] (Hi) at (0,1.5) {$H_i$};
\node[latent,circle, fill=red!30] (Hm) at (4,1.5) {$H_m$};

\node at (-2,1.5) {$\cdots$};
\node at (2,1.5) {$\cdots$};

\node[observed] (W11) at (-5,0) {$W_{1_{H_1}}$};
\node[observed] (W21) at (-4,0) {$W_{2_{H_1}}$};
\node[observed] (W31) at (-3,0) {$W_{3_{H_1}}$};

\node[observed] (W1i) at (-1,0) {$W_{1_{H_i}}$};
\node[observed] (W2i) at (0,0) {$W_{2_{H_i}}$};
\node[observed] (W3i) at (1,0) {$W_{3_{H_i}}$};

\node[observed] (W1m) at (3,0) {$W_{1_{H_m}}$};
\node[observed] (W2m) at (4,0) {$W_{2_{H_m}}$};
\node[observed] (W3m) at (5,0) {$W_{3_{H_m}}$};

\foreach \W in {W11,W21,W31}
    \draw[->] (H1) -- (\W);

\foreach \W in {W1i,W3i}
    \draw[->] (Hi) -- (\W);
\draw[->] (Hi) to[bend right=10] (W2i);

\foreach \W in {W1m,W2m,W3m}
    \draw[->] (Hm) -- (\W);

\draw[->] (Orest) -- (H1);
\draw[->] (Orest) to[bend right=10] (Hi);
\draw[->] (Orest) -- (Hm);

\foreach \W in {W21,W31,W3i,W2m,W3m}
    \draw[->] (Orest) -- (\W);

\draw[->] (Orest) to[bend right=-10] (W2i);

\end{tikzpicture}%
}

\caption{ADMG consistent with the characterization in
Proposition~\ref{prop:graph_char_4} and thereby
Assumptions~\ref{assump:mutual_independence}--\ref{assump:m_sep2}
and Assumption~\ref{assump:m_sep3}.}
\label{fig:graphical_structure2_a}
\end{subfigure}
\hfill
\begin{subfigure}[H]{0.49\textwidth}
\centering
\resizebox{\linewidth}{!}{%
\begin{tikzpicture}[
    >=stealth,
    every node/.style={font=\small},
    latent/.style={circle, draw, minimum size=9mm, inner sep=1pt},
    observed/.style={circle, draw, minimum size=10mm, inner sep=1pt},
    remainder/.style={
        rectangle,
        draw,
        rounded corners,
        minimum width=6cm,
        minimum height=10mm,
        align=center,
        inner sep=4pt
    }
]

\node[remainder] (Orest) at (0,3.5) {
    $\displaystyle
    \vec O \setminus
    \bigcup_{H\in\vec H}
    \{W_{1_H},W_{2_H},W_{3_H}\}
    $
};

\node[latent,circle, fill=red!30] (H1) at (-4,1.5) {$H_1$};
\node[latent,circle, fill=red!30] (Hi) at (0,1.5) {$H_i$};
\node[latent,circle, fill=red!30] (Hm) at (4,1.5) {$H_m$};

\node at (-2,1.5) {$\cdots$};
\node at (2,1.5) {$\cdots$};

\node[observed] (W11) at (-5,0) {$W_{1_{H_1}}$};
\node[observed] (W21) at (-4,0) {$W_{2_{H_1}}$};
\node[observed] (W31) at (-3,0) {$W_{3_{H_1}}$};

\node[observed] (W1i) at (-1,0) {$W_{1_{H_i}}$};
\node[observed] (W2i) at (0,0) {$W_{2_{H_i}}$};
\node[observed] (W3i) at (1,0) {$W_{3_{H_i}}$};

\node[observed] (W1m) at (3,0) {$W_{1_{H_m}}$};
\node[observed] (W2m) at (4,0) {$W_{2_{H_m}}$};
\node[observed] (W3m) at (5,0) {$W_{3_{H_m}}$};

\foreach \W in {W11,W31}
    \draw[->] (H1) -- (\W);
\draw[<->,red,thick] (H1) -- (W21);

\foreach \W in {W1i,W3i}
    \draw[->] (Hi) -- (\W);
\draw[<->,red,thick] (Hi) to[bend right=10] (W2i);

\foreach \W in {W1m,W3m}
    \draw[->] (Hm) -- (\W);
\draw[<->,red,thick] (Hm) -- (W2m);

\draw[<->,red,thick] (Orest) -- (H1);
\draw[->] (Orest) to[bend right=10] (Hi);
\draw[->] (Orest) -- (Hm);

\foreach \W in {W21,W31,W3i,W2m,W3m}
    \draw[->] (Orest) -- (\W);

\draw[->] (Orest) to[bend right=-10] (W2i);

\end{tikzpicture}%
}

\caption{ADMG consistent with
Assumptions~\ref{assump:mutual_independence}--\ref{assump:m_sep2}
and Assumption~\ref{assump:m_sep3}.}
\label{fig:graphical_structure2_b}
\end{subfigure}

\caption{Examples of ADMGs consistent with
Assumptions~\ref{assump:mutual_independence}--\ref{assump:m_sep2}
and Assumption~\ref{assump:m_sep3}.}
\label{fig:graphical_structure2}

\end{figure}

\section{Estimating Equations}\label{sec:est_eqns}

\subsection{Estimating Equations under Full-Data Law Identification}\label{subsec:4a}
We first develop mean-zero observed-data estimating equations that give rise to consistent \(M\)-estimators for a large class of target functionals of the full-data law identified under Assumptions~\ref{assump:bounded_density}-\ref{assump:fix_labels} via Theorem~\ref{thm:id_thm_full}. These results establish a foundation that can be refined to yield observed-data influence functions, as shown in Section~\ref{sec:ifs}, inspired by previous work on binary latent treatments by \citet{zhou2024causalinferencehiddentreatment}.

For the purpose of developing mean-zero observed-data estimating equations, we consider a target \(\psi\) that can be expressed as a functional of
\[
p\left(
\vec{H},
\bigl(
\vec{O}\setminus
\bigcup_{H \in \vec H} W_{1_H}
\bigr)
\right),
\] for which we know a full-data function \(\phi(\vec{H}, \big(
\vec{O}\setminus
\bigcup_{H \in \vec H} W_{1_H}\big); \eta)\) satisfying
\begin{equation}\label{eq:full_data_eq1}
\mathbb{E}\big[\phi(\vec{H}, \big(
\vec{O}\setminus
\bigcup_{H \in \vec H} W_{1_H}\big), \eta) - \psi\big]=0.
\end{equation}
The choice of excluding \(\bigcup_{H \in \vec H} W_{1_H}\) is without loss of generality. All results hold if \(W_{1_H}\) is replaced by \(W_{2_H}\) or \(W_{3_H}\), provided that the corresponding roles of \(W_{1_H}\) and its replacement are interchanged consistently in our identifying assumptions.

Such targets include, for example, the mean of latent variable(s) of interest in Figure~\ref{fig:proxies}, Figures~\ref{fig:3}(a)-(e), and Figures~\ref{fig:graphical_structure1}(a)-(b), as well as the average counterfactual outcome \(\mathbb{E}[Y(a)]\) in Figures~\ref{fig:3}(a)-(e). 

Our results can be easily adapted for settings in which the full-data law is identified only up to the labeling of the latent variable(s). We delineate these adaptations in Section~\ref{subsec:4b}.

To develop an observed-data estimating equations for \(\psi\), we construct an observed-data function \(\varphi(\vec{O};\eta)\) satisfying
\begin{equation}\label{eq:bridge_1}
\mathbb{E}[\varphi(\vec{O};\eta)] =\mathbb{E}\Big[\phi(\vec{H}, \big(
\vec{O}\setminus
\bigcup_{H \in \vec H} W_{1_H}\big), \eta)\Big].
\end{equation}
Then, \(\varphi(\vec{O};\eta) - \psi\) is a mean-zero estimating equation for $\psi$. Hence, the estimator \(\hat \psi = \mathbb{P}_n\big[\varphi(O; \hat \eta)\big]\)
is a consistent \(M\)-estimator for \(\psi\) under standard regularity conditions.

Our identifying assumptions, together with the assumption that \(\vec H\) has finite support, are sufficient to guarantee the existence of \(\varphi(O; \eta)\) satisfying Equation~\ref{eq:bridge_1} (see Proposition~\ref{prop:guaranteed_solution1} and Theorem~\ref{thm:consistent_est}).
We leave the generalization of Proposition~\ref{prop:guaranteed_solution1} when \(\vec{H}\) is continuous to future work.

\begin{proposition}\label{prop:guaranteed_solution1}
Suppose Assumptions~\ref{assump:bounded_density}-\ref{assump:fix_labels} hold. If \(\vec H\) has finite support, then for each of its values \(\vec h_i\), there exist weight function
\(\omega_{\vec h_i}(\bigcup_{H \in \vec H} W_{1_H},\vec O_{-W};\eta)\) satisfying
\begin{equation}\label{eq:weight}
\mathbb{E}\!\left[
\omega_{\vec h_i}(\bigcup_{H \in \vec H} W_{1_H},\vec O_{-W};\eta)
\,\middle|\,
\vec H=\vec h_j,\vec O_{-W}
\right]
=
\mathbb{I}\{i=j\}.
\end{equation}
\end{proposition}

The proof of Proposition~\ref{eq:weight} is given in Appendix~\ref{app:weight1}, deriving weight \(\omega_{\vec h_i}(\bigcup_{H \in \vec H}W_{1_H},\vec O_{-W};\eta)\) that satisfies Equation~\ref{eq:weight}. Appendix~\ref{app:weight2} gives an alternative procedure for constructing such a weight
using simpler sets of nuisance functions, which may be preferable in practice. 

Theorem~\ref{thm:consistent_est} follows directly from Assumption~\ref{assump:mutual_independence}-\ref{assump:m_sep1} and iterated expectations.

\begin{thm}\label{thm:consistent_est}
Given
\(\phi(\vec H, \vec O \setminus \bigcup_{H \in \vec H} W_{1_H};\eta)\)
satisfying Equation~\ref{eq:full_data_eq1}, define
\begin{equation}\label{eq:obs_est_eq}
\varphi(\vec O;\eta)
=
\sum_{\vec h_i \in \operatorname{Supp}(\vec H \mid \vec O_{-W})}
\omega_{\vec h_i}
\!\left(
\bigcup_{H \in \vec H} W_{1_H},\vec O_{-W};\eta
\right)
\phi\!\left(
\vec H=\vec h_i,\,
\vec O \setminus \bigcup_{H \in \vec H} W_{1_H};
\eta
\right),
\end{equation}
where weight
\(\omega_{\vec h_i}(\bigcup_{H \in \vec H} W_{1_H},\vec O_{-W};\eta)\)
obeys Equation~\ref{eq:weight}. Then
\(\varphi(\vec O;\eta)\) satisfies Equation~\ref{eq:bridge_1}. Consequently,
\(\varphi(\vec O;\eta)-\psi\) is a mean-zero estimating equation for
\(\psi\).
\end{thm}

Because the full-data law is identified,
\(\varphi(\vec O;\eta)\) in Equation~\ref{eq:obs_est_eq} depends only on observed variables and the observed-data law.

\begin{example}\label{ex:mean}
Consider the target parameter \(\mu = \mathbb{E}[H]\) in the model depicted in Figure~\ref{fig:proxies}, with \(\vec H = \{H\}\) and \(\vec{O} = \{W_{1_H},W_{2_H},W_{3_H}\} = \{W,Z,Y\}\) satisfying Assumptions~\ref{assump:bounded_density}-\ref{assump:fix_labels}. Further suppose $H$ has finite support.

We know that
\(
\phi(\vec H, \vec O) = H 
\) has property \(\mathbb{E}[\phi(\vec H, \vec O) - \mu]=0.\) If \(H\) is binary, define for each of its values \(h_i\),
\[
\omega_{h_i}(W; \eta)
=
\frac{W - \mathbb{E}[W \mid H = h_j]}{\mathbb{E}[W \mid H = h_i] - \mathbb{E}[W \mid H = h_j]}, \quad h_i \neq h_j,
\] by taking \(b_H(W)= (
 b_{H,1}(W),
 b_{H,2}(W))^\top=(1,W)^\top\) in Appendix~\ref{app:weight2}. If \(H\) has more categories, define the weight
\(\omega_{h_i}(W;\eta)\) by expanding the number of basis functions in \(b_H(W)\) from Appendix~\ref{app:weight2}.
 
Then \[\varphi(W;\eta)= \sum_{h_i}
\omega_{h_i}(W; \eta) \cdot h_i\] has the property that \(\mathbb{E}[\varphi(W; 
\eta)] = H\). Hence, \(\mathbb{E}[\varphi(W; 
\eta) - \mu] = 0,\) which gives that \[\hat \mu =
\mathbb{P}_n[\varphi(W; \hat \eta)] =  \mathbb{P}_n \Big[
\sum_{h_i} \omega_{h_i}(W;\hat\eta) h_i
\Big]
\]
is a consistent \(M\)-estimator for \(\mu\) under standard regularity conditions.
\end{example} 

Example~\ref{ex:mean} can be adapted for the vector-valued target parameter
\[
\mu
=
\mathbb{E}[\vec H]
=
\bigl(\mathbb{E}[H_1],\ldots,\mathbb{E}[H_m]\bigr)^\top
\] in the model depicted in Figure~\ref{fig:graphical_structure1}(a)
or Figure~\ref{fig:graphical_structure1}(b). Letting
\(\vec h_i\) denote the \(i^{\text{th}}\) support
value of \(\vec H\), and \(h_i(H)\) denote the value of \(H\) in
\(\vec h_i\), we can construct a mean-zero estimating equation using the observed data by weighting each support
value \(\vec h_i\) using
\[
\omega_{\vec h_i}\!\left(
\bigcup_{H\in\vec H} W_{1_H};\eta
\right)
=
\prod_{H\in\vec H}
\omega_{h_i(H)}\!\left(
W_{1_H};\eta
\right),
\] where \(\omega_{h_i(H)}\!\left(
W_{1_H};\eta
\right)\) is constructed per Appendix~\ref{app:weight2}.

Generally, the basis functions in Appendix~\ref{app:weight2} may be selected from a rich class of functions, including polynomials, splines, and indicator functions, to ensure that each
\(
\omega_{\vec h_i}\!\left(
\bigcup_{H\in\vec H} W_{1_H}, \vec O_{-W};\eta
\right)
\)
satisfies Equation~\ref{eq:weight}. We next characterize a setting in which polynomial basis functions suffice.

\begin{proposition}\label{prop:polynomial_basis1}
Suppose the full-data law \(p(\vec H,\vec O)\) is identified. Further suppose \(\vec H\) has finite support and for each \(H \in \vec H,\)
\[
W_{1_H}\mid H=h_i,\vec O_{-W}
\sim
\mu_{H,i}(\vec O_{-W})+\epsilon_H(\vec O_{-W}),
\]
where the error term \(\epsilon_H(\vec O_{-W})\) does not depend on the value of \(H\), and for each \(\vec o_{-W}\),
\(
\mu_{H,i}(\vec o_{-W})
\neq
\mu_{H,j}(\vec o_{-W})
\text{ if } i\neq j.
\)
Then for each value \(\vec h_i\), weight
\(\omega_{\vec h_i}(\bigcup_{H \in \vec H} W_{1_H},\vec O_{-W};\eta)\)
satisfying Equation~\ref{eq:weight} may be constructed by taking for each \(H \in \vec H\) and value \(\vec o_{-W}\),
\[ b_H(W_{1_H})=
\bigl(
b_{H,1}(W_{1_H}),\ldots,
b_{H,K_H(\vec o_{-W})}(W_{1_H})
\bigr)^\top
=
\bigl(
1,W_{1_H},\ldots,
W_{1_H}^{K_H(\vec o_{-W})-1}
\bigr)^\top
\] in Appendix~\ref{app:weight2}.
\end{proposition}
A proof of Proposition~\ref{prop:polynomial_basis1} is provided in Appendix~\ref{app:polynomial_basis1}.

\subsection{Estimating Equations under Full-Data Law Identification up to Permutation of the Support Values of \(\vec H\)}\label{subsec:4b}

The results in Section~\ref{subsec:4a} can be adapted when the full-data law \(p(\vec H, \vec O)\) is identified only up to arbitrary permutation of the support values of \(\vec H\) under Assumptions~\ref{assump:bounded_density}-\ref{assump:distinctness} and Assumption~\ref{assump:m_sep3} via Theorem~\ref{thm:id_thm_up_to_label}. We consider a target functional \(\psi\) of
\[
p\left(
\vec H,
\vec O\setminus
\bigcup_{H\in\vec H} W_{1_H}
\right)
\]
that is invariant to relabeling of the support values of \(\vec H\), and for which we know a full-data function
\(
\phi\left(
\vec H,
\vec O\setminus
\bigcup_{H\in\vec H} W_{1_H};
\eta
\right)
\) that depends on \(\vec H\) only via nuisance functions and
satisfies Equation~\ref{eq:full_data_eq1}.

Examples of such target functionals include the average counterfactual outcome \(\mathbb{E}[Y(a)]\) in Figure~\ref{fig:3}(a), Figure~\ref{fig:3}(c), and Figure~\ref{fig:3}(e), given by
\(
\iint y\,p(y\mid u,a)p(u)\,dy\,du
\),
\(
\iiint y\,p(y\mid a',m)p(m\mid a)p(a')\,dy\,dm\,da'
\)
, and
\(
\iiiint y\,p(y\mid a',m,u)p(m\mid a,u)p(a',u)\,dy\,dm\,da'\,du
\), respectively 
\citep{pearl95causal,fulcher19robust,Guo2026_confounder}.

\begin{proposition}\label{prop:guaranteed_solution_alt}
Suppose Assumptions~\ref{assump:bounded_density}-\ref{assump:distinctness} and Assumption~\ref{assump:m_sep3} hold. If \(\vec H\) has finite support, then for each of its values \(\vec h_i\), there exist weight function
\(\omega_{\vec h_i}(\bigcup_{H \in \vec H} W_{1_H},\vec O_{-W};\eta)\) satisfying
\begin{equation}\label{eq:weight_alt}
\mathbb{E}\!\left[
\omega_{\vec h_i}(\bigcup_{H \in \vec H} W_{1_H},\vec O_{-W};\eta)
\,\middle|\,
\vec H=\vec h_j,\vec O_{-W}
\right]
=
\mathbb{I}\{i=j\}.
\end{equation}
\end{proposition}

The proof of Proposition~\ref{prop:guaranteed_solution_alt} is the same as that of Proposition~\ref{prop:guaranteed_solution1}, given in Appendix~\ref{app:weight1}. 

Theorem~\ref{thm:consistent_est_alt} follows directly from Assumption~\ref{assump:mutual_independence}-\ref{assump:m_sep1} and iterated expectations.

\begin{thm}\label{thm:consistent_est_alt}
Given \(\phi(\vec H, \vec O \setminus \bigcup_{H \in \vec H}W_{1_H}; \eta)\) satisfying Equation~\ref{eq:full_data_eq1} depends on \(\vec H\) only via nuisance functions, define
\begin{equation}\label{eq:obs_est_eq_alt}
\varphi(\vec O;\eta)
=
\sum_{\vec h_i \in \operatorname{Supp}(\vec H \mid \vec O_{-W})}
\omega_{\vec h_i}(\bigcup_{H \in \vec H} W_{1_H},\vec O_{-W};\eta)\,
\phi (\vec H= \vec h_i, \vec O \setminus \bigcup_{H \in \vec H}W_{1_H}; \eta),
\end{equation}
where weight
\(\omega_{\vec h_i}(\bigcup_{H \in \vec H} W_{1_H},\vec O_{-W};\eta)\)
obeys Equation~\ref{eq:weight_alt}. Then
\(\varphi(\vec O;\eta)\) satisfies Equation~\ref{eq:bridge_1}. Consequently,
\(\varphi(\vec O;\eta)-\psi\) is a mean-zero estimating equation for
\(\psi\).
\end{thm}

\begin{example}\label{ex:AIPW}
Consider the target parameter \(\psi= \mathbb{E}[Y(a)] = \mathbb{E}[\mathbb{E}[Y \mid a, U]]\) in the model depicted in Figure~\ref{fig:3}(a), with \(\vec H = \{U\}\) and \(\{W_{1_H},W_{2_H},W_{3_H}\} = \{W,Z,Y\}\) satisfying Assumptions~\ref{assump:bounded_density}-\ref{assump:distinctness} and Assumption~\ref{assump:m_sep3}. Suppose \(U\) has finite support.

We know that
\(
\phi(U; \eta) = \mathbb{E}[Y \mid a, U] 
\) has property \(\mathbb{E}[\phi(U; \eta) - \psi]=0.\) If $U$ is binary, define for each of its values \(u_i\),
\[
\omega_{u_i}(Z,A; \eta)
=
\frac{Z - \mathbb{E}[Z \mid U = u_j,A]}{\mathbb{E}[Z \mid U = u_i,A] - \mathbb{E}[Z \mid U = u_j,A]}, \quad u_i \neq u_j,
\] by taking \(b_{U}(Z)=
( b_{U,1}(Z),
 b_{U,2}(Z))^\top = (1,Z)^\top
\) in Appendix~\ref{app:weight2}. If \(U\) has more than two categories, then define the weight
\(\omega_{u_i}(Z,A;\eta)\) by expanding the number of basis functions in \(b_{U}(Z)\) from Appendix~\ref{app:weight2}.

Then for \[\varphi(Z,A;\eta)= \sum_{u_i \in \operatorname{Supp}(U \mid A)}
\omega_{u_i}(Z, A; \eta) \cdot \mathbb{E}[Y \mid a, u_i],\] we have that \(\mathbb{E}[\varphi(Z,A; 
\eta) - \psi] = \mathbb{E}[\varphi(Z,A; 
\eta)] - \psi  = \mathbb{E}[\mathbb{E}[\varphi(Z,A; 
\eta) \mid U,A]] - \mathbb{E}[\mathbb{E}[Y \mid a, U]] = 0.\) Hence, \[\hat \psi =\mathbb{P}_n \Big[
 \sum_{u_i \in \operatorname{Supp}(U \mid A)}
\omega_{u_i}(Z, A; \hat\eta) \cdot \mathbb{E}[Y \mid a, u_i] \Big]
\]
is a consistent \(M\)-estimator for \(\psi\) under standard regularity conditions.
\end{example}

Again, the basis functions in Appendix~\ref{app:weight2} may be selected from a rich class of functions to ensure that each
\(
\omega_{\vec h_i}\!\left(
\bigcup_{H\in\vec H} W_{1_H}, \vec O_{-W};\eta
\right)
\)
satisfies Equation~\ref{eq:weight}. In the setting described below, selecting polynomial basis functions suffice.

\begin{proposition}\label{prop:polynomial_basis_alt}
Suppose the full-data law \(p(\vec H, \vec O)\) is identified up to arbitrary permutation of the support values of \(\vec H\).  Further suppose \(\vec H\) has finite support and for each \(H \in \vec H,\)
\[
W_{1_H}\mid H=h_i,\vec O_{-W}
\sim
\mu_{H,i}(\vec O_{-W})+\epsilon_H(\vec O_{-W}),
\]
where the error term \(\epsilon_H(\vec O_{-W})\) does not depend on the value of \(H\), and for each \(o_{-W}\),
\(
\mu_{H,i}(\vec o_{-W})
\neq
\mu_{H,j}(\vec o_{-W})
\text{ if } i\neq j.
\)
Then for each \(\vec h_i\), weight
\(\omega_{\vec h_i}(\bigcup_{H \in \vec H} W_{1_H},\vec O_{-W};\eta)\)
satisfying Equation~\ref{eq:weight} may be constructed by taking for each \(H \in \vec H\) and value \(\vec o_{-W}\),
\[
b_H(W_{1_H})=
\bigl(
b_{H,1}(W_{1_H}),\ldots,
b_{H,K_H(\vec o_{-W})}(W_{1_H})
\bigr)^\top
=
\bigl(
1,W_{1_H},\ldots,
W_{1_H}^{K_H(\vec o_{-W})-1}
\bigr)^\top
\]
in Appendix~\ref{app:weight2}. 
\end{proposition}

The proof of Proposition~\ref{prop:polynomial_basis_alt} is the same as that of Proposition~\ref{prop:polynomial_basis1} given in Appendix~\ref{app:polynomial_basis1}.

\section{Influence Functions}\label{sec:ifs}

We refine our approach for constructing mean-zero estimating equations to yield observed-data influence functions, which give rise to estimators exhibiting desirable properties, such as multiple robustness and \(\sqrt{n}\)-consistency even when nuisance function estimators converge at slower-than-parametric rates. 

To develop the main ideas progressively, we start in Section~\ref{subsec:class1} by considering a target functional \(\psi\) of
\(
p\!\left(
\vec{H},
\vec O_{-W}
\right)
\), with known full-data influence function \(\phi(\cdot; \psi, \eta) = \phi(\vec{H}, \vec O_{-W}; \eta) -\psi\).
Such targets include, for example, the mean of the latent variable(s) of interest in Figure~\ref{fig:proxies}, Figures~\ref{fig:3}(a)--(e), and Figures~\ref{fig:graphical_structure1}(a)-(b), as well as the average counterfactual outcome \(\mathbb{E}[Y(a)]\) in Figure~\ref{fig:3}(d).

Work from \citet{egami2026debiasedinferenceaigenerateddata} considers targets similar to this in settings with a single latent variable and develops influence function-based estimators when the latent variable and its proxies have finite support. Our influence function-based estimation results additionally accommodate multiple latent variables and continuous proxy variables.  

Moreover, in Section~\ref{subsec:class2}, we expand this target class to functionals of
\[
p\!\left(
\vec H,
\vec O \setminus
\bigcup_{H \in \vec H}
\{W_{1_H},W_{2_H}\}
\right), 
\] with known full-data influence function \(\phi(\cdot; \psi, \eta) = \phi_1\left(\vec{H}, \big(
\vec{O}\setminus
\bigcup_{H \in \vec H}\{W_{1_H},W_{2_H}\}\big); \eta\right) + \phi_2\left(\vec{H},
\vec{O}_{-W}; \eta\right) -\psi,\)
where \(
\mathbb{E}\left[\phi_1\left(\cdot; \eta\right) \mid \vec{H}, \vec O_{-W}\right] = 0.\)
Note that all results continue to hold if \(\{W_{1_H},W_{2_H}\}\) is replaced throughout by any pair from \(\{W_{1_H},W_{2_H},W_{3_H}\}\). 

These broader results are inspired by \citet{zhou2024causalinferencehiddentreatment}, who derived an influence function for \(\mathbb{E}[Y(a)]\) in Figure~\ref{fig:3}(b) when the latent variable \(A\) is binary. Our results provides a general procedure that applies to other target functionals of \(
p\!\left(
\vec H,
\vec O \setminus
\bigcup_{H \in \vec H}
\{W_{1_H},W_{2_H}\}
\right).
\) See Examples~\ref{ex:latent_conf_if} and~\ref{ex:latent_conf_med_if}.

The results in Section~\ref{subsec:class1} and~\ref{subsec:class2} admit adaptations for target functionals of the full-data law invariant to arbitrary permutation of the support values of
\(\vec H\), closely mirroring the adaptation of results from Section~\ref{subsec:4a} to those in Section~\ref{subsec:4b}. We omit this discussion for the sake of brevity.  

\subsection{Class 1}\label{subsec:class1}

\begin{lemma}\label{lemma:if_est1}
Suppose the full-data law \(p(\vec{H},\vec{O})\) is identified under Assumptions~\ref{assump:bounded_density}-\ref{assump:fix_labels}. Let \(\psi\) be a functional of \(
p\!\left(
\vec{H},
\vec O_{-W}
\right).\)
Given a full-data influence function for \(\psi\), \[\phi(\cdot; \psi, \eta) = \phi(\vec{H}, \vec O_{-W}; \eta) -\psi,\] 
and an observed-data function \(\varphi(\vec{O};\eta)\) satisfying 
that for each \(l \in \{1,2,3\}\)
\begin{equation}\label{eq:bridge_2}
\mathbb{E}[\varphi(\vec{O};\eta) \mid  \bigcup_{H \in \vec H} W_{l_H}, \vec{H},
\vec O_{-W}] =\phi(\vec{H}, \vec O_{-W} ; \eta), 
\end{equation}
then 
\( \varphi(\vec O; \psi, \eta) = 
\varphi(O; \eta) - \psi\)
is an observed-data influence function for \(\psi\).
\end{lemma}

We prove Lemma~\ref{lemma:if_est1} in Appendix~\ref{app:if1}. Full-data influence functions \(\phi(\cdot; \psi, \eta)  =\phi(\vec{H}, \vec O_{-W} ; \eta) - \psi\)
can often be derived via standard methods in semiparametric theory (e.g., taking Gateaux derivatives) \citep{KennedyTutorial}. To construct \(\varphi(O; \eta)\) satisfying Equation~\ref{eq:bridge_2}, we can proceed per Proposition~\ref{prop:guaranteed_solution2} and Theorem~\ref{thm:if_est1} in settings where \(\vec{H}\) has finite support, under a strengthening of Assumption~\ref{assump:distinctness}.

\begin{proposition}\label{prop:guaranteed_solution2}
Suppose Assumptions~\ref{assump:bounded_density}-\ref{assump:completeness}, Assumption~\ref{assump:fix_labels}, and a strengthening of Assumption~\ref{assump:distinctness} to a completeness condition (namely, for each \(H \in \vec{H}\), each value $\vec o_{-W}$, and any square-integrable function \(g\),
\(\mathbb{E}\{g(H) \mid W_{3_H}, \vec O_{-W}= \vec o_{-W}\} = 0 \ \text{a.s. iff } g(H) = 0 \ \text{a.s.})\) hold. If \(\vec H\) has finite support, then for each of its values \(\vec h_i\), there exist weight function
\(\lambda_{\vec h_i}\!\left( \bigcup_{H \in \vec H,k}W_{k_H},
\vec O_{-W};
\eta
\right)\), \(k \in \{1,2,3\}\), satisfying that for each \(l \in \{1,2,3\}\),
\begin{equation}\label{eq:weight_if}
\mathbb{E}\!\left[
\lambda_{\vec h_i}\!\left(
\bigcup_{H \in \vec H,k}W_{k_H},
\vec O_{-W};
\eta
\right)
\,\middle|\, \bigcup_{H \in \vec H} W_{l_H},
\vec H=\vec h_j,\vec O_{-W}
\right]
=
\mathbb{I}\{i=j\}.
\end{equation}
\end{proposition}

In particular, we may take
\(
\lambda_{\vec h_i}\!\left(
\bigcup_{H \in \vec H,k}W_{k_H},
\vec O_{-W};
\eta
\right)
=
\omega_1\omega_2+\omega_1\omega_3+\omega_2\omega_3
-2\omega_1\omega_2\omega_3,
\)
where for \(k \in \{1,2,3\}\),
\(
\omega_k
=
\omega_{\vec h_i}\!\left(
\bigcup_{H \in \vec H} W_{k_H},
\vec O_{-W};
\eta
\right)
\) is constructed as in Appendix~\ref{app:weight2}. Then
by Assumption~\ref{assump:mutual_independence}-\ref{assump:m_sep1}, for each \(l \in \{1,2,3\}\), 
 
\begin{align*}
&
\mathbb{E}\!\left[
\lambda_{\vec h_i}\!\left(
\bigcup_{H \in \vec H,k}W_{k_H},
\vec O_{-W};
\eta
\right)
\,\middle|\, \bigcup_{H \in \vec H} W_{l_H},
\vec H=\vec h_j,\vec O_{-W}
\right] =\\
& 
\omega_l \cdot \mathbb{I}\{i=j\} + \omega_l \cdot \mathbb{I}\{i=j\} + \mathbb{I}\{i=j\} -
2 \omega_l \cdot \mathbb{I}\{i=j\} = \mathbb{I}\{i=j\}.
\end{align*}
Theorem~\ref{thm:if_est1} follows from the independences in Assumptions~\ref{assump:mutual_independence}-\ref{assump:m_sep1}.

\begin{thm}\label{thm:if_est1}
Given a full-data influence function \(\phi(\cdot; \psi, \eta) = \phi(\vec{H}, \big(
\vec{O}\setminus
\bigcup_{H \in \vec H,k}W_{k_H}\big); \eta) -\psi\), define
\begin{equation}\label{eq:obs_est_eq_if}
\varphi(\vec O;\eta)
=
\sum_{\vec h_i \in \operatorname{Supp}(\vec H \mid \vec O_{-W})}
\lambda_{\vec h_i}\!\left(
\bigcup_{H \in \vec H,k}W_{k_H},
\vec O_{-W};
\eta
\right)\,
\phi\!\left(
\vec H=\vec h_i,\,
 \vec O_{-W};
\eta
\right),
\end{equation}
where weight
\(\lambda_{\vec h_i}\!\left(
\bigcup_{H \in \vec H,k}W_{k_H},
\vec O_{-W};
\eta
\right)\)
obeys Equation~\ref{eq:weight_if}. Then
\(\varphi(\vec O;\eta)\) satisfies Equation~\ref{eq:bridge_2}. Consequently,
\(\varphi(\vec O;\eta)-\psi\) is an observed-data influence function for
\(\psi\).
\end{thm}

\begin{example}\label{ex:mean_if}
Consider the target parameter \(\mu = \mathbb{E}[H]\) in the model depicted in Figure~\ref{fig:proxies}, with \(\vec H = \{H\}\) and \(\vec{O} = \{W_{1_H},W_{2_H},W_{3_H}\} = \{W,Z,Y\}\) satisfying Assumptions~\ref{assump:bounded_density}-\ref{assump:fix_labels}. Further suppose $H$ has finite support, and the strengthening of Assumption~\ref{assump:distinctness} in Proposition~\ref{prop:guaranteed_solution3} holds.

We know that
\(
\phi(H; \mu) = H - \mu 
\) is a full-data influence function for \(\mu\). If \(H\) is binary, define for each of its values \(h_i\),
\[
\omega_{h_i}(W; \eta)
=
\frac{W - \mathbb{E}[W \mid H = h_j]}{\mathbb{E}[W \mid H = h_i] - \mathbb{E}[W \mid H = h_j]}, \quad h_i \neq h_j,
\] by taking \(b_H(W)=(
 b_{H,1}(W),
 b_{H,2}(W))^\top=(1,W)^\top\) in Appendix~\ref{app:weight2}, and define \(\omega_{h_i}(Z; \eta)\) and \(\omega_{h_i}(V; \eta)\) similarly. If \(H\) has more than two categories, then define
\(\omega_{h_i}(W;\eta)\), \(\omega_{h_i}(Z;\eta)\), and \(\omega_{h_i}(V;\eta)\) by expanding the number of basis functions in \(b_{H}(W)\), \(b_{H}(Z)\), and \(b_{H}(V)\) from Appendix~\ref{app:weight2}.

Let 
\begin{align*}
&\lambda_{h_i}\!\left(
W,Z,V;\eta
\right)=\Big[\omega_{h_i}(W; \eta)\omega_{h_i}(Z; \eta) + 
\omega_{h_i}(W; \eta)\omega_{h_i}(V; \eta) + \\
& \qquad 
\omega_{h_i}(Z; \eta)\omega_{h_i}(V; \eta) -2 
\omega_{h_i}(W; \eta)\omega_{h_i}(Z; \eta)
\omega_{h_i}(V; \eta)\Big].
\end{align*}

Then 
\begin{align*}
& \varphi(W,Z,V;\eta)= 
\sum_{h_i}
\Big[\lambda_{h_i}\!\left(
W,Z,V;\eta
\right)\Big]\cdot h_i
\end{align*} is an observed-data influence function for \(\mu\) by Lemma~\ref{lemma:if_est1}.
\end{example}

Example~\ref{ex:mean_if} can be adapted for the vector-valued target parameter
\[
\mu
=
\mathbb{E}[\vec H]
=
\bigl(\mathbb{E}[H_1],\ldots,\mathbb{E}[H_m]\bigr)^\top
\] in the model depicted in Figure~\ref{fig:graphical_structure1}(a)
or Figure~\ref{fig:graphical_structure1}(b). Letting
\(\vec h_i\) denote the \(i^{\text{th}}\) support
value of \(\vec H\), and \(h_i(H)\) denote the value of \(H\) in
\(\vec h_i\), define for each \(k \in \{1,2,3\},\)
\[
\omega_{\vec h_i}\!\left(
\bigcup_{H\in\vec H} W_{k_H};\eta
\right)
=
\prod_{H\in\vec H}
\omega_{h_i(H)}\!\left(
W_{k_H};\eta
\right),
\] as per Appendix~\ref{app:weight2}. The construction then proceeds by weighting each support
value \(\vec h_i\) using

\begin{align*}
&
\lambda_{\vec h_i}\!\left(
\bigcup_{H\in\vec H,k} W_{k_H};\eta
\right)
= 
\omega_1\omega_2+\omega_1\omega_3+\omega_2\omega_3
-2\omega_1\omega_2\omega_3.
\end{align*}
where for \(k\in\{1,2,3\}\),
\(
\omega_k
=
\omega_{\vec h_i}\!\left(
\bigcup_{H \in \vec H} W_{k_H},
\vec O_{-W};
\eta
\right)\).

We provide an additional example to further illustrate the range of applications. 

\begin{example}\label{ex:latent_outcome_if}
Consider the target parameter \(\psi = \mathbb{E}[Y(a)] = \mathbb{E}[\mathbb{E}[Y \mid a, C]]\) in the model depicted in Figure~\ref{fig:3}(d), with \(\vec H = \{Y\}\) and \(\vec{O} = \{W_{1_H},W_{2_H},W_{3_H}\} = \{W,Z,V\}\) satisfying Assumptions~\ref{assump:bounded_density}-\ref{assump:fix_labels}. Further suppose $Y$ is has finite support, and the strengthening of Assumption~\ref{assump:distinctness} in Proposition~\ref{prop:guaranteed_solution3} holds.

We know that
\(
\phi(Y,A,C; \psi, \eta) = \frac{\mathbb{I}[A=a]}{p(A \mid C)}\{Y - \mathbb{E}[Y \mid a,C]\} +  \mathbb{E}[Y \mid a,C] - \psi 
\) is a full-data influence function for \(\psi\). If $Y$ is binary, define for each of its values \(y_i\),
\[
\omega_{y_i}(W,A,C; \eta)
=
\frac{W - \mathbb{E}[W \mid Y = y_j,A,C]}{\mathbb{E}[W \mid Y = y_i,A,C] - \mathbb{E}[W \mid Y = y_j,A,C]}, \quad y_i \neq y_j,
\] by taking \(b_Y(W)=(
 b_{Y,1}(W),
 b_{Y,2}(W))^\top=(1,W)^\top\) in Appendix~\ref{app:weight2}, and define \(\omega_{y_i}(Z,A,C; \eta)\) and \(\omega_{y_i}(V,A,C; \eta)\) similarly. If \(Y\) has more than two categories,
\(\omega_{y_i}(W,A,C;\eta)\), then define \(\omega_{y_i}(Z,A,C;\eta)\), and \(\omega_{y_i}(V,A,C;\eta)\) by expanding the number of basis functions in \(b_{Y}(W)\), \(b_{Y}(Z)\), and \(b_{Y}(V)\) from Appendix~\ref{app:weight2}.

Let 
\begin{align*}
&\lambda_{y_i}(W,Z,Y,A,C;\eta) = \Big[\omega_{y_i}(W,A,C; \eta)\omega_{y_i}(Z,A,C; \eta) + 
\omega_{y_i}(W,A,C; \eta)\omega_{y_i}(V,A,C; \eta) + \\
& \qquad 
\omega_{y_i}(Z,A,C; \eta)\omega_{y_i}(V,A,C; \eta) -2 
\omega_{y_i}(W,A,C; \eta)\omega_{y_i}(Z,A,C; \eta)
\omega_{y_i}(V,A,C; \eta)\Big].
\end{align*}

Then 
\begin{align*}
& \varphi(W,Z,V,A,C;\eta)= 
\sum_{y_i \in \operatorname{Supp}(Y \mid A,C)}
\Big[\lambda_{y_i}(W,Z,Y,A,C;\eta)\Big] \left\{\frac{\mathbb{I}[A=a]}{p(A \mid C)}\{y_i - \mathbb{E}[Y \mid a,C]\} +  \mathbb{E}[Y \mid a,C]\right\}
\end{align*} is an observed-data influence function for \(\psi\) by Lemma~\ref{lemma:if_est1}.
\end{example}

An influence function yielding estimators with the same asymptotic variance as those based on the influence function in Example~\ref{ex:latent_outcome_if} was previously derived by \citet{Guo2026_outcome}.

Recall that the basis functions in Appendix~\ref{app:weight2} may be chosen from a broad class of functions to satisfy Equation~\ref{eq:weight}. We next give conditions under which polynomials suffice.

\begin{proposition}\label{prop:polynomial_basis2}
Suppose the full-data law \(p(\vec{H},\vec{O})\) is identified. Further suppose $\vec{H}$ has finite support and for each \(H \in \vec H\), for \(W_{k_H} \in \{W_{1_H},W_{2_H},W_{3_H}\}\),
\[
W_{k_H}\mid H=h_i,\vec O_{-W}
\sim
\mu_{H,i}(\vec O_{-W})+\epsilon_H(\vec O_{-W}),
\]
where the error term \(\epsilon_H(\vec O_{-W})\) does not depend on the value of \(H\), and for each \(o_{-W}\)
\(
\mu_{H,i}(\vec o_{-W})
\neq
\mu_{H,j}(\vec o_{-W})
\text{ if } i\neq j,
\). Then for each \(\vec h_i\), weight \(\lambda_{\vec{h}_i}(\bigcup_{H \in \vec H,k} W_{k_H}, \vec O_{-W};\eta)\) respecting  Equation~\ref{eq:weight_if} may be constructed by taking
\[
\lambda_{\vec h_i}\!\left(
\bigcup_{H \in \vec H,k}W_{k_H},
\vec O_{-W};
\eta
\right)
=
\omega_1\omega_2+\omega_1\omega_3+\omega_2\omega_3
-2\omega_1\omega_2\omega_3,
\]
where for \(k\in\{1,2,3\}\),
\(
\omega_k
=
\omega_{\vec h_i}\!\left(
\bigcup_{H \in \vec H} W_{k_H},
\vec O_{-W};
\eta
\right),
\) is constructed by taking for each \(H \in \vec H\) and \(\vec o_{-W}\),
\[ b_H(W_{k_H})=
\bigl(
b_{H,1}(W_{k_H}),\ldots,
b_{H,K_H(\vec o_{-W})}(W_{k_H})
\bigr)^\top
=
\bigl(
1,W_{k_H},\ldots,
W_{k_H}^{K_H(\vec o_{-W})-1}
\bigr)^\top
\] in Appendix~\ref{app:weight2}.
\end{proposition}

The proof of Propostion~\ref{prop:polynomial_basis2} follows from Appendix~\ref{app:polynomial_basis1} together with Assumptions~\ref{assump:mutual_independence}-\ref{assump:m_sep1}.

Next, we characterize the multiple robustness and efficiency properties of estimators that our observed-data influence functions yield.

\begin{thm}\label{thm:multiple_robustness1}
Let \(\varphi(O;\psi, \eta) = \varphi(O;\eta) -\psi\) denote an observed-data influence function, where \(\varphi(O;\eta)\) is constructed per Equation~\ref{eq:obs_est_eq_if}, taking
\(
\lambda_{\vec h_i}\!\left(
\bigcup_{H \in \vec H,k}W_{k_H},
\vec O_{-W};
\eta
\right)
=
\omega_1\omega_2+\omega_1\omega_3+\omega_2\omega_3
-2\omega_1\omega_2\omega_3,
\)
where for \(k \in \{1,2,3\}\),
\(
\omega_k
=
\omega_{\vec h_i}\!\left(
\bigcup_{H \in \vec H} W_{k_H},
\vec O_{-W};
\eta
\right)
\)
is constructed as in Appendix~\ref{app:weight2}.

Suppose \(\hat{\psi}_{full} = \mathbb{P}_n\!\left[\phi(\cdot;\hat\eta)\right]\overset{p}{\longrightarrow} \psi\) whenever at least one of the nuisance function sets
\(
\eta_1,\ldots,\eta_n \in \eta
\)
is correctly specified. Then \(\hat{\psi}_{obs} = 
\mathbb{P}_n\!\left[\varphi(\vec O;\hat\eta)\right] \overset{p}{\longrightarrow} \psi
\) whenever at least one of the nuisance function sets
\(
\eta_1,\ldots,\eta_n \in \eta
\)
is correctly specified, and at least one of the nuisance function sets \(
\eta'_1,\eta'_2,\eta'_3 \in \eta
\) below correctly specified.
\renewcommand{\theenumi}{\roman{enumi}}%
\begin{enumerate}
  \item $\eta'_1 = \left\{\mathbb E\!\left[
 b_H(W_{1_H})
\mid H,\vec O_{-W}
\right],
\mathbb E\!\left[
 b_H(W_{2_H})
\mid H,\vec O_{-W}
\right]: H \in \vec H\right\}$;  
  \item $\eta'_2 = \left\{\mathbb E\!\left[
 b_H(W_{1_H})
\mid H,\vec O_{-W}
\right],
\mathbb E\!\left[
 b_H(W_{3_H})
\mid H,\vec O_{-W}
\right]: H \in \vec H\right\}$;  
  \item $\eta'_3 = \left\{\mathbb E\!\left[
 b_H(W_{2_H})
\mid H,\vec O_{-W}
\right],
\mathbb E\!\left[
 b_H(W_{3_H})
\mid H,\vec O_{-W}
\right]: H \in \vec H\right\}$,
\end{enumerate}
where for \(k \in \{1,2,3\}\),
\(b_H(W_{k_H})\) is defined in Appendix~\ref{app:weight2}.
\end{thm}

We provide a proof in Appendix~\ref{app:multiple_robustness1}.

\begin{thm}\label{thm:efficiency1}
Let \(\varphi(O;\psi, \eta) = \varphi(O;\eta) -\psi\) denote an observed-data influence function, where \(\varphi(O;\eta)\) is constructed per Equation~\ref{eq:obs_est_eq_if}, taking
\(
\lambda_{\vec h_i}\!\left(
\bigcup_{H \in \vec H,k}W_{k_H},
\vec O_{-W};
\eta
\right)
=
\omega_1\omega_2+\omega_1\omega_3+\omega_2\omega_3
-2\omega_1\omega_2\omega_3,
\)
where for \(k \in \{1,2,3\}\),
\(
\omega_k
=
\omega_{\vec h_i}\!\left(
\bigcup_{H \in \vec H} W_{k_H},
\vec O_{-W};
\eta
\right)
\)
is constructed as in Appendix~\ref{app:weight2}.

Suppose \(\mathbb{E}[\phi(\cdot; \hat \eta) - \psi] = o_p(n^{-1/2})\). Then under suitable regularity conditions including sample splitting and consistency of nuisance function estimation, the estimator \(\hat{\psi}_{obs}\) obtained by solving the empirical estimating equation
\(
\mathbb{P}_n\!\left[\varphi(\vec O;\hat\psi_{obs},\hat\eta)\right]=0
\) using the observed data has the property that
\[
\sqrt{n}\,(\hat{\psi}_{obs} - \psi)
\;\;\overset{d}{\longrightarrow}\;\;
\mathcal{N}\!\Big(0, \,\mathbb{E}\big[\varphi^2(O; \psi, \eta)\big]\Big)
\]
if conditions (i)-(iii) holds for each \(H \in \vec H\):
\renewcommand{\theenumi}{\roman{enumi}}%
\begin{enumerate}
  \item $\left\lVert
\mathbb{\hat E}\!\left[
 b_H(W_{1_H})
\mid H,\vec O_{-W}
\right]-\mathbb E\!\left[
 b_H(W_{1_H})
\mid H,\vec O_{-W}
\right]
\right\rVert$ 
  $\left\lVert
\mathbb{\hat E}\!\left[
 b_H(W_{2_H})
\mid H,\vec O_{-W}
\right]-\mathbb E\!\left[
 b_H(W_{2_H})
\mid H,\vec O_{-W}
\right]
\right\rVert = o_p(n^{-1/2})$; 
   
   \item $\left\lVert
\mathbb{\hat E}\!\left[
 b_H(W_{1_H})
\mid H,\vec O_{-W}
\right]-\mathbb E\!\left[
 b_H(W_{1_H})
\mid H,\vec O_{-W}
\right]
\right\rVert$ 
  $\left\lVert
\mathbb{\hat E}\!\left[
 b_H(W_{3_H})
\mid H,\vec O_{-W}
\right]-\mathbb E\!\left[
 b_H(W_{3_H})
\mid H,\vec O_{-W}
\right]
\right\rVert = o_p(n^{-1/2})$;   
\item $\left\lVert
\mathbb{\hat E}\!\left[
 b_H(W_{2_H})
\mid H,\vec O_{-W}
\right]-\mathbb E\!\left[
 b_H(W_{2_H})
\mid H,\vec O_{-W}
\right]
\right\rVert$ 
  $\left\lVert
\mathbb{\hat E}\!\left[
 b_H(W_{3_H})
\mid H,\vec O_{-W}
\right]-\mathbb E\!\left[
 b_H(W_{3_H})
\mid H,\vec O_{-W}
\right]
\right\rVert = o_p(n^{-1/2})$;,
\end{enumerate}
where for \(k \in \{1,2,3\}\),
\(b_H(W_{k_H})\) is defined in Appendix~\ref{app:weight2}.
\end{thm}

We provide a proof in Appendix~\ref{app:efficiency1}.

\begin{remark}
Theorem~\ref{thm:efficiency1}, as well as Theorem~\ref{thm:efficiency2} in Section~\ref{subsec:class2}, provide product-rate conditions on nuisance function estimation sufficient for constructing $\sqrt{n}$-consistent estimators, but does not prescribe specific procedures that achieve these rates. Under additional model restrictions, standard procedures may satisfy the required conditions; for example, the means of a Gaussian mixture may be estimated using the expectation-maximization algorithm at suitable rates \citep{balakrishnan2017, wang2015}. More generally, existing work has employed sieve-based approaches to estimating relevant nuisance functions \citep{hu2008instrumental, zhou2024causalinferencehiddentreatment}. Establishing a rich toolkit of procedures which satisfy the product-rate conditions stated remains an important direction for future work. 
\end{remark}

\begin{example}
By Theorem~\ref{thm:multiple_robustness1}, if \(H\) is binary, the estimator
\(
\hat{\mu}
=
\mathbb{P}_n\!\left[
\varphi(W,Z,V;\hat\eta)
\right]
\)
in Example~\ref{ex:mean_if} is consistent provided at least one of the following nuisance pairs is consistently estimated:
\begin{enumerate}
\item
\(\{\widehat{\mathbb E}[W\mid H],\,\widehat{\mathbb E}[Z\mid H]\}\);
\item
\(\{\widehat{\mathbb E}[W\mid H],\,\widehat{\mathbb E}[V\mid H]\}\);
\item
\(\{\widehat{\mathbb E}[Z\mid H],\,\widehat{\mathbb E}[V\mid H]\}\).
\end{enumerate}

Furthermore, by Theorem~\ref{thm:efficiency1},
\(
\sqrt{n}(\hat\mu-\mu)
\rightsquigarrow
N\!\left(
0,\,
\mathbb E\!\left[
\{\varphi(W,Z,V;\eta)-\mu\}^2
\right]
\right),
\)
under suitable regularity conditions including sample splitting and consistency of nuisance functions, provided the following conditions hold:
\begin{enumerate}
\item
\(
\left\|
\widehat{\mathbb E}[W\mid H]
-
\mathbb E[W\mid H]
\right\|
\left\|
\widehat{\mathbb E}[Z\mid H]
-
\mathbb E[Z\mid H]
\right\|
=
o_p(n^{-1/2});
\)

\item
\(
\left\|
\widehat{\mathbb E}[W\mid H]
-
\mathbb E[W\mid H]
\right\|
\left\|
\widehat{\mathbb E}[V\mid H]
-
\mathbb E[V\mid H]
\right\|
=
o_p(n^{-1/2});
\)

\item
\(
\left\|
\widehat{\mathbb E}[Z\mid H]
-
\mathbb E[Z\mid H]
\right\|
\left\|
\widehat{\mathbb E}[V\mid H]
-
\mathbb E[V\mid H]
\right\|
=
o_p(n^{-1/2}).
\)
\end{enumerate}
\end{example}

\begin{example}\label{ex:latent_outcome_robust}
By Theorem~\ref{thm:multiple_robustness1}, if \(Y\) is binary, the estimator
\(
\hat{\psi}
=
\mathbb{P}_n\!\left[
\varphi(W,Z,V,A,C;\hat\eta)
\right]
\)
in Example~\ref{ex:latent_outcome_if} is consistent provided \(p(A \mid C)\) or \(\mathbb{E}[Y \mid A,C]\) is consistently estimated, and at least one of the following nuisance pairs is consistently estimated:
\begin{enumerate}
\item
\(\{\widehat{\mathbb E}[W\mid  Y,A,C],\,\widehat{\mathbb E}[Z\mid  Y,A,C]\}\);
\item
\(\{\widehat{\mathbb E}[W\mid  Y,A,C],\,\widehat{\mathbb E}[V\mid  Y,A,C]\}\);
\item
\(\{\widehat{\mathbb E}[Z\mid  Y,A,C],\,\widehat{\mathbb E}[V\mid  Y,A,C]\}\).
\end{enumerate}

Furthermore, by Theorem~\ref{thm:efficiency1},
\(
\sqrt{n}(\hat\psi-\psi)
\rightsquigarrow
N\!\left(
0,\,
\mathbb E\!\left[
\{\varphi(W,Z,V,A,C;\eta)-\psi\}^2
\right]
\right),
\)
under suitable regularity conditions including sample splitting and consistency of nuisance functions, provided the following conditions hold:
\begin{enumerate}
\item
\(
\left\|
\widehat{p}[A\mid C]
-
p[A\mid C]
\right\|
\left\|
\widehat{\mathbb E}[Y\mid A,C]
-
\mathbb E[Y\mid A,C]
\right\|
=
o_p(n^{-1/2});
\)

\item
\(
\left\|
\widehat{\mathbb E}[W\mid  Y,A,C]
-
\mathbb E[W\mid  Y,A,C]
\right\|
\left\|
\widehat{\mathbb E}[Z\mid  Y,A,C]
-
\mathbb E[Z\mid  Y,A,C]
\right\|
=
o_p(n^{-1/2});
\)

\item
\(
\left\|
\widehat{\mathbb E}[W\mid  Y,A,C]
-
\mathbb E[W\mid  Y,A,C]
\right\|
\left\|
\widehat{\mathbb E}[V\mid  Y,A,C]
-
\mathbb E[V\mid  Y,A,C]
\right\|
=
o_p(n^{-1/2});
\)

\item
\(
\left\|
\widehat{\mathbb E}[Z\mid  Y,A,C]
-
\mathbb E[Z\mid Y,A,C]
\right\|
\left\|
\widehat{\mathbb E}[V\mid  Y,A,C]
-
\mathbb E[V\mid  Y,A,C]
\right\|
=
o_p(n^{-1/2}).
\)
\end{enumerate}
\end{example}

\subsection{Class 2}\label{subsec:class2}
In this section, we expand the class of target parameters considered to those which can be expressed as functionals of
\[
p\!\left(
\vec H,
\vec O \setminus
\bigcup_{H \in \vec H}
\{W_{1_H},W_{2_H}\}
\right),
\] with known full-data influence function \(\phi(\cdot; \psi, \eta) = \phi_1\left(\vec{H}, \big(
\vec{O}\setminus
\bigcup_{H \in \vec H}\{W_{1_H},W_{2_H}\}\big); \eta\right) + \phi_2\left(\vec{H},
\vec{O}_{-W}; \eta\right) -\psi,\)
where \(
\mathbb{E}\left[\phi_1\left(\cdot; \eta\right) \mid \vec{H}, \vec O_{-W}\right] = 0.\)
We can replace \(\{W_{1_H},W_{2_H}\}\) in the results that follow with any pair from \(\{W_{1_H},W_{2_H},W_{3_H}\}\).

\begin{lemma}\label{lemma:if_est2}
Suppose the full-data law \(p(\vec{H},\vec{O})\) is identified under Assumptions~\ref{assump:bounded_density}-\ref{assump:fix_labels}. Let \(\psi\) be a functional of \(
p\!\left(
\vec H,
\vec O \setminus
\bigcup_{H \in \vec H}
\{W_{1_H},W_{2_H}\}
\right)\).
Given a full-data influence function for \(\psi\), \[\phi(\vec H,
\vec O \setminus
\bigcup_{H \in \vec H}
\{W_{1_H},W_{2_H}\} ; \psi, \eta) = \phi_1\left(\vec{H}, \big(
\vec{O}\setminus
\bigcup_{H \in \vec H}\{W_{1_H},W_{2_H}\}\big); \eta\right) + \phi_2\left(\vec{H},
\vec{O}_{-W}; \eta\right) -\psi,\]
where \(
\mathbb{E}\left[\phi_1\left(\cdot; \eta\right) \mid \vec{H}, \vec O_{-W}\right] = 0,\) and an observed-data function \(\varphi(\vec{O};\eta)\)  satisfying
\begin{equation}\label{eq:bridge_3a}
\mathbb{E}[\varphi(\vec{O};\eta) \mid  \vec H,
\vec O \setminus
\bigcup_{H \in \vec H}
\{W_{1_H},W_{2_H}\}] =\phi_1\left(\vec{H}, \big(
\vec{O}\setminus
\bigcup_{H \in \vec H}\{W_{1_H},W_{2_H}\}\big); \eta\right) + \phi_2\left(\vec{H},
\vec{O}_{-W}; \eta\right), 
\end{equation}
and that for each \(m \in \{1,2\}\),
\begin{equation}\label{eq:bridge_3b}
\mathbb{E}\left[
\varphi(\vec{O};\eta)
\mid
\bigcup_{H \in \vec H} W_{m_H}, \vec H,
\vec O_{-W} 
\right] =
\phi_2\left(\vec{H},
\vec{O}_{-W}; \eta\right),
\end{equation}
then 
\( \varphi(\vec O; \psi, \eta) = 
\varphi(O; \eta) - \psi\)
is an observed-data influence function for \(\psi\).
\end{lemma}

We prove Lemma~\ref{lemma:if_est2} in Appendix~\ref{app:if2}.

\begin{proposition}\label{prop:guaranteed_solution3}
Suppose Assumptions~\ref{assump:bounded_density}-\ref{assump:completeness}, Assumption~\ref{assump:fix_labels}, and a strengthening of Assumption~\ref{assump:distinctness} to a completeness condition (namely, for each \(H \in \vec{H}\), each value $\vec o_{-W}$, and any square-integrable function \(g\),
\(\mathbb{E}\{g(H) \mid W_{3_H}, \vec O_{-W}= \vec o_{-W}\} = 0 \ \text{a.s. iff } g(H) = 0 \ \text{a.s.})\) hold. Suppose further that \(\vec H\) has finite support. Then for each of its values \(\vec h_i\), there exist weight function
\(\lambda_{\vec h_i}\!\left(
\bigcup_{H \in \vec H,k}W_{k_H},
\vec O_{-W};
\eta
\right)\), \(k \in \{1,2,3\}\), satisfying that for each \(l \in \{1,2,3\}\),

\begin{equation}\label{eq:weight_if2}
\mathbb{E}\!\left[
\lambda_{\vec h_i}\!\left(
\bigcup_{H \in \vec H,k}W_{k_H},
\vec O_{-W};
\eta
\right)
\,\middle|\, \bigcup_{H \in \vec H} W_{l_H},
\vec H=\vec h_j,\vec O_{-W}
\right]
=
\mathbb{I}\{i=j\},
\end{equation}
and there exist weight function
\(\gamma_{\vec{h}_i}\left(\bigcup_{H \in \vec H}\{W_{1_H},W_{2_H}\}, \vec O_{-W}; \eta\right)\), satisfying both 
\begin{equation}\label{eq:weight_if1a}
\mathbb{E}\!\left[\gamma_{\vec{h}_i}\left(\bigcup_{H \in \vec H} \{W_{1_H},W_{2_H}\}, \vec O_{-W};\eta\right)\mid \vec{H}=\vec{h}_j,\vec O \setminus \bigcup_{H \in \vec H}  \{W_{1_H},W_{2_H}\}\right]
=
\mathbb{I}\{i=j\}
\end{equation}
and that for each \(m \in \{1,2\}\),
\begin{equation}\label{eq:weight_if1b}
\begin{aligned}
&\mathbb{E}\!\left[\gamma_{\vec{h}_i}\left(\bigcup_{H \in \vec H} \{W_{1_H},W_{2_H}\}, \vec O_{-W};\eta\right)\mid \bigcup_{H \in \vec H} W_{m_H}= \bigcup_{H \in \vec H} w_{m_H}, \vec{H}=\vec{h}_j,\vec O_{-W} = \vec o_{-W}\right]
=
c \cdot \mathbb{I}\{i=j\}, \\
& \quad \text{ where } c \text{ is a constant.}
\end{aligned}
\end{equation}
\end{proposition}

Similar to Section~\ref{subsec:4a}, we may take
\(
\lambda_{\vec h_i}\!\left(
\bigcup_{H \in \vec H,k}W_{k_H},
\vec O_{-W};
\eta
\right)
=
\omega_1\omega_2+\omega_1\omega_3+\omega_2\omega_3
-2\omega_1\omega_2\omega_3,
\)
where for \(k \in \{1,2,3\}\),
\(
\omega_k
=
\omega_{\vec h_i}\!\left(
\bigcup_{H \in \vec H} W_{k_H},
\vec O_{-W};
\eta
\right)
\) is constructed as in Appendix~\ref{app:weight2}. Additionally, we may take
\(
\gamma_{\vec h_i}\!\left(
\bigcup_{H \in \vec H}\{W_{1_H},W_{2_H}\},
\vec O_{-W};
\eta
\right)
=
\omega_1\omega_2.
\) 

\begin{thm}\label{thm:if_est2}
Given a full-data influence function \[\phi(\vec H,
\vec O \setminus
\bigcup_{H \in \vec H}
\{W_{1_H},W_{2_H}\} ; \psi, \eta) = \phi_1\left(\vec{H}, \big(
\vec{O}\setminus
\bigcup_{H \in \vec H}\{W_{1_H},W_{2_H}\}\big); \eta\right) + \phi_2\left(\vec{H},
\vec{O}_{-W}; \eta\right) -\psi,\]

where \(
\mathbb{E}\left[\phi_1\left(\cdot; \eta\right) \mid \vec{H}, \vec O_{-W}\right] = 0,\)
define
\begin{equation}\label{eq:obs_est_eq_if2}
\begin{aligned}
\varphi(\vec O;\eta)
&=
\sum_{\vec h_i \in \operatorname{Supp}(\vec H \mid \vec O_{-W})}
\gamma_{\vec h_i}\!\left(
\bigcup_{H \in \vec H}\{W_{1_H},W_{2_H}\},
\vec O_{-W};
\eta
\right)\,
\phi_1\!\left(
\vec h_i, \big(
\vec{O}\setminus
\bigcup_{H \in \vec H}\{W_{1_H},W_{2_H}\}\big);
\eta
\right)  \\
& \quad
+ 
\sum_{\vec h_i \in \operatorname{Supp}(\vec H \mid \vec O_{-W})}
\lambda_{\vec h_i}\!\left(
\bigcup_{H \in \vec H,k}W_{k_H},
\vec O_{-W};
\eta
\right)\,
\phi_2\!\left( \vec h_i,
\vec O_{-W};
\eta
\right).
\end{aligned}
\end{equation}

If the weights
\(\lambda_{\vec h_i}\!\left(
\bigcup_{H \in \vec H,k}W_{k_H},
\vec O_{-W};
\eta
\right)\) and \(\gamma_{\vec h_i}\!\left(
\bigcup_{H \in \vec H}\{W_{1_H},W_{2_H}\},
\vec O_{-W};
\eta
\right)\)
obey Equations~\ref{eq:weight_if2}-\ref{eq:weight_if1b}, then
\(\varphi(\vec O;\eta)\) satisfies Equations~\ref{eq:bridge_3a} and~\ref{eq:bridge_3b}. Hence,
\(\varphi(\vec O;\eta)-\psi\) is an observed-data influence function for
\(\psi\).
\end{thm}

To see that Theorem~\ref{thm:if_est2} holds, note that by Assumptions~\ref{assump:mutual_independence}-\ref{assump:m_sep1}, 
\begin{align*}
&\mathbb{E}[\varphi(\vec O; \eta) \mid \vec H, \vec O \setminus \bigcup_{H \in \vec H}\{W_{1_H},W_{2_H}\}] = \phi_1\left(\cdot; \eta\right) + \phi_2\left(\cdot; \eta\right),
\end{align*}
satisfying Equation~\ref{eq:bridge_3a}.

Additionally by Assumptions~\ref{assump:mutual_independence}-\ref{assump:m_sep1}, for each \(m \in \{1,2\}\),

\begin{align*}
&\mathbb{E}[\varphi(\vec O; \eta) \mid \bigcup_{H \in \vec H} W_{m_H}, \vec H, \vec O_{-W}] =
\omega_m \cdot  \mathbb{E}[\phi_1\left(\cdot; \eta\right) \mid \vec H, \vec O_{-W}] + \mathbb{E}[\phi_2\left(\cdot; \eta\right)\mid \vec H, \vec O_{-W}] 
=  \phi_2(\cdot; \eta),
\end{align*}
satisfying Equation~\ref{eq:bridge_3b}.

\begin{example}\label{ex:latent_conf_if}
Consider the target parameter \(\psi = \mathbb{E}[Y(a)] =\mathbb{E}[\mathbb{E}[Y \mid a, U]]\) in the model depicted in Figure~\ref{fig:3}(a), with node partition in its description satisfying Assumptions~\ref{assump:bounded_density}-\ref{assump:fix_labels}. Suppose \(U\) has finite support and the strengthening of Assumption~\ref{assump:distinctness} in Proposition~\ref{prop:guaranteed_solution3} holds.

We know that
\[
\phi(Y,A,U; \psi, \eta) =
\underbrace{
\frac{\mathbb{I}[A=a]}{p(A \mid U)}\{Y - \mathbb{E}[Y \mid a,U]\}
}_{\phi_1(Y,A,U)}
+
\underbrace{
\mathbb{E}[Y \mid a,U] 
}_{\phi_2(U)} - \psi, 
\] 
is a full-data influence function for \(\psi\). If $U$ is binary, define for each of its values \(u_i\),
\[
\omega_{u_i}(W,A; \eta)
=
\frac{W - \mathbb{E}[W \mid U = u_j,A]}{\mathbb{E}[W \mid U = u_i,A] - \mathbb{E}[W \mid U = u_j,A]}, \quad u_i \neq u_j,
\] by taking \(b_U(W)=(
 b_{U,1}(W),
 b_{U,2}(W))^\top=(1,W)^\top\) in Appendix~\ref{app:weight2}, and define \(\omega_{u_i}(Z,A; \eta)\) and \(\omega_{u_i}(Y,A; \eta)\) similarly. If \(U\) has more than two categories,
then define \(\omega_{u_i}(W,A;\eta)\), \(\omega_{u_i}(Z,A;\eta)\), and \(\omega_{u_i}(Y,A;\eta)\) by expanding the number of basis functions in \(b_{U}(W)\), \(b_{U}(Z)\), and \(b_{U}(V)\) from Appendix~\ref{app:weight2}.

Let
\begin{align*}
& \lambda_{u_i}(W,Z,Y,A; \eta)
=
\Big[\omega_{u_i}(W,A; \eta)\omega_{u_i}(Z,A; \eta) + 
\omega_{u_i}(W,A; \eta)\omega_{u_i}(Y,A; \eta) + \\
& \qquad 
\omega_{u_i}(Z,A; \eta)\omega_{u_i}(Y,A; \eta) -2 
\omega_{u_i}(W,A; \eta)\omega_{u_i}(Z,A; \eta)
\omega_{u_i}(Y,A; \eta)\Big]
\end{align*}
and
\begin{align*}
& \gamma_{u_i}(W,Z,A; \eta)
=
\omega_{u_i}(W,A; \eta)\omega_{u_i}(Z,A; \eta).
\end{align*}

Then 
\begin{align*}
& \varphi(W,Z,Y,A;\eta)= 
\sum_{u_i \in \operatorname{Supp}(U \mid A)}
\Big[\gamma_{u_i}(W,Z,A; \eta) \Big]\cdot 
\phi_1(Y,A,u_i) +
\sum_{u_i \in \operatorname{Supp}(U \mid A)}
\Big[\lambda_{u_i}(W,Z,Y,A; \eta)\Big]\cdot 
\phi_2(u_i)
\end{align*} is an observed-data influence function for \(\psi\) by Lemma~\ref{lemma:if_est2}.
\end{example}

\begin{example}\label{ex:latent_conf_med_if}
Consider the target parameter \[\psi = \mathbb{E}[Y(a)] = \iiiint y
p(y\mid a',m,u)\,
p(m\mid a,u)\,
p(a',u)\,
dy\,dm\,da'\,du\] in the model depicted in Figure~\ref{fig:3}(e) with binary variable \(A\) and node partition in its description satisfying Assumptions~\ref{assump:bounded_density}-\ref{assump:fix_labels}. Suppose \(U\) has finite support and the strengthening of Assumption~\ref{assump:distinctness} in Proposition~\ref{prop:guaranteed_solution3} holds.

We know that
\[
\phi(Y,M,A,U; \psi, \eta) =
\underbrace{
\frac{p(M \mid A=a, U)}{p(M \mid A, U)}
\Big\{
  Y - \mu(M, A, U)
\Big\}
}_{\phi_1(Y,M,A,U)}
+
\underbrace{
\frac{\mathbb{I}(A = a)}{\pi(A \mid U)}
\Big\{
  \xi(M, U) - \theta(U)
\Big\}
+\zeta(A,U) 
}_{\phi_2(M,A,U)} - \psi, 
\] 
where
\[
\mu(M,A,U)
:= \mathbb{E}[Y \mid M, A, U],
\hspace{0.2cm}
\pi(A \mid U)
:= p(A \mid U),
\hspace{0.2cm}
\xi(M, U)
:= \sum_{a'=0}^1 \mu(M, a', U)\,\pi(a' \mid U),
\]
\[
\zeta(A,U)
:= \int \mu(m,A,U)\, p(m \mid a,U)\, dm,
\qquad
\theta(U)
:= \int \xi(m,U)\, p(m \mid a,U)\, dm.
\]
is a full-data influence function for \(\psi\) \citep{fulcher19robust}. If $U$ is binary, define for each of its values \(u_i\),
\[
\omega_{u_i}(W,A,M; \eta)
=
\frac{W - \mathbb{E}[W \mid U = u_j,A,M]}{\mathbb{E}[W \mid U = u_i,A,M] - \mathbb{E}[W \mid U = u_j,A,M]}, \quad u_i \neq u_j,
\] by taking \(b_U(W)=(
 b_{U,1}(W),
 b_{U,2}(W))^\top=(1,W)^\top\) in Appendix~\ref{app:weight2}, and define \(\omega_{u_i}(Z,A,M; \eta)\) and \(\omega_{u_i}(Y,A,M; \eta)\) similarly. If \(U\) has more than two categories, define
\(\omega_{u_i}(W,A,M;\eta)\), \(\omega_{u_i}(Z,A,M;\eta)\), and \(\omega_{u_i}(Y,A,M;\eta)\) by expanding the number of basis functions in \(b_{U}(W)\), \(b_{U}(Z)\), and \(b_{U}(V)\) from Appendix~\ref{app:weight2}.

Let 
\begin{align*}
&\lambda_{u_i}(W,Z,Y,A,M; \eta) =\omega_{u_i}(W,A,M; \eta)\omega_{u_i}(Z,A,M; \eta) + 
\omega_{u_i}(W,A,M; \eta)\omega_{u_i}(Y,A,M; \eta) + \\
& \qquad 
\omega_{u_i}(Z,A,M; \eta)\omega_{u_i}(Y,A,M; \eta) -2 
\omega_{u_i}(W,A,M; \eta)\omega_{u_i}(Z,A,M; \eta)
\omega_{u_i}(Y,A,M; \eta)
\end{align*}
and
\begin{align*}
&\gamma_{u_i}(W,Z,A,M; \eta) =\omega_{u_i}(W,A,M; \eta)\omega_{u_i}(Z,A,M; \eta).
\end{align*}

Then
\begin{align*}
& \varphi(W,Z,Y,A,M;\eta)= 
\sum_{u_i \in \operatorname{Supp}(U \mid A,M)}
\Big[\gamma_{u_i}(W,Z,A,M; \eta) \Big]\cdot 
\phi_1(Y,M,A,u_i) \\
& +
\sum_{u_i \in \operatorname{Supp}(U \mid A,M)}
\Big[\lambda_{u_i}(W,Z,Y,A,M; \eta)\Big]\cdot 
\phi_2(M,A,u_i)
\end{align*} is an observed-data influence function for \(\psi\) by Lemma~\ref{lemma:if_est2}.
\end{example}

The influence function in Example~\ref{ex:latent_conf_med_if} has previously been derived in \citet{Guo2026_confounder}.

Recall that the basis functions in Appendix~\ref{app:weight2} may be chosen from a broad class of functions to satisfy Equation~\ref{eq:weight}. We next give conditions under which polynomials suffice.

\begin{proposition}\label{prop:polynomial_basis3}
Suppose the full-data law \(p(\vec{H},\vec{O})\) is identified.  Further suppose $\vec{H}$ has finite support and for each \(H \in \vec H\), for \(W_{k_H} \in \{W_{1_H},W_{2_H},W_{3_H}\}\),
\[
W_{k_H}\mid H=h_i,\vec O_{-W}
\sim
\mu_{H,i}(\vec O_{-W})+\epsilon_H(\vec O_{-W}),
\]
where the error term \(\epsilon_H(\vec O_{-W})\) does not depend on the value of \(H\), and for each \(o_{-W}\),
\(
\mu_{H,i}(\vec o_{-W})
\neq
\mu_{H,j}(\vec o_{-W})
\text{ if } i\neq j,
\). Then for each \(\vec h_i\), weight \(\lambda_{\vec{h}_i}(\bigcup_{H \in \vec H,k} W_{k_H}, \vec O_{-W};\eta)\) respecting  Equation~\ref{eq:weight_if2}, and weight \(\gamma_{\vec{h}_i}(\bigcup_{H \in \vec H} \{W_{1_H},W_{2_H}\}, \vec O_{-W};\eta)\) respecting  Equation~\ref{eq:weight_if1a}-\ref{eq:weight_if1b} may be constructed by taking
\[
\lambda_{\vec h_i}\!\left(
\bigcup_{H \in \vec H,k}W_{k_H},
\vec O_{-W};
\eta
\right)
=
\omega_1\omega_2+\omega_1\omega_3+\omega_2\omega_3
-2\omega_1\omega_2\omega_3, \text{ and}
\]
\[
\gamma_{\vec h_i}\!\left(
\bigcup_{H \in \vec H}\{W_{1_H},W_{2_H}\},
\vec O_{-W};
\eta
\right)
=
\omega_1\omega_2,
\]
where for \(k\in\{1,2,3\}\),
\(
\omega_k
=
\omega_{\vec h_i}\!\left(
\bigcup_{H \in \vec H} W_{k_H},
\vec O_{-W};
\eta
\right)
\) is constructed by taking for each \(H \in \vec H\) and \(\vec o_{-W}\),
\[b_H(W_{k_H})=
\bigl(
b_{H,1}(W_{k_H}),\ldots,
b_{H,K_H(\vec o_{-W})}(W_{k_H})
\bigr)^\top
=
\bigl(
1,W_{k_H},\ldots,
W_{k_H}^{K_H(\vec o_{-W})-1}
\bigr)^\top
\] in Appendix~\ref{app:weight2}.
\end{proposition}
The proof of Proposition~\ref{prop:polynomial_basis3} follows from Appendix~\ref{app:polynomial_basis1} together with Assumptions~\ref{assump:mutual_independence}-\ref{assump:m_sep1}.

Next, we characterize the multiple robustness and efficiency properties of estimators that our observed-data influence functions yield.

\begin{thm}\label{thm:multiple_robustness2}
Let \(\varphi(O;\psi, \eta) = \varphi(O;\eta) -\psi\) denote an observed-data influence function, where \(\varphi(O;\eta)\) is constructed per Equation~\ref{eq:obs_est_eq_if2}, by taking
\(
\lambda_{\vec h_i}\!\left(
\bigcup_{H \in \vec H,k}W_{k_H},
\vec O_{-W};
\eta
\right)
=
\omega_1\omega_2+\omega_1\omega_3+\omega_2\omega_3
-2\omega_1\omega_2\omega_3
\) and 
\(
\gamma_{\vec h_i}\!\left(
\bigcup_{H \in \vec H}\{W_{1_H},W_{2_H}\},
\vec O_{-W};
\eta
\right)
=
\omega_1\omega_2,
\)
where for \(k \in \{1,2,3\}\),
\(
\omega_k
=
\omega_{\vec h_i}\!\left(
\bigcup_{H \in \vec H} W_{k_H},
\vec O_{-W};
\eta
\right)
\)
is constructed as in Appendix~\ref{app:weight2}.

Suppose \(\hat{\psi}_{full} = \mathbb{P}_n\!\left[\phi(\cdot;\hat\eta)\right]\overset{p}{\longrightarrow} \psi\) whenever at least one of the nuisance function sets
\(
\eta_1,\ldots,\eta_n \in \eta
\)
is correctly specified, and \( \mathbb{P}_n\!\left[\phi_1(\cdot;\hat\eta)\right]\overset{p}{\longrightarrow} 0\) whenever \(\{\mathbb{E}[b_H({W_{3_H})}\mid H, \vec O_{-W}]: H \in \vec H\}\)
is correctly specified. Then \(\hat{\psi}_{obs} = 
\mathbb{P}_n\!\left[\varphi(\vec O;\hat\eta)\right]\overset{p}{\longrightarrow} \psi
\)
is consistent whenever at least one of the nuisance function sets
\(
\eta_1,\ldots,\eta_n \in \eta
\)
is correctly specified, and at least one of the nuisance function sets \(
\eta'_1,\eta'_2,\eta'_3 \in \eta
\) below correctly specified.
\renewcommand{\theenumi}{\roman{enumi}}%
\begin{enumerate}
  \item $\eta'_1 = \left\{\mathbb E\!\left[
 b_H(W_{1_H})
\mid H,\vec O_{-W}
\right],
\mathbb E\!\left[
 b_H(W_{2_H})
\mid H,\vec O_{-W}
\right]: H \in \vec H\right\}$;  
  \item $\eta'_2 = \left\{\mathbb E\!\left[
 b_H(W_{1_H})
\mid H,\vec O_{-W}
\right],
\mathbb E\!\left[
 b_H(W_{3_H})
\mid H,\vec O_{-W}
\right]: H \in \vec H \right\}$;  
  \item $\eta'_3 = \left\{\mathbb E\!\left[
 b_H(W_{2_H})
\mid H,\vec O_{-W}
\right],
\mathbb E\!\left[
 b_H(W_{3_H})
\mid H,\vec O_{-W}
\right]: H \in \vec H\right\}$,
\end{enumerate}
where for \(k \in \{1,2,3\}\),
\(b_H(W_{k_H})\) is defined in Appendix~\ref{app:weight2}.
\end{thm}

We provide a proof in Appendix~\ref{app:multiple_robustness2}.

\begin{thm}\label{thm:efficiency2}
Let \(\varphi(O;\psi, \eta) = \varphi(O;\eta) -\psi\) denote an observed-data influence function, where \(\varphi(O;\eta)\) is constructed per Equation~\ref{eq:obs_est_eq_if2}, by taking
\(
\lambda_{\vec h_i}\!\left(
\bigcup_{H \in \vec H,k}W_{k_H},
\vec O_{-W};
\eta
\right)
=
\omega_1\omega_2+\omega_1\omega_3+\omega_2\omega_3
-2\omega_1\omega_2\omega_3
\) and 
\(
\gamma_{\vec h_i}\!\left(
\bigcup_{H \in \vec H}\{W_{1_H},W_{2_H}\},
\vec O_{-W};
\eta
\right)
=
\omega_1\omega_2,
\)
where for \(k \in \{1,2,3\}\),
\(
\omega_k
=
\omega_{\vec h_i}\!\left(
\bigcup_{H \in \vec H} W_{k_H},
\vec O_{-W};
\eta
\right)
\)
is constructed as in Appendix~\ref{app:weight2}.

Suppose \(\mathbb{E}[\phi(\cdot; \hat \eta) - \psi] = o_p(n^{-1/2})\), in addition to
\begin{align*}
\phi_1(\cdot; \eta)
={}&
g\left(\vec H,\vec O_{-W};\eta\right)
\Bigg[
t_{\vec O_{-W}}\!\left(
\left\{
b_{H,q}(W_{3_H})
:
H\in\vec H,\;
q=1,\ldots,K_H(\vec O_{-W})
\right\}
\right) \\
&
-
\mathbb E\!\left[
t_{\vec O_{-W}}\!\left(
\left\{
b_{H,q}(W_{3_H})
:
H\in\vec H,\;
q=1,\ldots,K_H(\vec O_{-W})
\right\}
\right)
\mid \vec H,\vec O_{-W}
\right]
\Bigg],
\end{align*}
where
\(t_{\vec O_{-W}}\) is a linear function of the components
\(
\left\{
b_{H,q}(W_{3_H})
:
H\in\vec H,\;
q=1,\ldots,K_H(\vec O_{-W})
\right\}
\) from Appendix~\ref{app:weight2}.

Then under suitable regularity conditions including sample splitting and consistency of nuisance functions, the estimator \(\hat{\psi}_{obs}\) obtained by solving the empirical estimating equation
\(
\mathbb{P}_n\!\left[\varphi(\vec O;\hat\psi_{obs},\hat\eta)\right]=0
\) using the observed data has the property that
\[
\sqrt{n}\,(\hat{\psi}_{obs} - \psi)
\;\;\overset{d}{\longrightarrow}\;\;
\mathcal{N}\!\Big(0, \,\mathbb{E}\big[\varphi^2(O; \psi, \eta)\big]\Big)
\]
if conditions (i)-(iii) hold for each \(H \in \vec H\):
\renewcommand{\theenumi}{\roman{enumi}}%
\begin{enumerate}
  \item $\left\lVert
\mathbb{\hat E}\!\left[
 b_H(W_{1_H})
\mid H,\vec O_{-W}
\right]-\mathbb E\!\left[
 b_H(W_{1_H})
\mid H,\vec O_{-W}
\right]
\right\rVert$ 
  $\left\lVert
\mathbb{\hat E}\!\left[
 b_H(W_{2_H})
\mid H,\vec O_{-W}
\right]-\mathbb E\!\left[
 b_H(W_{2_H})
\mid H,\vec O_{-W}
\right]
\right\rVert = o_p(n^{-1/2})$; 
   
   \item $\left\lVert
\mathbb{\hat E}\!\left[
 b_H(W_{1_H})
\mid H,\vec O_{-W}
\right]-\mathbb E\!\left[
 b_H(W_{1_H})
\mid H,\vec O_{-W}
\right]
\right\rVert$ 
  $\left\lVert
\mathbb{\hat E}\!\left[
 b_H(W_{3_H})
\mid H,\vec O_{-W}
\right]-\mathbb E\!\left[
 b_H(W_{3_H})
\mid H,\vec O_{-W}
\right]
\right\rVert = o_p(n^{-1/2})$;   
   \item $\left\lVert
\mathbb{\hat E}\!\left[
 b_H(W_{2_H})
\mid H,\vec O_{-W}
\right]-\mathbb E\!\left[
 b_H(W_{2_H})
\mid H,\vec O_{-W}
\right]
\right\rVert$ 
  $\left\lVert
\mathbb{\hat E}\!\left[
 b_H(W_{3_H})
\mid H,\vec O_{-W}
\right]-\mathbb E\!\left[
 b_H(W_{3_H})
\mid H,\vec O_{-W}
\right]
\right\rVert = o_p(n^{-1/2})$;,
\end{enumerate}
where for \(k \in \{1,2,3\}\),
\(b_H(W_{k_H})\) is defined in Appendix~\ref{app:weight2}.
\end{thm}

We provide a proof in Appendix~\ref{app:efficiency2}.

\begin{example}\label{ex:latent_conf_robust}
By Theorem~\ref{thm:multiple_robustness2}, if $U$ is binary, the estimator
\(
\hat{\psi}
=
\mathbb{P}_n\!\left[
\varphi(W,Z,Y,A;\hat\eta)
\right]
\)
in Example~\ref{ex:latent_conf_if} is consistent provided \(p(A \mid U)\) or \(\mathbb{E}[Y \mid A,U]\) is consistently estimated, and at least one of the following nuisance pairs is consistently estimated:
\begin{enumerate}
\item
\(\{\widehat{\mathbb E}[W\mid  A,U],\,\widehat{\mathbb E}[Z\mid  A,U]\}\);
\item
\(\{\widehat{\mathbb E}[W\mid  A,U],\,\widehat{\mathbb E}[Y\mid  A,U]\}\);
\item
\(\{\widehat{\mathbb E}[Z\mid  A,U],\,\widehat{\mathbb E}[Y\mid  A,U]\}\).
\end{enumerate}

Furthermore, by Theorem~\ref{thm:efficiency2},
\(
\sqrt{n}(\hat\psi-\psi)
\rightsquigarrow
N\!\left(
0,\,
\mathbb E\!\left[
\{\varphi(W,Z,Y,A;\eta)-\psi\}^2
\right]
\right),
\)
under suitable regularity conditions including sample splitting and consistency of nuisance functions, provided the following conditions hold:
\begin{enumerate}
\item
\(
\left\|
\widehat{p}[A\mid U]
-
p[A\mid U]
\right\|
\left\|
\widehat{\mathbb{E}}[Y\mid A,U]
-
\mathbb{E}[Y\mid A,U]
\right\|
=
o_p(n^{-1/2});
\)

\item
\(
\left\|
\widehat{\mathbb E}[W\mid  A,U]
-
\mathbb E[W\mid  A,U]
\right\|
\left\|
\widehat{\mathbb E}[Z\mid  A,U]
-
\mathbb E[Z\mid  A,U]
\right\|
=
o_p(n^{-1/2});
\)

\item
\(
\left\|
\widehat{\mathbb E}[W\mid  A,U]
-
\mathbb E[W\mid  A,U]
\right\|
\left\|
\widehat{\mathbb E}[Y\mid  A,U]
-
\mathbb E[Y\mid  A,U]
\right\|
=
o_p(n^{-1/2});
\)

\item
\(
\left\|
\widehat{\mathbb E}[Z\mid  A,U]
-
\mathbb E[Z\mid A,U]
\right\|
\left\|
\widehat{\mathbb E}[Y\mid  A,U]
-
\mathbb E[Y\mid  A,U]
\right\|
=
o_p(n^{-1/2}).
\)
\end{enumerate}
\end{example}

\begin{example}\label{ex:latent_conf_med_robust}
By Theorem~\ref{thm:multiple_robustness2}, if $U$ is binary, the estimator
\(
\hat{\psi}
=
\mathbb{P}_n\!\left[
\varphi(W,Z,Y,A,M;\hat\eta)
\right]
\)
in Example~\ref{ex:latent_conf_med_if} is consistent provided \(p(M \mid A,U)\) or \(\{p(A \mid U),\mathbb{E}[Y \mid M,A,U]\}\) is consistently estimated, and at least one of the following nuisance pairs is consistently estimated:
\begin{enumerate}
\item
\(\{\widehat{\mathbb E}[W\mid  M,A,U],\,\widehat{\mathbb E}[Z\mid  M,A,U]\}\);
\item
\(\{\widehat{\mathbb E}[W\mid  M,A,U],\,\widehat{\mathbb E}[Y\mid  M,A,U]\}\);
\item
\(\{\widehat{\mathbb E}[Z\mid  M,A,U],\,\widehat{\mathbb E}[Y\mid  M,A,U]\}\).
\end{enumerate}

Furthermore, by Theorem~\ref{thm:efficiency2},
\(
\sqrt{n}(\hat\psi-\psi)
\rightsquigarrow
N\!\left(
0,\,
\mathbb E\!\left[
\{\varphi(W,Z,Y,A,M;\eta)-\psi\}^2
\right]
\right),
\)
under suitable regularity conditions including sample splitting and consistency of nuisance functions, provided the following conditions hold:
\begin{enumerate}
\item
\(
\left\|
\widehat{p}[A\mid U]
-
p[A\mid U]
\right\|
\left\|
\widehat{p}[M\mid A,U]
-
p[M\mid A,U]
\right\|
=
o_p(n^{-1/2});
\)

\item
\(
\left\|
\widehat{\mathbb{E}}[Y\mid M,A,U]
-
\mathbb{E}[Y\mid M,A,U]
\right\|
\left\|
\widehat{p}[M\mid A,U]
- p[M\mid A,U]
\right\|
=
o_p(n^{-1/2});
\)

\item
\(
\left\|
\widehat{\mathbb E}[W\mid  M,A,U]
-
\mathbb E[W\mid  M,A,U]
\right\|
\left\|
\widehat{\mathbb E}[Z\mid  M,A,U]
-
\mathbb E[Z\mid  M,A,U]
\right\|
=
o_p(n^{-1/2});
\)

\item
\(
\left\|
\widehat{\mathbb E}[W\mid  M,A,U]
-
\mathbb E[W\mid  M,A,U]
\right\|
\left\|
\widehat{\mathbb E}[Y\mid  M,A,U]
-
\mathbb E[Y\mid  M,A,U]
\right\|
=
o_p(n^{-1/2});
\)

\item
\(
\left\|
\widehat{\mathbb E}[Z\mid  M,A,U]
-
\mathbb E[Z\mid M,A,U]
\right\|
\left\|
\widehat{\mathbb E}[Y\mid  M,A,U]
-
\mathbb E[Y\mid  M,A,U]
\right\|
=
o_p(n^{-1/2}).
\)
\end{enumerate}
\end{example}

\section{Conclusion}
In this work, we study identification and estimation in latent variable models where variable(s) of interest are unobserved but key observed variables associated with the latent variable(s) are available. We establish conditions under which the full-data law is identifiable from the observed-data law. Moreover, we develop mean-zero estimating equations for target functionals of the full-data law. We additionally adapt our approach to give influence functions for target functions of the full-data law, which can be multiply robust and $\sqrt{n}$-consistent despite slower than parametric convergence of nuisance functions. The estimation theory we develop in this paper may extend to other settings in which the full-data law is identified but does not admit a simple tractable representation in terms of the observed-data law, such as missing-not-at-random and complex selection bias settings, pointing toward a promising direction for principled estimation in areas of growing research interest.

\bibliographystyle{plainnat}
\bibliography{references}

\newpage
\appendix

\section{Appendix A}\label{app:id_proof1}

We aim to show that for each value \(\vec o_{-W}\), the law \[
p(\vec H, \vec o_{-W}, \bigcup_{H \in \vec H} \{W_{1_H},W_{2_H},W_{3_H}\})
\]
is identified up to arbitrary permutation of the support values of \(\vec H\).

\begin{proof}
By the proof of Theorem 1 in \citet{hu2008instrumental}, under Assumptions~\ref{assump:bounded_density}, \ref{assump:mutual_independence}, \ref{assump:completeness}, and~\ref{assump:distinctness}, for each \(H \in \vec H\) and value \(\vec o_{-W}\), we have that \(p(W_{1_H},W_{2_H},W_{3_H},H\mid \vec o_{-W})\) is identified up to arbitrary permutation of the support values of \(H\). In other words, for each \(H \in \vec H\) and value \(\vec o_{-W}\),
\[
p\left(W_{1_H},W_{2_H},W_{3_H} \mid H, \vec o_{-W}\right)
\quad \text{and} \quad
p\left(H \mid \vec o_{-W}
\right)
\]
are jointly identified up to arbitrary permutation of the support values of \(H\).

By Assumptions~\ref{assump:m_sep1} and \ref{assump:m_sep2}, it follows that for each value \(\vec o_{-W}\),
\[
p\left(
\bigcup_{H \in \vec H} \{W_{1_H},W_{2_H},W_{3_H}\}
\,\middle|\,
\vec H,
\vec o_{-W}
\right)
\quad \text{and} \quad
p(\vec H \mid \vec o_{-W})
\]
 are jointly identified up to arbitrary permutation of the support values of \(
\vec H\).

For each value \(\vec o_{-W}\),
\begin{align*}
p(\vec H = \vec h_i, \vec o_{-W}, \bigcup_{H \in \vec H} \{W_{1_H},W_{2_H},W_{3_H}\}) = & p\left(
\bigcup_{H \in \vec H} \{W_{1_H},W_{2_H},W_{3_H}\}
\,\middle|\,
\vec H = \vec h_i, 
\vec \vec o_{-W}
\right) \\
&\quad \quad
p\left(
\vec H = \vec h_i \mid \vec o_{-W}
\right) \\
&\quad \quad 
p\left(
\vec o_{-W}
\right).
\end{align*}

Then since the variables in $\vec O \setminus
\bigcup_{H' \in \vec H} \{W_{1_{H'}}, W_{2_{H'}}, W_{3_{H'}}\}$ are observed, for each \(\vec o_{-W}\),
\[
p(\vec H = \vec h_i, \vec o_{-W}, \bigcup_{H \in \vec H} \{W_{1_H},W_{2_H},W_{3_H}\})
\]
 is identified up to arbitrary permutation of the support values of \(\vec H\).
\end{proof}

\section{Appendix B}\label{app:weight1}

We first prove that for each value \(\vec h_i\), there exist weight function
\(\omega_{\vec h_i}(\bigcup_{H \in \vec H} W_{1_H},\vec O_{-W};\eta)\) satisfying Equation~\ref{eq:weight}.

\begin{proof}
By Assumption~\ref{assump:completeness}, for each \(H \in \vec H\), and \(\vec o_{-W}\), the conditional laws in
\(\{
p(W_{1_H} \mid H=h_i,\vec o_{-W}):
h_i \text{ (in the support of } H \text{ when } \vec O_{-W} = \vec o_{-W})\}
\)
are linearly independent. For each \(H \in \vec H\) and \(\vec o_{-W}\), let
\(
K_H(\vec o_{-W})
=
\left|\operatorname{supp}(H \mid \vec O_{-W}=\vec o_{-W})\right|
\) and define the dominating measure
\[
\nu_{H,\vec o_{-W}}
=
\sum_{i=1}^{K_H(\vec o_{-W})}
p(W_{1_H} \mid H=h_i,\vec o_{-W}).
\]
Let
\[
b_{H,i}(W_{1_H},\vec o_{-W})
=
\frac{
dp(W_{1_H} \mid H=h_i, \vec o_{-W})
}{
d\nu_{H,\vec o_{-W}}
}(W_{1_H}).
\]

Write
\[
b_H(W_{1_H},\vec o_{-W})
=
\begin{pmatrix}
b_{H,1}(W_{1_H},\vec o_{-W})\\
\vdots\\
b_{H,K_H(\vec o_{-W})}(W_{1_H},\vec o_{-W})
\end{pmatrix},
\]
and define
\[
\mathbf{M}_H(\vec o_{-W})
=
\int
b_H(w,\vec o_{-W})b_H(w,\vec o_{-W})^\top
\,d\nu_{H,\vec o_{-W}}(w).
\]

Because the conditional laws in
\(\{
p(W_{1_H} \mid H=h_i, \vec o_{-W}):
h_i\}
\)
are linearly independent, the vectors in \(\{
b_{H,i}(W_{1_H},\vec o_{-W}):i\}
\) are also linearly independent. Hence, \(\mathbf{M}_H(\vec o_{-W})\) is invertible.

For each
\(h_i\),
define
\[
\omega_{i}(W_{1_H},\vec o_{-W}; \eta)
=
e_i^\top
\mathbf{M}_H(\vec o_{-W})^{-1}
b_H(W_{1_H},\vec o_{-W}),
\]
where \(e_i\) is the \(i\)th standard basis vector. Then,
\begin{align*}
&\mathbb E\!\left[
\omega_{i}(W_{1_H},\vec O_{-W} = \vec o_{-W}; \eta)
\,\middle|\,
H=h_j,\vec O_{-W}=\vec o_{-W}
\right] \\
&\qquad=
e_i^\top
\mathbf{M}_H(\vec o_{-W})^{-1}
\int
b_H(w,\vec o_{-W})b_{H,j}(w,\vec o_{-W})
\,d\nu_{H,\vec o_{-W}}(w)\\
&\qquad=
e_i^\top
\mathbf{M}_H(\vec o_{-W})^{-1}
\mathbf{M}_H(\vec o_{-W})
e_j\\
&\qquad=
e_i^\top e_j\\
&\qquad=
\mathbb{I}\{i=j\}.
\end{align*}

Hence
\[
\omega_{\vec h_i}\!\left(\bigcup_{H \in \vec H} W_{1_H},\vec O_{-W};\eta\right)
=
\prod_{H\in\vec H}
\omega_{i}\!\left(W_{1_H},\vec O_{-W};\eta\right)
\] satisfies Equation~\ref{eq:weight} 
by Assumption~\ref{assump:m_sep1}.
\end{proof}

Next, we prove that 
\[
\varphi(\vec O;\eta)
=
\sum_{\vec h_i \in \operatorname{Supp}(\vec H \mid \vec O_{-W})}
\omega_{\vec h_i}(\bigcup_{H \in \vec H} W_{1_H},\vec O_{-W};\eta)\,
\phi\!\left(
\vec H=\vec h_i,\,
 \vec O \setminus 
\bigcup_{H \in \vec H}
W_{1_H};
\eta
\right)
\]
satisfies Equation~\ref{eq:bridge_1}.

\begin{proof}
By Assumptions~\ref{assump:mutual_independence}-\ref{assump:m_sep1}, 
\begin{align*}
&\mathbb E\!\left[
\omega_{\vec h_i}(\bigcup_{H \in \vec H} W_{1_H},\vec O_{-W}; \eta)
\,\middle|\,
\vec H=\vec h_j,\vec O_{-W}
\right] \\
&=
\mathbb E\!\left[
\omega_{\vec h_i}(\bigcup_{H \in \vec H} W_{1_H},\vec O_{-W} ; \eta)
\,\middle|\,
\vec H =\vec h_j, \vec O \setminus \bigcup_{H \in \vec H}{W_{1_H}}
\right]\\
&=
\mathbb{I}\{i=j\}.
\end{align*}

Then 
\(\mathbb{E}[\varphi(\vec O; \eta) \mid \vec H =\vec h_j, \vec O \setminus \bigcup_{H \in \vec H}{W_{1_H}}] = \phi(\vec H = h_j, \vec O \setminus \bigcup_{H \in \vec H}{W_{1_H}}; \eta).
\)
Finally, Equation~\ref{eq:bridge_1} holds by iterated expectations.
\end{proof}

\section{Appendix C}\label{app:weight2}

While Appendix~\ref{app:weight1} establishes the existence of weight function
\(\omega_{\vec h_i}\!\left(\bigcup_{H \in \vec H} W_{1_H},\vec O_{-W};\eta\right)\) satisfying Equation~\ref{eq:weight} for each \(\vec h_i\), it is often preferable to construct \(\omega_{\vec h_i}\!\left(\bigcup_{H \in \vec H} W_{1_H},\vec O_{-W};\eta\right)\) with a simpler collection of nuisance functions via the following procedure.

For each \(H\in\vec H\) and \(\vec o_{-W}\), let
\(
K_H(\vec o_{-W})
=
\left|
\operatorname{supp}(H\mid\vec O_{-W}=\vec o_{-W})
\right|,
\)
and choose basis functions
\[
 b_{H,1}(W_{1_H}),\ldots,
 b_{H,K_H(\vec o_{-W})}(W_{1_H})\] such that
\[
{\mathbf M}_H(\vec o_{-W})
=
\begin{pmatrix}
\mathbb E\!\left[
 b_H(W_{1_H})
\mid H=h_1,\vec O_{-W}=\vec o_{-W}
\right]
&
\cdots
&
\mathbb E\!\left[
 b_H(W_{1_H})
\mid H=h_{K_H(\vec o_{-W})},\vec O_{-W}=\vec o_{-W}
\right]
\end{pmatrix}
\] is invertible. In practice, the functions
\( b_{H,1},\ldots, b_{H,K_H(\vec o_{-W})}\)
may be chosen from a convenient function class, such as indicator functions, polynomials, splines, or Fourier features.

Write
\[
 b_H(W_{1_H})
=
\begin{pmatrix}
 b_{H,1}(W_{1_H})\\
\vdots\\
 b_{H,K_H(\vec o_{-W})}(W_{1_H})
\end{pmatrix},
\]
and for each
\(h_i\), define
\[
\omega_{h_i}(W_{1_H},\vec o_{-W};\eta)
=
e_i^\top
{\mathbf M}_H(\vec o_{-W})^{-1}
b_H(W_{1_H}).
\]
Then, 
\begin{align*}
&\mathbb E\!\left[
\omega_{h_i}(W_{1_H},\vec O_{-W}=\vec o_{-W};\eta)
\,\middle|\,
H=h_j,\vec O_{-W}=\vec o_{-W}
\right]
\\
&\qquad=
e_i^\top
{\mathbf M}_H(\vec o_{-W})^{-1}
\mathbb E\!\left[
 b_H(W_{1_H})
\mid H=h_j,\vec O_{-W}=\vec o_{-W}
\right]
\\
&\qquad=
e_i^\top
{\mathbf M}_H(\vec o_{-W})^{-1}
{\mathbf M}_H(\vec o_{-W})e_j
\\
&\qquad=
e_i^\top e_j
=
\mathbb{I}\{i=j\}.
\end{align*}
Hence letting \(\vec h_i\) denote the
\(i^{\text{th}}\) joint support value of \(\vec H\), and \(h_i(H)\)
denote the value of \(H\) in \(\vec h_i\),
\[
\omega_{\vec h_i}\!\left(
\bigcup_{H\in\vec H} W_{1_H},
\vec O_{-W};\eta
\right)
=
\prod_{H\in\vec H}
\omega_{h_i(H)}\!\left(
W_{1_H},\vec O_{-W};\eta
\right)
\]
satisfies Equation~\ref{eq:weight} 
by Assumption~\ref{assump:m_sep1}.

\section{Appendix D}\label{app:polynomial_basis1}

\begin{proof}
For each \(H\in\vec H\) and value \(\vec o_{-W}\), let
\(
K_H(\vec o_{-W})
=
\left|
\operatorname{supp}(H\mid\vec O_{-W}=\vec o_{-W})
\right|.
\)
Under the location-shift model, 
\[
W_{1_H}\mid H=h_i,\vec O_{-W}=\vec o_{-W}
=
\mu_{H,i}(\vec o_{-W})+\epsilon_H(\vec o_{-W}),
\]
where \(\epsilon_H(\vec o_{-W})\) does not depend on the value of \(H\). Thus,
\[
\mathbb E[W_{1_H}^{\,k}\mid H=h_i,\vec O_{-W}=\vec o_{-W}]
=
\mathbb E\!\left[
\{\mu_{H,i}(\vec o_{-W})+\epsilon_H(\vec o_{-W})\}^k
\right].
\]

Choose
\[
 b_H(W_{1_H},\vec o_{-W})
=
\bigl(
1,
W_{1_H},
\ldots,
W_{1_H}^{K_H(\vec o_{-W})-1}
\bigr)^\top.
\]
Then the matrix
\[
{\mathbf M}_H(\vec o_{-W})
=
\begin{pmatrix}
\mathbb E[ b_H(W_{1_H},\vec o_{-W})\mid H=h_1,\vec O_{-W}=\vec o_{-W}]
&
\cdots
&
\mathbb E[ b_H(W_{1_H},\vec o_{-W})\mid H=h_{K_H(\vec o_{-W})},\vec O_{-W}=\vec o_{-W}]
\end{pmatrix}
\]
has \((r,i)\)-entry
\[
\bigl[{\mathbf M}_H(\vec o_{-W})\bigr]_{ri}
=
\mathbb E\!\left[
W_{1_H}^{\,r-1}
\mid
H=h_i,\vec O_{-W}=\vec o_{-W}
\right].
\]

By the binomial theorem,
\({\mathbf M}_H(\vec o_{-W})\)
is an invertible linear transformation of the Vandermonde matrix
\[
\mathbf V_H(\vec o_{-W})
=
\begin{pmatrix}
1 & 1 & \cdots & 1\\
\mu_{H,1}(\vec o_{-W}) &
\mu_{H,2}(\vec o_{-W}) &
\cdots &
\mu_{H,K_H(\vec o_{-W})}(\vec o_{-W})\\
\vdots & \vdots & & \vdots\\
\mu_{H,1}(\vec o_{-W})^{K_H(\vec o_{-W})-1} &
\mu_{H,2}(\vec o_{-W})^{K_H(\vec o_{-W})-1} &
\cdots &
\mu_{H,K_H(\vec o_{-W})}(\vec o_{-W})^{K_H(\vec o_{-W})-1}
\end{pmatrix}.
\]
Since
\[
\mu_{H,1}(\vec o_{-W}),\ldots,
\mu_{H,K_H(\vec o_{-W})}(\vec o_{-W})
\]
are distinct, \(\mathbf V_H(\vec o_{-W})\) is invertible. Hence,
\({\mathbf M}_H(\vec o_{-W})\) is invertible.

Therefore, 
\[
\omega_i(W_{1_H},\vec O_{-W};\eta)
=
e_i^\top
{\mathbf M}_H(\vec O_{-W})^{-1}
 b_H(W_{1_H},\vec O_{-W})
\]
satisfies
\[
\mathbb E\!\left[
\omega_i(W_{1_H},\vec O_{-W};\eta)
\mid
H=h_j,\vec O_{-W}
\right]
=
\mathbb{I}\{i=j\}.
\]
Hence letting \(\vec h_i\) denote the
\(i^{\text{th}}\) joint support value of \(\vec H\), and \(h_i(H)\)
denote the value of \(H\) in \(\vec h_i\), 
\[
\omega_{\vec h_i}\!\left(
\bigcup_{H\in\vec H} W_{1_H},
\vec O_{-W};\eta
\right)
=
\prod_{H\in\vec H}
\omega_{h_i(H)}\!\left(
W_{1_H},\vec O_{-W};\eta
\right)
\]
satisfies Equation~\ref{eq:weight} 
by Assumption~\ref{assump:m_sep1}.
\end{proof}

\section{Appendix E}\label{app:if1}

\begin{proof}
Let $\{P_\epsilon(\vec O) : \epsilon \in \mathbb{R}\}$ be a regular parametric submodel of the observed-data law with score
\[
s(\vec O) = \left.\frac{\partial}{\partial \epsilon} \log p_\epsilon(\vec O)\right|_{\epsilon=0}.
\]

We aim to show $\varphi(\vec O; \psi, \eta) = \varphi(\vec O; \eta) - \psi$ from Lemma~\ref{lemma:if_est1} is an observed-data influence function for target \(\psi\), i.e., that for every regular parametric submodel $\{P_\epsilon(\vec O) : \epsilon \in \mathbb{R}\}$,
\begin{equation}\label{eq:obs_if_proof}
\left.\frac{d}{d\epsilon}\psi(P_\epsilon(\vec O))\right|_{\epsilon=0}
= \mathbb{E}\big[\varphi(\vec O; \psi, \eta)\, s(\vec O)\big],
\qquad 
\mathbb{E}\big[\varphi(\vec O; \psi, \eta)\big]=0.
\end{equation}
The regular parametric submodel \(\{P_\epsilon(\vec{H}, \vec O) : \epsilon \in \mathbb{R}\}\) of the full-data law that induces $\{P_\epsilon(\vec O) : \epsilon \in \mathbb{R}\}$ through marginalization is indexed by the same parameter \(\epsilon\), with score
\[
s(\vec{H}, \vec O)
=
\left.
\frac{\partial}{\partial \epsilon}
\log p_\epsilon(\vec{H}, \vec O)
\right|_{\epsilon=0},
\]
such that
\[
s(\vec O)
=
\mathbb{E}[s(\vec{H}, \vec O)\mid \vec O].
\]

Then  since \(\phi(\cdot; \psi, \eta) = \phi(\cdot; \eta) - \psi\) is an influence function in the full-data model, by Assumptions~\ref{assump:mutual_independence}-\ref{assump:m_sep1}, we have that
\begin{align*}
&\left.\frac{d}{d\epsilon}\psi(P_\epsilon(\vec O))\right|_{\epsilon=0}= \left.\frac{d}{d\epsilon}\psi(P_\epsilon(\vec{H}, \vec O))\right|_{\epsilon=0}= \mathbb{E}\big[\phi(
\cdot; \psi, \eta)\, s(\vec{H}, \vec O)\big] \\ 
& \quad =
\mathbb{E}\big[\phi(
\cdot; \psi, \eta)\, s(\bigcup_{H \in \vec H} W_{1_H} \mid \vec H, \vec O_{-W})\big] + 
\mathbb{E}\big[\phi(
\cdot; \psi, \eta)\, s(\bigcup_{H \in \vec H} W_{2_H} \mid \vec H, \vec O_{-W})\big]  \\
& \quad + 
\mathbb{E}\big[\phi(
\cdot; \psi, \eta)\, s(\bigcup_{H \in \vec H} W_{3_H} \mid \vec H, \vec O_{-W})\big] + 
\mathbb{E}\big[\phi(
\cdot; \psi, \eta)\, s(\vec H, \vec O_{-W})\big]\\ 
& \quad =
\mathbb{E}\big[\mathbb{E}\big[\phi(
\cdot; \psi, \eta) \mid \bigcup_{H \in \vec H} W_{1_H}, \vec H, \vec O_{-W}) \big]\, s(\bigcup_{H \in \vec H} W_{1_H} \mid \vec H, \vec O_{-W}) \big] \\
& \quad + 
\mathbb{E}\big[\mathbb{E}\big[\phi(
\cdot; \psi, \eta) \mid \bigcup_{H \in \vec H} W_{2_H}, \vec H, \vec O_{-W} \big]\, s(\bigcup_{H \in \vec H} W_{2_H} \mid \vec H, \vec O_{-W})\big]  \\
& \quad + 
\mathbb{E}\big[\mathbb{E}\big[\phi(
\cdot; \psi, \eta) \mid \bigcup_{H \in \vec H} W_{3_H}, \vec H, \vec O_{-W} \big]\, s(\bigcup_{H \in \vec H} W_{3_H} \mid \vec H, \vec O_{-W} )\big] \\
& \quad + 
\mathbb{E}\big[\mathbb{E}\big[\phi(
\cdot; \psi, \eta) \mid \vec H, \vec O_{-W}\big]\, s(\vec H, \vec O_{-W})\big].
\end{align*}

Since the following conditional expectations hold by Equation~\ref{eq:bridge_2},
\begin{align}
\mathbb{E}[\varphi(\vec O; \psi, \eta) \mid \bigcup_{H \in \vec H} W_{1_H}, \vec{H}, \vec O_{-W}] &= \mathbb{E}[\phi(\cdot; \psi, \eta)\mid \bigcup_{H \in \vec H} W_{1_H}, \vec{H}, \vec O_{-W}] = \phi(\cdot; \psi, \eta) - \psi, \\
\mathbb{E}[\varphi(\vec O; \psi, \eta) \mid \bigcup_{H \in \vec H} W_{2_H}, \vec{H}, \vec O_{-W}] &= \mathbb{E}[\phi(\cdot; \psi, \eta)\mid \bigcup_{H \in \vec H} W_{2_H}, \vec{H}, \vec O_{-W}] = \phi(\cdot; \psi, \eta) - \psi, \\
\mathbb{E}[\varphi(\vec O; \psi, \eta) \mid \bigcup_{H \in \vec H} W_{3_H}, \vec{H}, \vec O_{-W}] &= \mathbb{E}[\phi(\cdot; \psi, \eta)\mid \bigcup_{H \in \vec H} W_{3_H}, \vec{H}, \vec O_{-W}] = \phi(\cdot; \psi, \eta) - \psi, \\
\mathbb{E}[\varphi(\vec O; \psi, \eta) \mid \vec{H}, \vec O_{-W}] &= \mathbb{E}[\phi(\cdot; \psi, \eta)\mid \vec{H}, \vec O_{-W}] = \phi(\cdot; \psi, \eta) - \psi,
\end{align}
we can substitute these into the earlier expression to obtain
\begin{align*}
\left.\frac{d}{d\epsilon}\psi(P_\epsilon)\right|_{\epsilon=0} 
&= \mathbb{E}\big[ \varphi(\vec O; \psi, \eta)\, s(\vec{H}, \vec O ) \big] .
\end{align*}

Since \(\varphi(O)\) is a function of the observed data,
\begin{align*}
\left.\frac{d}{d\epsilon}\psi(P_\epsilon)\right|_{\epsilon=0}
&= \mathbb{E}\big[\varphi(\vec O; \psi, \eta)\, s(\vec{H}, \vec O)\big] \\
&= \mathbb{E}\big[ \varphi(\vec O; \psi, \eta)\, \mathbb{E}[s(\vec{H}, \vec O) \mid \vec O] \big] \\
&= \mathbb{E}\big[ \varphi(\vec O; \psi, \eta)\, s(\vec O) \big],
\end{align*}
satisfying the first part of Equation~\ref{eq:obs_if_proof}.

Next, fix \(l \in \{1,2,3\}\) and note by Equation~\ref{eq:bridge_2}, we have that
\begin{align*}
\mathbb{E}[\varphi(\vec O; \psi, \eta)] 
&= \mathbb{E}\Bigg[
\mathbb{E}[\varphi(\vec O; \psi, \eta) \mid \vec{H}, \vec O_{-W}
\Bigg] \\
&= \mathbb{E}[\phi(\cdot; \psi, \eta)] = 0,
\end{align*}
satisfying the second part of Equation~\ref{eq:obs_if_proof}.
\end{proof}

\section{Appendix F}\label{app:multiple_robustness1}

\begin{proof}
By the conditional independences in Assumptions~\ref{assump:mutual_independence}-\ref{assump:m_sep1} and at least one of the nuisance functions sets \(\eta'_1,\eta'_2,\eta'_3 \in \eta\) being correctly specified, 

\[
\mathbb{E}\!\left[
\lambda_{\vec h_i}\!\left(
\bigcup_{H \in \vec H,k}W_{k_H},
\vec O_{-W};
\hat\eta
\right)
\,\middle|\,
\vec H=\vec h_j,\vec O_{-W}
\right] = \mathbb{I}[i=j]
\]

Then by iterated expectations,
\(
\mathbb{P}_n\!\left[\hat{\psi}_{\mathrm{obs}}\right]
-
\mathbb{P}_n\!\left[\phi(\cdot, \hat \eta)\right]
\overset{p}{\longrightarrow}
0
\) under standard regularity conditions.
Since at least one of the original nuisance function sets
\(\eta_1,\ldots,\eta_m \in \eta\) is correctly specified,
\(
\mathbb{P}_n\!\left[\phi(\cdot, \hat \eta)\right]
\overset{p}{\longrightarrow}
\psi.
\)
Therefore,
\(
\mathbb{P}_n\!\left[\hat{\psi}_{\mathrm{obs}}\right]
\overset{p}{\longrightarrow}
\psi.
\)

\end{proof}

\section{Appendix G}\label{app:efficiency1}

\begin{proof}
Let $ \varphi = \varphi(O; \psi, \eta) = \varphi(O; \eta)- \psi$. Expand \[
\hat{\psi} - \psi 
= \mathbb{P}_n \hat{\varphi} - \mathbb{P}\varphi
= (\mathbb{P}_n - \mathbb{P})(\hat{\varphi} - \varphi) + \mathbb{P}(\hat{\varphi} - \varphi) + (\mathbb{P}_n - \mathbb{P})\varphi
= R_1 + R_2 + (\mathbb{P}_n - \mathbb{P})\varphi,
\]
where 
\(
R_1 = (\mathbb{P}_n - \mathbb{P})(\hat{\varphi} - \varphi), 
\quad 
R_2 = \mathbb{P}(\hat{\varphi} - \varphi).
\)

$R_1 = o_p(n^{-1/2})$ under consistency of nuisance functions and sample splitting by Lemma 2 in \citet{kennedy2020sharp}.

Below we show $R_2 = o_p(n^{-1/2})$ under conditions (i)-(iii), therefore
\[
\hat{\psi} - \psi
= (\mathbb{P}_n - \mathbb{P})\varphi + o_p(n^{-1/2})
= \mathbb{P}_n\varphi + o_p(n^{-1/2})
\]
which is equivalent to \(
\sqrt{n}\,(\hat{\psi}-\psi) \;\;\overset{d}{\longrightarrow}\;\; 
\mathcal{N}\!\big(0, \,\mathbb{E}[\varphi^2]\big).
\)

By Assumption~\ref{assump:m_sep1},
\begin{align*}
R_2 
&= \mathbb{P}[\hat{\varphi} - \varphi] = \mathbb{E}[\hat{\varphi} - \varphi] = \mathbb{E}\!\left[
   \varphi(O; \hat \eta) - \psi
\right] =\mathbb{E}\!\left[\mathbb{E}[
   \varphi(O; \hat \eta) - \psi \mid \vec H, \vec O_{-W}]
\right] \\[10pt]
& = \mathbb{E}\left[\phi(\cdot; \hat \eta)\mathbb{E}\!\left[
\lambda_{\vec h_i}\!\left(
\bigcup_{H \in \vec H,k}W_{k_H},
\vec O_{-W};
\hat \eta
\right)
\,\middle|\,
\vec H,\vec O_{-W}
\right] - \psi \right].
\\[10pt]
\end{align*}

From Appendix~\ref{app:weight2} and Assumptions~\ref{assump:mutual_independence}-\ref{assump:m_sep1}, we have that

\begin{align*}
&\mathbb{E}\!\left[
\lambda_{\vec h_i}\!\left(
\bigcup_{H \in \vec H,k}W_{k_H},
\vec O_{-W};
\hat \eta
\right)
\,\middle|\,
\vec H = \vec h_j, \vec O_{-W}
\right] \\
&=
\mathbb{E}\!\left[
\omega_{1}\omega_{2}; \hat \eta
\,\middle|\,
\vec H= \vec h_j, \vec O_{-W}
\right] + \mathbb{E}\!\left[
\omega_{1}\omega_{3}; \hat \eta
\,\middle|\,
\vec H= \vec h_j, \vec O_{-W}
\right] + \mathbb{E}\!\left[
\omega_{2}\omega_{3}; \hat \eta
\,\middle|\,
\vec H= \vec h_j, \vec O_{-W}
\right] \\
&- 2 \mathbb{E}\!\left[
\omega_{1}\omega_{2}\omega_{3}
; \hat \eta \,\middle|\,
\vec H= \vec h_j, \vec O_{-W}
\right] \\[0.75em]
&=
\mathbb{E}\!\left[
\omega_{1}; \hat \eta
\,\middle|\,
\vec H= \vec h_j, \vec O_{-W}
\right]\mathbb{E}\!\left[
\omega_{2}; \hat \eta
\,\middle|\,
\vec H= \vec h_j, \vec O_{-W}
\right] + \mathbb{E}\!\left[
\omega_{1}; \hat \eta
\,\middle|\,
\vec H= \vec h_j, \vec O_{-W}
\right]\mathbb{E}\!\left[
\omega_{3}; \hat \eta
\,\middle|\,
\vec H= \vec h_j, \vec O_{-W}
\right] + \\
& \quad
\mathbb{E}\!\left[
\omega_{2}; \hat \eta
\,\middle|\,
\vec H= \vec h_j, \vec O_{-W}
\right]\mathbb{E}\!\left[
\omega_{3}; \hat \eta
\,\middle|\,
\vec H= \vec h_j, \vec O_{-W}
\right] \\
&- 2 \mathbb{E}\!\left[
\omega_{1}; \hat \eta
\,\middle|\,
\vec H= \vec h_j, \vec O_{-W}
\right]\mathbb{E}\!\left[
\omega_{2}; \hat \eta
\,\middle|\,
\vec H= \vec h_j, \vec O_{-W}
\right]\mathbb{E}\!\left[
\omega_{3}; \hat \eta
\,\middle|\,
\vec H= \vec h_j, \vec O_{-W}
\right]. \\[0.75em]
\end{align*}

Consider the first term
\begin{align*}
&\mathbb{E}\!\left[
\omega_{1}; \hat \eta
\,\middle|\,
\vec H= \vec h_j, \vec O_{-W}
\right]\mathbb{E}\!\left[
\omega_{2}; \hat \eta
\,\middle|\,
\vec H= \vec h_j, \vec O_{-W}
\right] \\
&=\left\{\mathbb{E}\!\left[
\omega_{1}; \hat \eta
\,\middle|\,
\vec H= \vec h_j, \vec O_{-W}
\right] - \mathbb{\hat E}\!\left[
\omega_{1}; \hat \eta
\,\middle|\,
\vec H= \vec h_j, \vec O_{-W}
\right] + \mathbb{\hat E}\!\left[
\omega_{1}; \hat \eta
\,\middle|\,
\vec H= \vec h_j, \vec O_{-W}
\right] \right\}\cdot \\
& \left\{\mathbb{E}\!\left[
\omega_{2}; \hat \eta
\,\middle|\,
\vec H= \vec h_j, \vec O_{-W}
\right] - \mathbb{\hat E}\!\left[
\omega_{2}; \hat \eta
\,\middle|\,
\vec H= \vec h_j, \vec O_{-W}
\right] + \mathbb{\hat E}\!\left[
\omega_{2}; \hat \eta
\,\middle|\,
\vec H= \vec h_j, \vec O_{-W}
\right] \right\} \\
&=\left\{\mathbb{E}\!\left[
\omega_{1}; \hat \eta
\,\middle|\,
\vec H= \vec h_j, \vec O_{-W}
\right] - \mathbb{\hat E}\!\left[
\omega_{1}; \hat \eta
\,\middle|\,
\vec H= \vec h_j, \vec O_{-W}
\right] + \mathbb{I}(i=j) \right\}\cdot \\
& \left\{\mathbb{E}\!\left[
\omega_{2}; \hat \eta
\,\middle|\,
\vec H= \vec h_j, \vec O_{-W}
\right] - \mathbb{\hat E}\!\left[
\omega_{2}; \hat \eta
\,\middle|\,
\vec H= \vec h_j, \vec O_{-W}
\right] + \mathbb{I}(i=j) \right\}.
\end{align*}

A similar breakdown holds for the second term
\(
\mathbb{E}\!\left[
\omega_{1}; \hat \eta
\,\middle|\,
\vec H= \vec h_j, \vec O_{-W}
\right]\mathbb{E}\!\left[
\omega_{3}; \hat \eta
\,\middle|\,
\vec H= \vec h_j, \vec O_{-W}
\right]
\)
and the third term
\(
\mathbb{E}\!\left[
\omega_{2}; \hat \eta
\,\middle|\,
\vec H= \vec h_j, \vec O_{-W}
\right]\mathbb{E}\!\left[
\omega_{3}; \hat \eta
\,\middle|\,
\vec H= \vec h_j, \vec O_{-W}
\right].
\) For the fourth term, we have that
\begin{align*}
&-2\mathbb{E}\!\left[
\omega_{1}; \hat \eta
\,\middle|\,
\vec H= \vec h_j, \vec O_{-W}
\right]\mathbb{E}\!\left[
\omega_{2}; \hat \eta
\,\middle|\,
\vec H= \vec h_j, \vec O_{-W}
\right]\mathbb{E}\!\left[
\omega_{3}; \hat \eta
\,\middle|\,
\vec H= \vec h_j, \vec O_{-W}
\right] \\
& = -2 \left\{\mathbb{E}\!\left[
\omega_{1}; \hat \eta
\,\middle|\,
\vec H= \vec h_j, \vec O_{-W}
\right] - \mathbb{\hat E}\!\left[
\omega_{1}; \hat \eta
\,\middle|\,
\vec H= \vec h_j, \vec O_{-W}
\right] + \mathbb{I}(i=j) \right\} \\
& \left\{\mathbb{E}\!\left[
\omega_{2}; \hat \eta
\,\middle|\,
\vec H= \vec h_j, \vec O_{-W}
\right] - \mathbb{\hat E}\!\left[
\omega_{2}; \hat \eta
\,\middle|\,
\vec H= \vec h_j, \vec O_{-W}
\right] + \mathbb{I}(i=j) \right\} \\
& \left\{\mathbb{E}\!\left[
\omega_{3}; \hat \eta
\,\middle|\,
\vec H= \vec h_j, \vec O_{-W}
\right] - \mathbb{\hat E}\!\left[
\omega_{3}; \hat \eta
\,\middle|\,
\vec H= \vec h_j, \vec O_{-W}
\right] + \mathbb{I}(i=j) \right\}.
\end{align*}

Altogether, given conditions (i)-(iii),
\begin{align*}
&\mathbb{E}\!\left[
\lambda_{\vec h_i}\!\left(
\bigcup_{H \in \vec H,k}W_{k_H},
\vec O_{-W};
\hat \eta
\right)
\,\middle|\,
\vec H = \vec h_j, \vec O_{-W}
\right] = \mathbb{I}(i=j) + o_p(n^{-1/2}).
\end{align*}

Then,
\begin{align*}
R_2 
= \mathbb{E}[\phi(\cdot; \hat \eta) - \psi] + o_p(n^{-1/2}) = o_p(n^{-1/2}).\end{align*}

\end{proof}

\section{Appendix H}\label{app:if2}
\begin{proof}
Let $\{P_\epsilon(\vec O) : \epsilon \in \mathbb{R}\}$ be a regular parametric submodel of the observed-data law with score
\[
s(\vec O) = \left.\frac{\partial}{\partial \epsilon} \log p_\epsilon(\vec O)\right|_{\epsilon=0}.
\]

We aim to show $\varphi(\vec O; \psi, \eta) = \varphi(\vec O; \eta) - \psi$ from Lemma~\ref{lemma:if_est2} is an observed-data influence function for target \(\psi\), i.e., that for every regular parametric submodel $\{P_\epsilon(\vec O) : \epsilon \in \mathbb{R}\}$,
\begin{equation}\label{eq_obs_if_proof2}
\left.\frac{d}{d\epsilon}\psi(P_\epsilon(\vec O))\right|_{\epsilon=0}
= \mathbb{E}\big[\varphi(\vec O; \psi, \eta)\, s(\vec O)\big],
\qquad 
\mathbb{E}\big[\varphi(\vec O; \psi, \eta)\big]=0.
\end{equation}

The regular parametric submodel
\(\{P_\epsilon(\vec{H}, \vec O) : \epsilon \in \mathbb{R}\}\) of the full-data law that induces \(\{P_\epsilon(\vec O) : \epsilon \in \mathbb{R}\}\) through marginalization is indexed by the same parameter \(\epsilon\), with score
\[
s(\vec{H}, \vec O)
=
\left.
\frac{\partial}{\partial \epsilon}
\log p_\epsilon(\vec{H}, \vec O)
\right|_{\epsilon=0},
\]
such that
\[
s(\vec O)
=
\mathbb{E}[s(\vec{H}, \vec O)\mid \vec O].
\]

Then since \(\phi(\cdot; \psi, \eta) = \phi_1\left(\vec{H}, \big(
\vec{O}\setminus
\bigcup_{H \in \vec H}\{W_{1_H},W_{2_H}\}\big); \eta\right) + \phi_2\left(\vec{H},
\vec{O}_{-W}; \eta\right) - \psi\) is an influence function in the full-data model, by Assumptions~\ref{assump:mutual_independence}-\ref{assump:m_sep1}, we have that
\begin{align*}
&\left.\frac{d}{d\epsilon}\psi(P_\epsilon(\vec O))\right|_{\epsilon=0}= \left.\frac{d}{d\epsilon}\psi(P_\epsilon(\vec{H}, \vec O))\right|_{\epsilon=0}= \mathbb{E}\big[\phi(
\cdot; \psi, \eta)\, s(\vec{H}, \vec O)\big] \\ 
& \quad =
\mathbb{E}\big[\phi(
\cdot; \psi, \eta)\, s(\bigcup_{H \in \vec H} W_{1_H} \mid \vec H, \vec O_{-W})\big] + 
\mathbb{E}\big[\phi(
\cdot; \psi, \eta)\, s(\bigcup_{H \in \vec H} W_{2_H} \mid \vec H, \vec O_{-W})\big]  \\
& \quad + 
\mathbb{E}\big[\phi(
\cdot; \psi, \eta)\, s(\vec H, \big(\vec O \setminus \bigcup_{H \in \vec H} \{W_{1_H},W_{2_H}\}\big))\big]\\ 
& \quad =
\mathbb{E}\big[\mathbb{E}\big[\phi(
\cdot; \psi, \eta) \mid \bigcup_{H \in \vec H} W_{1_H}, \vec H, \vec O_{-W} \big]\, s(\bigcup_{H \in \vec H} W_{1_H} \mid \vec H, \vec O_{-W} )\big] \\
& \quad + 
\mathbb{E}\big[\mathbb{E}\big[\phi(
\cdot; \psi, \eta) \mid \bigcup_{H \in \vec H} W_{2_H}, \vec H, \vec O_{-W} \big]\, s(\bigcup_{H \in \vec H} W_{2_H} \mid \vec H, \vec O_{-W})\big]  \\
& \quad + 
\mathbb{E}\big[\mathbb{E}\big[\phi(
\cdot; \psi, \eta) \mid \vec H, \big(\vec O \setminus \bigcup_{H \in \vec H} \{W_{1_H},W_{2_H}\}\big)\big]\, s(\vec H, \big(\vec O \setminus \bigcup_{H \in \vec H} \{W_{1_H},W_{2_H}\}\big))\big]
\end{align*}

Since the following conditional expectations hold by Assumptions~\ref{assump:mutual_independence}-\ref{assump:m_sep1} and Equations~\ref{eq:bridge_3a}-\ref{eq:bridge_3b},
\begin{align*}
\mathbb{E}[\varphi(\vec O; \psi, \eta) \mid \bigcup_{H \in \vec H} W_{1_H}, \vec{H}, \vec O_{-W}] &= \mathbb{E}[\phi(\cdot; \psi, \eta)\mid \bigcup_{H \in \vec H} W_{1_H}, \vec{H}, \vec O_{-W}], \\
& = \phi_2\left(\vec{H},
\vec{O}_{-W}; \eta\right) - \psi\\
\mathbb{E}[\varphi(\vec O; \psi, \eta) \mid \bigcup_{H \in \vec H} W_{2_H}, \vec{H}, \vec O_{-W}] &= \mathbb{E}[\phi(\cdot; \psi, \eta)\mid \bigcup_{H \in \vec H} W_{2_H}, \vec{H}, \vec O_{-W}], \\
& = \phi_2\left(\vec{H},
\vec{O}_{-W}; \eta\right) - \psi\\
\mathbb{E}[\varphi(\vec O; \psi, \eta) \mid \vec{H}, \big(\vec O \setminus \bigcup_{H \in \vec H} \{W_{1_H},W_{2_H}\}\big)] &= \mathbb{E}[\phi(\cdot; \psi, \eta)\mid \vec{H}, \big(\vec O \setminus \bigcup_{H \in \vec H} \{W_{1_H},W_{2_H}\}\big)]
\\
& = \phi_1\left(\vec{H}, \big(
\vec{O}\setminus
\bigcup_{H \in \vec H}\{W_{1_H},W_{2_H}\}\big); \eta\right) + \phi_2\left(\vec{H},
\vec{O}_{-W}; \eta\right) - \psi,
\end{align*}
we can substitute these into the earlier expression to obtain
\begin{align*}
\left.\frac{d}{d\epsilon}\psi(P_\epsilon)\right|_{\epsilon=0} 
&= \mathbb{E}\big[ \varphi(\vec O; \psi, \eta)\, s(\vec{H}, \vec O ) \big] .
\end{align*}

Since \(\varphi(O)\) is a function of the observed data,
\begin{align*}
\left.\frac{d}{d\epsilon}\psi(P_\epsilon)\right|_{\epsilon=0}
&= \mathbb{E}\big[\varphi(\vec O; \psi, \eta)\, s(\vec{H}, \vec O)\big] \\
&= \mathbb{E}\big[ \varphi(\vec O; \psi, \eta)\, \mathbb{E}[s(\vec{H}, \vec O) \mid \vec O] \big] \\
&= \mathbb{E}\big[ \varphi(\vec O; \psi, \eta)\, s(\vec O) \big],
\end{align*}
satisfying the first part of Equation~\ref{eq_obs_if_proof2}.

Next, fix \(l \in \{1,2,3\}\) and note by Equation~\ref{eq:bridge_3a}, we have that
\begin{align*}
\mathbb{E}[\varphi(\vec O; \psi, \eta)] 
&= \mathbb{E}\Bigg[
\mathbb{E}[\varphi(\vec O; \psi, \eta) \mid \vec{H}, \big(\vec O \setminus \bigcup_{H \in \vec H} \{W_{1_H},W_{2_H}\}\big)]
\Bigg] \\
&= \mathbb{E}[\phi(\cdot; \psi, \eta)] = 0,
\end{align*}
satisfying the second part of Equation~\ref{eq_obs_if_proof2}.
\end{proof}

\section{Appendix I}\label{app:multiple_robustness2}
\begin{proof}
By the conditional independences in Assumptions~\ref{assump:mutual_independence}-\ref{assump:m_sep1} and at least one of the nuisance functions sets \(\eta'_1,\eta'_2,\eta'_3 \in \eta\) being correctly specified, 

\[
\mathbb{E}\!\left[
\lambda_{\vec h_i}\!\left(
\bigcup_{H \in \vec H,k}W_{k_H},
\vec O_{-W};
\hat\eta
\right)
\,\middle|\,
\vec H=\vec h_j,\vec O_{-W}
\right] = \mathbb{I}[i=j].
\]

Suppose \(\eta'_1\) is correctly specified. Then by Assumptions~\ref{assump:mutual_independence}-\ref{assump:m_sep1},
\begin{align*}
&\mathbb{E}\!\left[
\gamma_{\vec h_i}\!\left(
\bigcup_{H \in \vec H}\{W_{1_H},W_{2_H}\},
\vec O_{-W};
\hat\eta
\right)
\,\middle|\,
\vec H=\vec h_j,\big(
\vec{O}\setminus
\bigcup_{H \in \vec H}\{W_{1_H},W_{2_H}\}\big)
\right] = \\
&
\mathbb{E}\!\left[
\gamma_{\vec h_i}\!\left(
\bigcup_{H \in \vec H,k}W_{k_H},
\vec O_{-W};
\hat\eta
\right)
\,\middle|\,
\vec H=\vec h_j,\vec O_{-W}
\right] = \mathbb{I}[i=j].
\end{align*}

Then by iterated expectations,
\(
\mathbb{P}_n\!\left[\hat{\psi}_{\mathrm{obs}}\right] =  \mathbb{P}_n\!\left[\phi_1(\cdot; \hat \eta)\right] + \mathbb{P}_n\!\left[\phi_2(\cdot; \hat \eta)\right] \).
Since at least one of the original nuisance function sets
\(\eta_1,\ldots,\eta_m \in \eta\) is correctly specified,
\(
\mathbb{P}_n\!\left[\phi_1(\cdot; \hat \eta)\right] + \mathbb{P}_n\!\left[\phi_2(\cdot; \hat \eta)\right]
\overset{p}{\longrightarrow}
\psi.
\)
Therefore,
\(
\mathbb{P}_n\!\left[\hat{\psi}_{\mathrm{obs}}\right]
\overset{p}{\longrightarrow}
\psi.
\)

Suppose \(\eta'_2\) or \(\eta'_3\) is correctly specified. Then by Assumptions~\ref{assump:mutual_independence}-\ref{assump:m_sep1},
\begin{align*}
&\mathbb{E}\!\left[
\gamma_{\vec h_i}\!\left(
\bigcup_{H \in \vec H}\{W_{1_H},W_{2_H}\},
\vec O_{-W};
\hat\eta
\right)
\,\middle|\,
\vec H=\vec h_j,\big(
\vec{O}\setminus
\bigcup_{H \in \vec H}\{W_{1_H},W_{2_H}\}\big)
\right] = \\
&
\mathbb{E}\!\left[
\gamma_{\vec h_i}\!\left(
\bigcup_{H \in \vec H,k}W_{k_H},
\vec O_{-W};
\hat\eta
\right)
\,\middle|\,
\vec H=\vec h_j,\vec O_{-W}
\right] = c \cdot \mathbb{I}[i=j], \text{ where } c \text{ is a constant.}
\end{align*}

Then by iterated expectations,
\(
\mathbb{P}_n\!\left[\hat{\psi}_{\mathrm{obs}}\right] = c \cdot \mathbb{P}_n\!\left[\phi_1(\cdot; \hat \eta)\right] + \mathbb{P}_n\!\left[\phi_2(\cdot; \hat \eta)\right] \).
Since \(\{\mathbb{E}[b_H(W_{3_H})\mid \vec H, \vec O_{-W}]:H\}\) is correctly specified,
\(
\mathbb{P}_n\!\left[\phi_1(\cdot; \hat \eta)\right] 
\overset{p}{\longrightarrow}
0.
\)
Therefore,
\(
\mathbb{P}_n\!\left[\hat{\psi}_{\mathrm{obs}}\right]
\overset{p}{\longrightarrow}
\mathbb{P}_n\!\left[\phi_2(\cdot; \hat \eta)\right].
\) Since at least one of the original nuisance function sets
\(\eta_1,\ldots,\eta_m \in \eta\) is correctly specified, \(
\mathbb{P}_n\!\left[\phi_2(\cdot; \hat \eta)\right] 
\overset{p}{\longrightarrow}
\psi.
\) Then, \(
\mathbb{P}_n\!\left[\hat{\psi}_{\mathrm{obs}}\right]
\overset{p}{\longrightarrow}
\psi.
\)

\end{proof}

\section{Appendix J}\label{app:efficiency2}
\begin{proof}
Let $ \varphi = \varphi(O; \psi, \eta) = \varphi(O; \eta)- \psi$. Expand \[
\hat{\psi} - \psi 
= \mathbb{P}_n \hat{\varphi} - \mathbb{P}\varphi
= (\mathbb{P}_n - \mathbb{P})(\hat{\varphi} - \varphi) + \mathbb{P}(\hat{\varphi} - \varphi) + (\mathbb{P}_n - \mathbb{P})\varphi
= R_1 + R_2 + (\mathbb{P}_n - \mathbb{P})\varphi,
\]
where 
\(
R_1 = (\mathbb{P}_n - \mathbb{P})(\hat{\varphi} - \varphi), 
\quad 
R_2 = \mathbb{P}(\hat{\varphi} - \varphi).
\)

$R_1 = o_p(n^{-1/2})$ under consistency of nuisance functions and sample splitting by Lemma 2 in \citet{kennedy2020sharp}.

Below we show $R_2 = o_p(n^{-1/2})$ under conditions (i)-(iii), therefore
\[
\hat{\psi} - \psi
= (\mathbb{P}_n - \mathbb{P})\varphi + o_p(n^{-1/2})
= \mathbb{P}_n\varphi + o_p(n^{-1/2})
\]
which is equivalent to \(
\sqrt{n}\,(\hat{\psi}-\psi) \;\;\overset{d}{\longrightarrow}\;\; 
\mathcal{N}\!\big(0, \,\mathbb{E}[\varphi^2]\big).
\)

By Assumptions~\ref{assump:mutual_independence}-\ref{assump:m_sep1},
\begin{align*}
R_2 
&= \mathbb{P}[\hat{\varphi} - \varphi] 
= \mathbb{E}[\hat{\varphi} - \varphi] 
= \mathbb{E}\!\left[
   \varphi(O; \hat \eta) - \psi
\right] \\[10pt]
&= 
\mathbb{E}\Bigg[
g\left(\vec H, \vec O_{-W}\right)
\Bigg[
t_{\vec O_{-W}}\!\left(
\left\{
\mathbb{E}\!\left[
b_{H,q}(W_{3_H})
\mid H,\vec O_{-W}
\right]
:
H\in\vec H,\;
q=1,\ldots,K_H(\vec O_{-W})
\right\}
\right) \\
&\qquad\qquad\qquad
-
t_{\vec O_{-W}}\!\left(
\left\{
\mathbb{\hat E}\!\left[
b_{H,q}(W_{3_H})
\mid H,\vec O_{-W}
\right]
:
H\in\vec H,\;
q=1,\ldots,K_H(\vec O_{-W})
\right\}
\right)
\Bigg] \\
&\qquad\qquad\cdot
\mathbb{E}\!\left[
\gamma_{\vec h_i}\!\left(
\bigcup_{H \in \vec H}\{W_{1_H},W_{2_H}\},
\vec O_{-W};
\hat \eta
\right)
\,\middle|\,
\vec H, \vec O_{-W}
\right]
\Bigg] \\
&\quad+
\mathbb{E}\left[
\phi_2(\cdot; \hat \eta)
\mathbb{E}\!\left[
\lambda_{\vec h_i}\!\left(
\bigcup_{H \in \vec H,k}W_{k_H},
\vec O_{-W};
\hat \eta
\right)
\,\middle|\,
\vec H, \vec O_{-W}
\right]
\right]
-\psi.
\end{align*}

From Appendix~\ref{app:weight2} and Assumptions~\ref{assump:mutual_independence}-\ref{assump:m_sep1}, we have that

\begin{align*}
&\mathbb{E}\!\left[
\gamma_{\vec h_i}\!\left(
\bigcup_{H \in \vec H}\{W_{1_H},W_{2_H}\},
\vec O_{-W};
\hat \eta
\right)
\,\middle|\,
\vec H = \vec h_j, \vec O_{-W}
\right] \\
&=
\mathbb{E}\!\left[
\omega_{1}\omega_{2}; \hat \eta
\,\middle|\,
\vec H= \vec h_j, \vec O_{-W}
\right] \\
&=\left\{\mathbb{E}\!\left[
\omega_{1}; \hat \eta
\,\middle|\,
\vec H= \vec h_j, \vec O_{-W}
\right] - \mathbb{\hat E}\!\left[
\omega_{1}; \hat \eta
\,\middle|\,
\vec H= \vec h_j, \vec O_{-W}
\right] + \mathbb{\hat E}\!\left[
\omega_{1}; \hat \eta
\,\middle|\,
\vec H= \vec h_j, \vec O_{-W}
\right] \right\}\cdot \\
& \left\{\mathbb{E}\!\left[
\omega_{2}; \hat \eta
\,\middle|\,
\vec H= \vec h_j, \vec O_{-W}
\right] - \mathbb{\hat E}\!\left[
\omega_{2}; \hat \eta
\,\middle|\,
\vec H= \vec h_j, \vec O_{-W}
\right] + \mathbb{\hat E}\!\left[
\omega_{2}; \hat \eta
\,\middle|\,
\vec H= \vec h_j, \vec O_{-W}
\right] \right\} \\
& =\left\{\mathbb{E}\!\left[
\omega_{1}; \hat \eta
\,\middle|\,
\vec H= \vec h_j, \vec O_{-W}
\right] - \mathbb{\hat E}\!\left[
\omega_{1}; \hat \eta
\,\middle|\,
\vec H= \vec h_j, \vec O_{-W}
\right] + \mathbb{I}(i=j) \right\} \cdot\\
& \left\{\mathbb{E}\!\left[
\omega_{2}; \hat \eta
\,\middle|\,
\vec H= \vec h_j, \vec O_{-W}
\right] - \mathbb{\hat E}\!\left[
\omega_{2}; \hat \eta
\,\middle|\,
\vec H= \vec h_j, \vec O_{-W}
\right] + \mathbb{I}(i=j) \right\}.
\end{align*}

Then by conditions (i)-(iii) and Assumptions~\ref{assump:mutual_independence}-\ref{assump:m_sep1},

\begin{align*}
&\mathbb{E}\Bigg[
g\left(\vec H, \vec O_{-W}\right)
\Bigg[
t_{\vec O_{-W}}\!\left(
\left\{
\mathbb{E}\!\left[
b_{H,q}(W_{3_H})
\mid H,\vec O_{-W}
\right]
:
H\in\vec H,\;
q=1,\ldots,K_H(\vec O_{-W})
\right\}
\right) \\
&\qquad\qquad\qquad
-
t_{\vec O_{-W}}\!\left(
\left\{
\mathbb{\hat E}\!\left[
b_{H,q}(W_{3_H})
\mid H,\vec O_{-W}
\right]
:
H\in\vec H,\;
q=1,\ldots,K_H(\vec O_{-W})
\right\}
\right)
\Bigg] \\
&\qquad\qquad\cdot
\mathbb{E}\!\left[
\gamma_{\vec h_i}\!\left(
\bigcup_{H \in \vec H}\{W_{1_H},W_{2_H}\},
\vec O_{-W};
\hat \eta
\right)
\,\middle|\,
\vec H, \vec O_{-W}
\right]
\Bigg]  \\
&=
\mathbb{E}\Bigg[
g\left(\vec H, \vec O_{-W}\right)
\Bigg[
t_{\vec O_{-W}}\!\left(
\left\{
\mathbb{E}\!\left[
b_{H,q}(W_{3_H})
\mid H,\vec O_{-W}
\right]
:
H\in\vec H,\;
q=1,\ldots,K_H(\vec O_{-W})
\right\}
\right) \\
&\qquad\qquad\qquad
-
t_{\vec O_{-W}}\!\left(
\left\{
\mathbb{\hat E}\!\left[
b_{H,q}(W_{3_H})
\mid H,\vec O_{-W}
\right]
:
H\in\vec H,\;
q=1,\ldots,K_H(\vec O_{-W})
\right\}
\right)
\Bigg]
\Bigg]
+o_p(n^{-1/2}) \\
&=
\mathbb{E}[\phi_1(\cdot;\hat\eta)]
+o_p(n^{-1/2}).
\end{align*}

Similar results hold for the second term
\(
\mathbb{E}\!\left[
\omega_{1}; \hat \eta
\,\middle|\,
\vec H= \vec h_j, \vec O_{-W}
\right]\mathbb{E}\!\left[
\omega_{3}; \hat \eta
\,\middle|\,
\vec H= \vec h_j, \vec O_{-W}
\right]
\)
and the third term
\(
\mathbb{E}\!\left[
\omega_{2}; \hat \eta
\,\middle|\,
\vec H= \vec h_j, \vec O_{-W}
\right]\mathbb{E}\!\left[
\omega_{3}; \hat \eta
\,\middle|\,
\vec H= \vec h_j, \vec O_{-W}
\right].
\) For the fourth term, we have that
\begin{align*}
&-2\mathbb{E}\!\left[
\omega_{1}; \hat \eta
\,\middle|\,
\vec H= \vec h_j, \vec O_{-W}
\right]\mathbb{E}\!\left[
\omega_{2}; \hat \eta
\,\middle|\,
\vec H= \vec h_j, \vec O_{-W}
\right]\mathbb{E}\!\left[
\omega_{3}; \hat \eta
\,\middle|\,
\vec H= \vec h_j, \vec O_{-W}
\right] \\
& = -2 \left\{\mathbb{E}\!\left[
\omega_{1}; \hat \eta
\,\middle|\,
\vec H= \vec h_j, \vec O_{-W}
\right] - \mathbb{\hat E}\!\left[
\omega_{1}; \hat \eta
\,\middle|\,
\vec H= \vec h_j, \vec O_{-W}
\right] + \mathbb{I}(i=j) \right\} \\
& \left\{\mathbb{E}\!\left[
\omega_{2}; \hat \eta
\,\middle|\,
\vec H= \vec h_j, \vec O_{-W}
\right] - \mathbb{\hat E}\!\left[
\omega_{2}; \hat \eta
\,\middle|\,
\vec H= \vec h_j, \vec O_{-W}
\right] + \mathbb{I}(i=j) \right\} \\
& \left\{\mathbb{E}\!\left[
\omega_{3}; \hat \eta
\,\middle|\,
\vec H= \vec h_j, \vec O_{-W}
\right] - \mathbb{\hat E}\!\left[
\omega_{3}; \hat \eta
\,\middle|\,
\vec H= \vec h_j, \vec O_{-W}
\right] + \mathbb{I}(i=j) \right\}.
\end{align*}

Altogether, given conditions (i)-(iii),
\begin{align*}
&\mathbb{E}\!\left[
\lambda_{\vec h_i}\!\left(
\bigcup_{H \in \vec H,k}W_{k_H},
\vec O_{-W};
\hat \eta
\right)
\,\middle|\,
\vec H = \vec h_j, \vec O_{-W}
\right] = \mathbb{I}(i=j) + o_p(n^{-1/2}).
\end{align*}

Then,
\[\mathbb{E}\left[\phi_2(\cdot; \hat \eta)\mathbb{E}\!\left[
\lambda_{\vec h_i}\!\left(
\bigcup_{H \in \vec H,k}W_{k_H},
\vec O_{-W};
\hat \eta
\right)
\,\middle|\,
\vec H, \vec O_{-W}
\right]\right] = \mathbb{E}[\phi_2(\cdot; \hat \eta)]  + o_p(n^{-1/2}).\]

Finally,
\begin{align*}
R_2 
= \mathbb{E}[\phi_1(\cdot; \hat \eta) + \phi_2(\cdot; \hat \eta) - \psi] + o_p(n^{-1/2}) = o_p(n^{-1/2}).\end{align*}

\end{proof}

\end{document}